\documentclass[
a4paper
]{article}

\usepackage[utf8]{inputenc}

\usepackage{xspace}
\usepackage{algorithm}
\usepackage[noend]{algpseudocode}
\usepackage{amsmath}
\usepackage{amssymb}
\usepackage{amsthm}

\usepackage{MnSymbol}
\usepackage{mathtools}
\usepackage{textcomp}
\usepackage{fullpage}

\usepackage{varwidth}
\usepackage{graphicx}

\theoremstyle{plain}
\newtheorem{theorem}{Theorem}[section]
\newtheorem*{theorem*}{Theorem}
\newtheorem{proposition}[theorem]{Proposition}
\newtheorem*{proposition*}{Proposition}
\newtheorem{lemma}[theorem]{Lemma}
\newtheorem*{lemma*}{Lemma}
\newtheorem{corollary}[theorem]{Corollary}
\newtheorem*{corollary*}{Corollary}

\newtheorem*{claim*}{Claim}

\theoremstyle{definition}

\newtheorem*{definition*}{Definition}

\theoremstyle{remark}

\newtheorem*{observation*}{Observation}
\newtheorem{example}[theorem]{Example}
\newtheorem*{example*}{Example}

\makeatletter
\@fleqnfalse
\@mathmargin\@centering
\makeatother

\usepackage{tikz}
\usetikzlibrary{automata, graphs,positioning,chains,arrows,snakes,decorations.pathmorphing}

\tikzset{
	defaultstyle/.style={
		>=stealth, 
		semithick, 
		auto,
		initial text= {},
		initial distance= {3mm},
		accepting distance= {3mm}}}

\DeclareMathOperator{\eword}{\varepsilon}

\newcommand{\firstExampleEx}{E_1}
\newcommand{\firstExampleTrans}{T_1}
\newcommand{\secondExampleEx}{E_2}
\newcommand{\secondExampleTrans}{T_2}

\newcommand\Vtextvisiblespace[1][.3em]{%
	\mbox{\kern.06em\vrule height.3ex}%
	\vbox{\hrule width#1}%
	\hbox{\vrule height.3ex}}

\DeclareMathOperator{\spans}{\textsf{Spans}}
\newcommand{\varset}{\mbox{\ensuremath{\mathcal{X}}}}
\newcommand{\varsetSmallFont}{\ensuremath{\mathcal{X}}}
\DeclareMathOperator{\openmark}{\vdash}
\DeclareMathOperator{\closemark}{\dashv}
\newcommand{\doc}{\ensuremath{\mathbf{D}}}
\newcommand{\spann}[2]{\ensuremath{[#1,#2\rangle}}
\newcommand{\open}[1]{\prescript{}{#1}{\openmark}}
\newcommand{\close}[1]{\closemark_{#1}}

\newcommand{\getDoc}[1]{\doc_{#1}}
\newcommand{\getSpanTuple}[1]{\textsf{t}_{#1}}

\newcommand{\ta}{\ensuremath{\mathtt{a}}}
\newcommand{\tb}{\ensuremath{\mathtt{b}}}

\newcommand{\varsx}{\ensuremath{\textsf{x}}}
\newcommand{\varsy}{\ensuremath{\textsf{y}}}
\newcommand{\varsz}{\ensuremath{\textsf{z}}}

\newcommand{\lang}[1]{\ensuremath{\mathcal{L}(#1)}}

\DeclareMathOperator{\altop}{+}

\newcommand{\nrefw}[1]{\ensuremath{\mathtt{enc}(#1)}}

\newcommand{\dom}{\operatorname{dom}}
\newcommand{\pmap}{\rightharpoonup}

\newcommand{\mult}[2]{\operatorname{mult}_{#1}(#2)}

\newcommand{\Reg}{\mathcal{R}}
\newcommand{\Asg}{\operatorname{Asg}}
\newcommand{\Asgs}{\operatorname{Asg}^*}
\newcommand{\Val}{\operatorname{Val}}

\newcommand{\ndetbf}[1]{\textsf{ndet}(#1)}

\newcommand{\nsst}{\mathcal{T}}
\newcommand{\sst}{\mathcal{T}}
\newcommand{\out}{\operatorname{out}}

\DeclareMathOperator{\oout}{%
	\ooalign{\scalebox{1.0}[1.0]{$o$}\cr\hidewidth\rotatebox[origin=c]{-30}{\raisebox{.4ex}{\scalebox{0.65}{$\  \boldsymbol{\smile}$}}}\hidewidth}}%
	
\newcommand{\Runs}{\operatorname{Runs}}
\newcommand{\reg}{\operatorname{reg}}

\newcommand{\indx}{\operatorname{index}}
\newcommand{\indxsigma}{\sigma_{\indx}}
\newcommand{\indxnu}{\nu_{\indx}}
\newcommand{\indxw}{w_{\indx}}

\newcommand{\D}{\mathcal{D}}
\newcommand{\n}{\textsf{n}}
\newcommand{\m}{\textsf{m}}
\newcommand{\s}{\textsf{s}}
\newcommand{\unode}{\textsf{u}}
\newcommand{\N}{\textsf{N}}
\newcommand{\enumECSA}{\textsc{EnumECSA}}
\newcommand{\odepth}{\operatorname{odepth}}
\newcommand{\dadd}{\operatorname{add}}
\newcommand{\dunion}{\operatorname{union}}
\newcommand{\dextend}{\operatorname{extend}}

\newcommand{\hasnext}{\operatorname{hasnext}}
\newcommand{\bnodelabel}{\m}

\newcommand{\uit}{\iota}
\newcommand{\uid}{\operatorname{id}}

\newcommand{\ustack}{\operatorname{St}}

\algdef{SE}[DOWHILE]{do}{doWhile}{\algorithmicdo}[1]{\algorithmicwhile\ #1}%
\algnewcommand{\LineComment}[1]{\State \(\triangleright\) #1}

\newcommand{\Told}{T_{\operatorname{old}}}
\newcommand{\Tnew}{T_{\operatorname{new}}}

\newcommand{\cO}{\mathcal{O}}

\newcommand{\bbN}{\mathbb{N}}

\newcommand{\sem}[1]{{\lsem{}{#1}\rsem}}

\newcommand{\multiset}[1]{\{\!\!\{#1\}\!\!\}}

\DeclareMathOperator{\nondef}{\bot}

\renewcommand{\operatorname}[1]{\textsf{#1}}
\DeclareMathOperator{\bigo}{O}

\algrenewcommand\algorithmicindent{1em}%
\algnewcommand\algorithmicforeach{\textbf{for each}}
\algdef{S}[FOR]{ForEach}[1]{\algorithmicforeach\ #1\ \algorithmicdo}

\newcommand*{\set}[1]{\ensuremath{\{ #1 \}}}

\usepackage{authblk}

\title{A framework for extraction and transformation of documents}

\author[1]{Cristian Riveros}
\author[2]{Markus L.\ Schmid}
\author[2]{Nicole Schweikardt}
\affil[1]{Pontificia Universidad Católica de Chile, \texttt{cristian.riveros@uc.cl}}
\affil[2]{Humboldt-Universität zu Berlin, Germany, \texttt{MLSchmid@MLSchmid.de}, \texttt{schweikn@informatik.hu-berlin.de}}

\date{}
 
\begin{document}

\maketitle

\begin{abstract}
We present a theoretical framework for the extraction and transformation of text documents. We propose to use a two-phase process where the first phase extracts span-tuples from a document, and the second phase maps the content of the span-tuples into new documents. We base the extraction phase on the framework of document spanners and the transformation phase on the theory of polyregular functions, the class of regular string-to-string functions with polynomial growth.

For supporting practical extract-transform scenarios, we propose an extension of document spanners described by regex formulas from span-tuples to so-called multispan-tuples, where variables are mapped to sets of spans. We prove that this extension, called regex multispanners, has the same desirable properties as standard spanners described by regex formulas. 
In our framework, an Extract-Transform (ET) program is given by a regex multispanner followed by a polyregular~function. 

In this paper, we study the expressibility and evaluation problem of ET programs when the transformation function is linear, called linear ET programs. We show that linear ET programs are equally expressive as non-deterministic streaming string transducers under bag semantics.
Moreover, we show that linear ET programs are closed under composition.
Finally, we present an enumeration algorithm for evaluating every linear ET program over a document with linear time preprocessing and constant delay.  
\end{abstract}

\section{Introduction}\label{sec:introduction}
	
Information extraction (IE) of text documents found its theoretical foundations in the framework of document spanners~\cite{fagin2015document}. Introduced by Fagin, Kimelfeld, Reiss, and Vansummeren one decade ago~\cite{FaginKRV13}, this framework sets the grounds for rule-based IE through the notion of spanners and how to combine them by using relational operators. Moreover, it inspired a lot of research on languages~\cite{Freydenberger2019,Peterfreund21}, expressiveness~\cite{freydenberger2018document,maturana2018document,PeterfreundCFK19,DoleschalEtAl2019,SchmidSchweikardt_ICDT21},
evaluation~\cite{florenzano2020efficient,amarilli2021constant,doleschal2019split,BourhisEtAl2021,FreydenbergerEtAl2018,PeterfreundEtAl2019,FreydenbergerThompson2022}, provenance~\cite{doleschal2022weight,DoleschalKM23}, and compressed evaluation~\cite{schmid2021spanner,SchmidSchweikardtPODS2022,MunozR23} for IE. Initially conceived to understand SystemT, IBM's rule-based IE system, it has found some recent promising implementations in~\cite{RiverosJV23}. See~\cite{SchmidSchweikardt_Gems,DBLP:journals/sigmod/AmarilliBMN20} for surveys of the area.

Although IE is a crucial task for data management, document extraction is usually followed by a transformation into new documents. Indeed, data transformation is essential for communicating middleware systems, where data does not conform to the required format of a third-party system and needs to be converted. This forms the main task of the so-called ETL technologies~\cite{vassiliadis2009survey} (for Extract-Transform-Load), which are crucial parts of today's data management workflows. Even the search-and-replace operation over text, a ubiquitous task in all text systems, can be conceived as an extraction (i.e., search) followed by a transformation (i.e., replace). Hence, although document spanners provided the formal ground of ruled-based IE, these grounds are incomplete without understanding the following transformation~process.

In this paper, we study the extraction and transformation of text documents through the lens of document spanners. We propose to use a two-phase process where the first phase extracts span-tuples from a document, and the second phase maps the content of the span-tuples into new documents. The main idea here is to use the framework of document spanners for the extraction phase, completing it with an transformation phase using the theory of polyregular functions. 

Let us now explain our proposed framework and respective results with several examples. We consider the following document $\doc$ over alphabet $\Sigma$ that lists famous English-speaking singers in the format $\#$ [person] $\#$ [person] $\#$ \dots $\#$ [person] $\#$ with [person] $=$ [last name]; [name]; [birthplace]; [opt]; [opt];\dots, where [opt] are optional attributes like, e.\,g., age, nickname, honorary titles etc, which follow no fixed scheme. This means that `$\#$' and `$;$' are used as separators, and for convenience, we set $\widehat{\Sigma} = \Sigma \setminus \{;, \#\}$. 

\medskip

\hspace{-0.2cm}
\scalebox{1}{
\begin{minipage}[t]{0.93\textwidth}
$\doc = $ \# \texttt{Holiday}; \texttt{Billie}; \texttt{USA} \# \texttt{Bush}; \texttt{Kate}; \texttt{England} \# \texttt{Young}; \texttt{Neil}; \texttt{Canada}; 78; \texttt{``Godfather of Grunge''}  \# \texttt{King}; \texttt{Carole}; \texttt{USA}; 81 \# \texttt{McCartney}; \texttt{Paul}; \texttt{England}; \texttt{Sir}; \texttt{CH}; \texttt{MBE} \# \texttt{Mitchell}; \texttt{Joni}; \texttt{Canada}; \texttt{painter}  \# \texttt{Franklin}; \texttt{Aretha}; \texttt{USA}; \texttt{``Queen of Soul''} \# \texttt{O’Riordan}; \texttt{Dolores};  \texttt{Ireland}; \texttt{\textdagger 01/15/2018}  \# \texttt{Bowie}; \texttt{David}; \texttt{England} \# \texttt{Dylan}; \texttt{Bob}; \texttt{USA}; 82 \# \texttt{Young}; \texttt{Neil}; \texttt{USA}  \# \texttt{Gallagher}; \texttt{Rory}; \texttt{Ireland} \#
\end{minipage}
}

\medskip

Our general approach is to first perform information extraction on $\doc$ with (an extension of) \emph{regex formulas} (see~\cite{fagin2015document}). For example, $\firstExampleEx \ := \ \ \Sigma^* \# \, \varsx\{\widehat{\Sigma}^*\}; \varsy\{\widehat{\Sigma}^*\}; \Sigma^*$
is a regex formula that uses a variable $\varsx$ to extract just any factor over $\widehat{\Sigma}$ that occurs between separators `$\#$' and `$;$', and a variable $\varsy$ to extract the following factor over $\widehat{\Sigma}$ between the next two occurrences of `$;$'. The construct $\varsx\{r\}$ creates a \emph{span} $\spann{i}{j}$ pointing to factor $\doc[i] \doc[i+1] \ldots \doc[j-1]$, satisfying the subexpression $r$. 
The regex formula $\firstExampleEx$ specifies a spanner that, on our example document $\doc$
produces the set $\{(\spann{2}{9}, \spann{10}{16}),$ $(\spann{21}{25},\spann{26}{30}),$ $(\spann{39}{44}, \spann{45}{49}), \ldots\}$
of \emph{span-tuples}. Every span-tuple $t$ represents an annotated version of our document: the factor corresponding to $t(\varsx)$ is annotated by (or allocated to) $\varsx$ and the factor corresponding to $t(\varsy)$ is annotated by $\varsy$. And each such annotation represents another finding of relevant information in our document.

As the second step of our model, we map each annotated version of our document to another document, which therefore describes a transformation of our document. This mapping will explicitly use the information represented by the annotation from the extraction phase. Formally, we use a string-to-string function for the transformation phase that gets as argument a single string that represents the document and the information given by a span-tuple.
The transformation function is applied to each of the annotated documents. Since the transformation function can be non-injective, we may obtain the same output string several times. In our model, we keep such duplicates (by using bag semantics), since they represent the situation that the same output string can be constructed from completely different annotations.
Considering our example, the annotated documents are as follows
(we use $\textbar$ as separator between strings for better readability):
\begin{align*}
&\scalebox{1}{$\# \open{\varsx}\texttt{Holiday}\close{\varsx}; \open{\varsy}\texttt{Billie}\close{\varsy}; \texttt{USA} \# \texttt{Bush}; \texttt{Kate};  \texttt{England} \# \texttt{Young}; \texttt{Neil}; \ldots \hspace{2pt}\textbar$}\\
&\scalebox{1}{$\# \texttt{Holiday}; \texttt{Billie}; \texttt{USA} \# \open{\varsx}\texttt{Bush}\close{\varsx}; \open{\varsy}\texttt{Kate}\close{\varsy}; \texttt{England} \# \texttt{Young}; \texttt{Neil}; \ldots \hspace{2pt}\textbar$}\\
&\scalebox{1}{$\# \texttt{Holiday}; \texttt{Billie}; \texttt{USA} \# \texttt{Bush}; \texttt{Kate}; \texttt{England} \# \open{\varsx} \texttt{Young} \close{\varsx}; \open{\varsy} \texttt{Neil} \close{\varsy}; \ldots \hspace{2pt}\textbar \hspace{2pt} \ldots$}
\end{align*}
They are obtained by simply representing the span for 
each variable $\varsz\in\set{\varsx,\varsy}$ by a pair of parentheses~$\open{\varsz} \ldots \close{\varsz}$.
If we couple the spanner $\firstExampleEx$ with the string-to-string function:\label{T1Function}
\[
\firstExampleTrans \ := \ \ u \open{\varsx} v_{\varsx} \close{\varsx}; \open{\varsy} v_{\varsy} \close{\varsy} u' \:\: \mapsto \:\: v_{\varsy} \Vtextvisiblespace \, v_{\varsx}\,
\]
where $u, u' \in \Sigma^*$ and $v_{\varsx}, v_{\varsy} \in \widehat{\Sigma}^*$,
then we obtain what we call an \emph{extract-transform program} (\emph{ET program}). This ET program $(\firstExampleEx,\firstExampleTrans)$ maps our document $\doc$ to the following collection of strings:
\begin{align*}
&\scalebox{1}{\texttt{Billie}\:\texttt{Holiday} \hspace{2pt}\textbar\hspace{2pt} \texttt{Kate}\:\texttt{Bush} \hspace{2pt}\textbar\hspace{2pt} \texttt{Neil}\:\texttt{Young} \hspace{2pt}\textbar\hspace{2pt} \texttt{Carole}\:\texttt{King} \hspace{2pt}\textbar\hspace{2pt} \texttt{Paul}\:\texttt{McCartney} \hspace{2pt}\textbar}\\
&\scalebox{1}{\texttt{Joni}\:\texttt{Mitchell} \hspace{2pt}\textbar\hspace{2pt} \texttt{Aretha}\:\texttt{Franklin} \hspace{2pt}\textbar\hspace{2pt} \texttt{Dolores}\:\texttt{O’Riordan} \hspace{2pt}\textbar\hspace{2pt} \texttt{David}\:\texttt{Bowie} \hspace{2pt}\textbar}\\
&\scalebox{1}{\texttt{Bob}\:\texttt{Dylan} \hspace{2pt}\textbar\hspace{2pt} \texttt{Neil}\:\texttt{Young}  \hspace{2pt}\textbar\hspace{2pt} \texttt{Rory}\:\texttt{Gallagher}}
\end{align*}
Note that these are the names of all singers from the document, but in the format ``[first name] [last name]'' (i.\,e., the order of the names is swapped). Moreover, ``\texttt{Neil Young}'' constitutes a duplicate, since it is the image under \emph{different} extractions of the regex formula $E_1$.
This is an intended feature of our model, since every distinct way of information extraction should translate into a distinct output document of the transformation phase.

This example already points out that classical \emph{one-way transducers}~\cite{berstel2013transductions} 
are not suitable transformation functions, since they cannot change the order of the extracted factors.
We propose, instead, to
use the class of \emph{polyregular functions}~\cite{Bojanczyk22,engelfriet2002two}, which is a robust class of regular string-to-string functions with
(at most)
polynomial growth. The function from above is a polyregular function that even has linear growth; a subclass of special interest for this work. We will discuss polyregular functions (and representations for them) later in this paper. In the examples here, we describe them sufficiently formally but on an intuitive level. 

Let us move on to a more complicated example. We want to define an ET program which maps $\doc$ to the following strings:
\begin{align*}
&\scalebox{1}{\texttt{Billie}\:\texttt{Holiday}: \texttt{USA} \# \hspace{2pt}\textbar \hspace{2pt} \texttt{Kate}\:\texttt{Bush}: \texttt{England} \# \hspace{2pt}\textbar\hspace{2pt} \texttt{Neil}\:\texttt{Young}: \texttt{Canada} \hspace{2pt}\#}\\
&\scalebox{1}{\texttt{Neil}\:\texttt{Young}: \texttt{78} \# \texttt{Neil}\:\texttt{Young}: \texttt{``Godfather of Grunge''} \# \hspace{2pt}\textbar\hspace{2pt} \texttt{Carole}\:\texttt{King}: \texttt{USA} \#}\\
& \scalebox{1}{\texttt{Carole}\:\texttt{King}: \texttt{81} \# \hspace{2pt}\textbar \hspace{2pt} \texttt{Paul}\:\texttt{McCartney}: \texttt{England} \# \texttt{Paul}\:\texttt{McCartney}: \texttt{Sir} \hspace{2pt}\#}\\
&\scalebox{1}{\texttt{Paul}\:\texttt{McCartney}: \texttt{CH} \# \texttt{Paul}\:\texttt{McCartney}: \texttt{MBE} \# \hspace{2pt}\textbar  \hspace{2pt} \ldots}
\end{align*}
Specifically, we want to 
map each singer entry [last name]; [name]; [birthplace]; [opt-1]; \dots; [opt-k] to the string [name] [last name]: [birthplace] \# [name] [last name]: [opt-1] \# \ldots \# [name] [last name]: [opt-k] \#,
which is something like the cross-product between [name] [last name] and the list [birthplace], [opt-1], \dots, [opt-k]. For example, the substring \# \texttt{Young}; \texttt{Neil}; \texttt{Canada}; 78; \texttt{``Godfather of Grunge''} \# is mapped to 
\texttt{Neil} \texttt{Young}: \texttt{Canada} \# \texttt{Neil} \texttt{Young}: \texttt{78} \# \texttt{Neil} \texttt{Young}: \texttt{``Godfather of Crunge''} \#.

Intuitively, we would have to extract each element of the unbounded list [birthplace], [opt-1], \dots, [opt-k] in its own variable, which goes beyond the capability of document spanners.
As a remedy,
in this paper we propose to extend the classical spanner framework to \emph{multispanners}, which can extract in each variable a \emph{set of spans} instead of only a single span. Now assume that we apply to the document the expression:
\[
	\secondExampleEx \ := \ \ \Sigma^* \# \, \varsx\{\widehat{\Sigma}^*\}; \varsy\{\widehat{\Sigma}^*\} ; \varsz\{ \widehat{\Sigma}^* \} (; \varsz\{ \widehat{\Sigma}^* \})^*\, \#\, \Sigma^*\,.
\]

Then for the substring
\texttt{\#}\texttt{Young}; \texttt{Neil}; \texttt{Canada}; 78; \texttt{``Godfather of Grunge''}\texttt{\#}, we will extract \emph{each} attribute as a single span of variable $\varsz$, so the corresponding \emph{multispan} tuple and the corresponding string representation would be: 
\begin{align*}
	&\scalebox{1}{$\big(\underbrace{\spann{39}{44}}_{\varsx},\underbrace{\spann{45}{49}}_{\varsy}, \underbrace{\{\spann{50}{56}, \spann{57}{59}, \spann{60}{81}\}}_{\varsz}\big) \text{ and}$}\\
	&\scalebox{1}{$\ldots \open{\varsx} \texttt{Young} \close{\varsx};\open{\varsy}\texttt{Neil}\close{\varsy};\open{\varsz}\texttt{Canada}\close{\varsz}; \open{\varsz}\texttt{78}\close{\varsz};\open{\varsz}\texttt{``Godfather of Grunge''}\close{\varsz}\ldots$}
\end{align*}
Thus, by defining a suitable transformation function $\secondExampleTrans$, we get an ET program $(\secondExampleEx, \secondExampleTrans)$ for our task. This is again possible by a polyregular function, but not with linear growth.

These examples show that typical string transduction tasks can often be naturally divided into the job of locating the relevant parts of our input, followed by their manipulation (e.\,g., reordering, deleting, copying, etc.), and our framework is particularly tailored to this situation. Moreover, we can easily construct new ET programs by combining specifications for an extract task (i.\,e., a multispanner) with specifications for a transform task (a string-to-string function) of existing ET programs. For example, the same spanner that finds relevant text blocks, like names, email addresses, URLs, etc. can be combined with different string-to-string functions that pick and arrange some of these text blocks in a certain format. Or, for cleaning textual data, we use a string-to-string function that deletes all the annotated factors, and we combine this with different spanners for finding the factors to be deleted.  

As also pointed out by our examples, it is a desired feature that the output of an ET program can contain duplicates. Let us motivate this by a more practical example. Assume that the extraction phase marks the addresses in a list of customer orders, while the transformation phase then transforms these addresses into a format suitable for printing on the parcels to be shipped to the customers. In the likely situation that there are different orders by the same person, the extraction phase produces different markings that will all be transformed into the same address label. Such duplicates, however, are important, since we actually want to obtain an address label for each distinct order (i.\,e., distinct marking), even if these addresses are the same (although the orders are not). This means that the output sets of our ET programs should actually be bags.

Our main conceptual contribution is a sound formalisation of our ET framework, which includes the extension of the classical spanner framework to multispanners (which we believe are interesting in their own right). Our main algorithmic result is about \emph{linear ET programs}, i.\,e., the subclass where the multispanner of the extraction can be described by a regular expression (similar to the ones used above) and the transformation function is a polyregular function with linear growth (e.g., the program $(\firstExampleEx, \firstExampleTrans)$). We show that we can enumerate the output documents of linear ET programs with linear preprocessing and output-linear delay (in data complexity). Furthermore, we show that ET programs are closed under composition.
This implies
that we can enumerate the output of a composition of programs efficiently. 
As an important tool for our enumeration algorithm, we show that linear ET programs have the same expressive power as nondeterministic streaming string transducers (NSST). This also gives an important insight with respect to the expressive power of linear ET programs.

Specifically, the contributions are as follows.
\begin{itemize}
	\item We introduce the formalism of multispanners, for extracting an unbounded number of spans in a single  tuple (Section~\ref{sec:multispanners}).
	\item By combining multispanners and polyregular functions, we introduce ET programs as a framework for extraction and transformation of documents (Section~\ref{sec:framework}).
	\item We restrict our analysis to the linear case, and show that linear ET programs are equally expressive as NSSTs under bag-semantics (Section~\ref{sec:properties}).
	\item We present an enumeration algorithm for evaluating linear ET programs with linear-time preprocessing and output-linear delay (Section~\ref{sec:enumeration}).
	\item
          We show
          that linear ET programs are closed under composition and, therefore, we can evaluate the composition of linear ET programs efficiently (Section~\ref{sec:composition}).
\end{itemize}
We conclude in Section~\ref{sec:conclusions} with some remarks
on future work.
Further details are provided in an appendix.

\paragraph{Further related work}
To the best of our knowledge, this is the first work on extending document spanners with transformation. 

String transductions~\cite{berstel2013transductions,MuschollP19,Bojanczyk22} --
a classical model 
in computer science -- have recently gained renewed
attention with the theory of polyregular functions~\cite{dagstuhlreport}. Although we can see an ET program as a single transduction (cf., Section~\ref{sec:properties} and the equivalence with NSSTs), this work has several novel contributions compared to
classical string transductions. Firstly, our framework models the process by two declarative phases (which is natural from a data management perspective), contrary to string transductions that model the task as a single process. Secondly, our work uses bag semantics, which requires revisiting some results on transduction (see Section~\ref{sec:composition}). Moreover, efficient enumeration algorithms for (non-deterministic) linear polyregular functions have not been studied before.

There are systems for transforming documents into documents (e.g.,~\cite{FASTUS}). For example, they used a combination of regular expressions with replace operators~\cite{karttunen1997replace} or parsing followed by a transduction over the parsing tree~\cite{FASTUS}. Indeed, in practice, regular expressions
support some special commands for transforming data (also called substitutions). Our study has a theoretical emphasis on information extraction and transformation, starting from the framework of document spanners. To our knowledge, previous systems neither use the expressive power of document spanners nor polyregular functions.   
 
\section{Multispanners}\label{sec:multispanners}
	
Let us first give some standard notations. By $\powerset(A)$ we denote the power set of a set $A$. Let $\mathbb{N} = \{1, 2, 3, \ldots\}$ and $[n] = \{1, 2, \ldots, n\}$ for $n \in \mathbb{N}$. For a finite alphabet $A$, let $A^+$ denote the set of non-empty words over $A$, and $A^* = A^+ \cup \{\eword\}$, where $\eword$ is the empty word. For a word $w \in A^*$, $|w|$ denotes its length (in particular, $|\eword| = 0$). A word $v \in A^+$ is a \emph{factor} (or \emph{subword}) of $w$ if there are $u_1, u_2 \in A^*$ with $w = u_1 v u_2$; $v$ is a \emph{prefix} or \emph{suffix} of $w$, if $u_1 = \eword$ or $u_2 = \eword$, respectively. For every $i \in [|w|]$, let $w[i]$ denote the symbol at position $i$ of $w$. We use DFAs and NFAs (deterministic and nondeterministic finite automata, resp.) as commonly~defined.

\paragraph{Multispans and multispanners} For a document $\doc \in \Sigma^*$ and for every $i, j \in [|\doc| {+} 1]$ with $i\leq j$, $\spann{i}{j}$ is a \emph{span of $\doc$} and its \emph{value}, denoted by $\doc\spann{i}{j}$, is the substring of $\doc$ from symbol $i$ to symbol $j{-}1$. In particular, $\doc\spann{i}{i} = \eword$ (called an \emph{empty span}) and $\doc\spann{1}{|\doc|{+}1} = \doc$. By $\spans(\doc)$, we denote the set of spans of $\doc$, and by $\spans$ we denote the set of all spans $\{\spann{i}{j} \mid i,j \in \mathbb{N},\ i\leq j\}$. 
Two spans $\spann{i}{j}$ and $\spann{i'}{j'}$ are \emph{disjoint} if $j\leq i'$ or $j'\leq i$.
A \emph{multispan} is a (possibly empty) set of pairwise disjoint spans.

Let $\varset$ be a finite set of variables. A \emph{span-tuple} (over a document $\doc$ and variables $\varset$)~\cite{fagin2015document} is a function $t\colon \varset \to \spans(\doc)$. We define a \emph{multispan-tuple} as a function $t\colon \varset \to \powerset(\spans(\doc))$ such that, for every $\varsx \in \varset$, $t(\varsx)$ is a multispan.
Note that every span-tuple $t$ can be considered as a special case of a multispan-tuple where $t(x)
= \{\spann{i}{j}\}$. 
For simplicity, we usually denote multispan-tuples in tuple-notation, for which we assume an order $<$ on $\varset$. For example, if $\varset = \{\varsx_1, \varsx_2, \varsx_3\}$ with 
$\varsx_1< \varsx_2< \varsx_3$, then multispan $t = (\, \{\spann{1}{6}\}, \, \emptyset, \, \{\spann{2}{3}, \spann{5}{7}\} \, )$
maps $\varsx_1$ to $\{\spann{1}{6}\}$, $\varsx_2$ to $\emptyset$ and $\varsx_3$ to $\{\spann{2}{3}, \spann{5}{7}\}$. 
Note that $\spann{2}{3}$ and $\spann{5}{7}$ are disjoint, while $\spann{1}{6}$ and $\spann{5}{7}$ are not, which is allowed, since they are spans of different multispans. 

A \emph{multispan-relation} (over a document $\doc$ and variables $\varset$) is a possibly empty set of multispan-tuples over $\doc$ and $\varset$.
Given a finite alphabet $\Sigma$, 
a \emph{multispanner} (over $\Sigma$ and $\varset$) is a function that maps every document  $\doc \in \Sigma^*$ to a multispan-relation over $\doc$ and $\varset$. Note that the empty relation $\emptyset$ is also a valid image of a multi\-spanner. 

\begin{example}\label{multiSpannerExample}
Let $S$ be a multispanner over alphabet $\{\ta, \tb\}$ and variables $\{\varsx, \varsy\}$ that maps every document $\doc \in \{\ta, \tb\}^*$ to the set of all multispan-tuples $t$ such that $t(\varsx) = \{\spann{i}{j}\}$ where $\doc\spann{i}{j}$ is a factor that starts and ends with $\tb$, and is not directly preceded or followed by another $\tb$,
and $t(\varsy)$ is the multispan that contains a span for each maximal unary
(i.e., of the form $\ta^+$ or $\tb^+$)
factor of $\doc\spann{i}{j}$. For example, $t \in S(\ta \ta \tb \ta \tb \tb \tb \ta \ta \tb)$ with:
\begin{equation*}
t(\varsx) = \{\spann{3}{11}\} \text{ and } t(\varsy) = \{\spann{3}{4}, \spann{4}{5}, \spann{5}{8}, \spann{8}{10}, \spann{10}{11}\}\,.
\end{equation*}
\end{example}

Similarly
to the classical framework of document spanners~\cite{fagin2015document}, it is convenient to represent multispanners over $\Sigma$ and $\varset$ by formal languages over the alphabet $\Sigma \cup \{\open{\varsx}, \close{\varsx} \mid \varsx \in \varset\}$. This allows us to represent multispanners by formal language descriptors, e.\,g., regular expressions or finite automata.

\paragraph{Representing multispans by multiref-words} In this section, we adapt the concept of ref-words (commonly used for classical document spanners; 
see \cite{FreydenbergerThompson2020, Freydenberger2019, DoleschalEtAl2019, SchmidSchweikardt_Gems})
to multispanners.

For any set $\varset$ of variables, we shall use the set $\Gamma_{\varset} = \{\open{\varsx}, \close{\varsx} \mid \varsx \in \varset\}$ as an alphabet of meta-symbols. In particular, for every $\varsx \in \varset$, we interpret the pair of symbols $\open{\varsx}$ and $\close{\varsx}$ as a pair of opening and closing parentheses.
A \emph{multiref-word} (\emph{over terminal alphabet $\Sigma$ and variables $\varset$}) is a word $w \in (\Sigma \cup \Gamma_{\varset})^*$ such that, for every $\varsx \in \varset$, the subsequence of the occurrences of $\open{\varsx}$ and of $\close{\varsx}$ is well-balanced and unnested, namely, has the form $(\open{\varsx} \close{\varsx})^k$ for some $k \geq 0$. 

Intuitively, any multiref-word $w$ over $\Sigma$ and $\varset$
uniquely describes a document $\getDoc{w} \in \Sigma^*$ and a
multispan-tuple $\getSpanTuple{w}$ as follows. First, let $\getDoc{w}$
be obtained from $w$
by erasing all symbols from $\Gamma_{\varset}$. We note that, for every $\varsx \in \varset$, every matching pair $\open{\varsx}$ and $\close{\varsx}$ in $w$ (i.\,e., every occurrence of $\open{\varsx}$ and the following occurrence of $\close{\varsx}$) uniquely describes a span of $\getDoc{w}$: ignoring all other occurrences of symbols from $\Gamma_{\varset}$, this pair encloses a factor of $\getDoc{w}$.  Consequently, we simply define that, for every $\varsx \in \varset$, $\getSpanTuple{w}(\varsx)$ contains all spans defined by matching pairs $\open{\varsx}$ and $\close{\varsx}$ of $w$. The property that the subsequence of $w$ of the occurrences $\open{\varsx}$ and $\close{\varsx}$ has the form $(\open{\varsx} \close{\varsx})^k$ for some $k \geq 0$ implies that all spans of $\getSpanTuple{w}(\varsx)$ are pairwise disjoint.

\begin{example}
Let us consider the following multiref-word over the finite alphabet $\Sigma$ and variables $\varset = \{\varsx, \varsy\}$: $w \ = \ \  \ta \ta \open{\varsx} \open{\varsy} \tb \close{\varsy} \open{\varsy} \ta \close{\varsy} \open{\varsy} \tb \tb \tb \close{\varsy} \open{\varsy} \ta \ta \close{\varsy} \open{\varsy} \tb \close{\varsy} \close{\varsx}$. By definition, $w$ represents the document $\getDoc{w} = \ta \ta \tb \ta \tb \tb \tb \ta \ta \tb$ and the multispan-tuple (from Example~\ref{multiSpannerExample}):
\begin{equation*}
\getSpanTuple{w} = (\underbrace{\{\spann{3}{11}\}}_{\varsx}, \underbrace{\{\spann{3}{4}, \spann{4}{5}, \spann{5}{8}, \spann{8}{10}, \spann{10}{11}\}}_{\varsy})
\end{equation*}
\end{example}

The advantage of the notion of
multiref-words is that it allows to easily
describe both,
a document and some multispan-tuple
over this document.
Therefore,
one can use any set of multiref-words to define a multispanner as follows.
A \emph{multiref-language} (\emph{over terminal alphabet $\Sigma$ and variables $\varset$}) is a set of multiref-words over $\Sigma$ and~$\varset$.
Any multiref-language $L$
describes the multispanner $\sem{L}$ over $\Sigma$ and $\varset$ defined as follows. For every $\doc \in \Sigma^*$: $\sem{L}(\doc) \ = \ \  \{\getSpanTuple{w} \mid w \in L \text{ \ and \ } \getDoc{w} = \doc\}$.

Analogously to classical spanners
(cf.\ \cite{SchmidSchweikardt_ICDT21,FreydenbergerThompson2020, Freydenberger2019, DoleschalEtAl2019, SchmidSchweikardt_Gems}),
we 
define the class of \emph{regular multispanners} as those multispanners $S$ with $S = \sem{L}$ for some regular multiref-language $L$. Moreover, as
done in~\cite{fagin2015document} for classical spanners, we
will
use a class of regular expressions to define a subclass of regular multispanners that shall play a central role for our extraction and transformation framework.

\paragraph{Regex multispanners} Let $\Sigma$ be a finite alphabet and let $\varset$ be a finite set of variables. We now define 
\emph{multispanner-expressions} (\emph{over $\Sigma$ and $\varset$}). Roughly speaking, these expresssions are a
particular class of regular expressions for defining
sets of
multiref-words and therefore multispanners. 
More
specifically,
a multispanner-expression $R$ (over $\Sigma$ and $\varset$) satisfies the~syntax:
\[
\begin{array}{rcl}
	R & := & \eword \ \mid \ a \in \Sigma \mid \ (R \cdot R) \ \mid \ (R + R) \ \mid \ R^* \mid \ \varsx\{R\}
\end{array}
\]
for every $\varsx \in \varset$ such that $\varsx$ does not appear in $R$.
Such
a multispanner-expression $R$ naturally defines a set of multiref-words $\lang{R}$ as follows: $\lang{\eword} = \{\eword\}$, $\lang{a} = \{a\}$, $\lang{R \cdot R'} = \lang{R}\cdot \lang{R'}$, $\lang{R + R'} = \lang{R} \cup \lang{R'}$, $\lang{R^*} = \lang{R}^*$, and $\lang{\varsx\{R\}} \ = \ \ \{\open{\varsx}\} \cdot \: \lang{R} \cdot \{\close{\varsx}\}$,
where, for every $L, L' \subseteq \Sigma^*$, $L \cdot L' = \{uv \mid u \in L \wedge v \in L'\}$, $L^0 = \{\eword\}$, $L^i = L \cdot L^{i-1}$ for every $i \geq 1$, and $L^* = \bigcup_{i=0}^{\infty} L^i$. As usual, we use $R^+$ as a shorthand for $(R\cdot R^*)$.

One can easily prove (e.g., by induction over the size of the expressions) that any multispanner-expression defines a multiref-language, since we do not allow expressions of the form $\varsx\{R\}$ whenever~$R$ mentions $\varsx$. 
Thus, we can
define the \emph{multispanner $\sem{R}$ specified by~$R$} as $\sem{R} = \sem{\lang{R}}$.
Furthermore, we say that a multispanner~$S$ is a \emph{regex multispanner} if $S = \sem{R}$ for some multispanner-expression $R$. 

\begin{example}
$R \ := \ \ (\eword + \Sigma^* \ta^+) \cdot \varsx\big\{ (\varsy\{ \tb^+ \} \cdot \varsy\{ \ta^+ \})^* \cdot \varsy\{ \tb^+ \} \big\} \cdot (\ta^+ \Sigma^* + \eword)$ is a multispanner-expression with $\sem{R}$ being the multispanner $S$ from Example~\ref{multiSpannerExample}; thus, $S$ is a regex multispanner.
\end{example}

\paragraph{Comparison with classical spanners} Multispanners are designed to naturally extend the classical model of spanners from~\cite{fagin2015document} to the setting where variables are mapped to sets of spans instead of single spans. Let us discuss a few particularities of our definitions.

We first note that since classical span-tuples, span-relations and spanners (in the sense of~\cite{fagin2015document}) can be interpreted as special multi\-span-tuples, multispan-relations and multispanners, respectively, our framework properly extends the classical spanner framework. 

A multispan-tuple $t$ allows $\spann{i}{j} \in t(\varsx)$ and $\spann{k}{l} \in t(\varsy)$ with $i \leq k < j \leq l$ and $\varsx \neq \varsy$ (and this is also the case for classical span-tuples). However, 
$\spann{i}{j}, \spann{k}{l} \in t(\varsx)$ with $i \leq k < j \leq l$ is not possible for multispan-tuples, since then representing $\spann{i}{j}$ and $\spann{k}{l}$ by parentheses $\open{\varsx} \ldots \close{\varsx}$ in the document cannot be distinguished from representing $\spann{i}{l}$ and $\spann{k}{j}$. Furthermore, 
for distinct $s, s' \in t(\varsx)$ we require
that $s$ and $s'$ are disjoint, which is motivated by the fact that without this restriction, the subsequence of all $\open{\varsx}$ and $\close{\varsx}$ occurrences could be an arbitrary well-formed parenthesised expression (instead of a sequence $(\open{\varsx} \close{\varsx})^k$); thus, recognising whether a given string over $\Sigma \cup \Gamma_{\varset}$ 
is a proper multiref-word cannot be done by an NFA, as is the case for classical spanners.

Our main motivation for multispanners is that they can express information extraction tasks that are of interest in the context of our extract and transform framework (and that cannot be represented by classical spanners in a convenient way). However, there are other interesting properties of multispanners not related to their application in our extract and transform framework, which deserve further investigation. For example, if $L$ and $L'$ are multiref-languages (i.\,e., $\sem{L}$ and $\sem{L'}$ are multispanners), then $L \cup L$, $L \cdot L'$ and $L^*$ are also multi\-ref-languages (and therefore $\sem{L \cup L'}$, $\sem{L \cdot L}$ and $\sem{L^*}$ are also multispanners). For classical spanners, this is only true for the union. Consequently, multispanners show some robustness not provided by classical spanners.

\section{Extract transform  framework}\label{sec:framework}
	
Let us first present some more notation. As usual, $f: A \to B$ denotes a function from $A$ to $B$. When $f$ is partial, we write $f: A \pmap B$ and use $f(a) = \nondef$ to indicate that $f$ is undefined for element $a \in A$. Moreover, $\dom(f) = \{a \in A \mid f(a) \neq \bot\}$ denotes the domain of~$f$. Every partial function $T : A^* \pmap B^*$ is called a \emph{string-to-string function} with input alphabet $A$ and output alphabet $B$.

The main idea
of our extraction and transformation framework is to use
a document spanner
to extract information
and then use a string-to-string function to convert this information into new documents.
For this purpose, it is necessary to convert the output of a spanner (i.e., span-tuples) into the input of a
string-to-string function
(i.e., strings). In the following, we present
a unique way
to encode multispan-tuples into multiref-words, which will serve as our intermediate object between extraction and transformation.

\paragraph{A unique multiref-word representation} The representation of documents and multispan-tuples by multiref-words allows us to define multispanners by multiref-languages.
But this
representation is not unique, which is inconvenient if we want to use it for
encoding
multispan-tuples as strings.
As a remedy, we adopt the following approach: We
represent a document $\doc \in \Sigma^*$ and multispan-tuple $t$ as a
multiref-word $w$ such that factors of $w$ 
that belong to $\Gamma_{\varset}^*$ (i.e., factors between two letters of $\doc$)
have the form:
\begin{equation*}
	(\close{\varsx} \altop \eword) (\open{\varsx} \close{\varsx} \altop \eword) (\open{\varsx} \altop \eword) (\close{\varsy} \altop \eword) (\open{\varsy} \close{\varsy} \altop \eword) (\open{\varsy} \altop \eword) \ldots\,.
\end{equation*}
Specifically, let $\doc \in \Sigma^*$ be a document and $t$ a multispan-tuple over $\doc$ and variables $\varset$. For a variable $\varsx \in \varset$ and $i \in [|\doc| +1]$, define $t[i,x] \ = \ \ (\close{\varsx})^c (\open{\varsx} \close{\varsx})^e (\open{\varsx})^o$,
where $c, e, o \in \{0,1\}$
with $c = 1$ iff $\spann{j}{i} \in t(\varsx)$ for some $j \neq i$; $e = 1$ iff $\spann{i}{i} \in t(\varsx)$; and $o = 1$ iff $\spann{i}{j} \in t(\varsx)$ for some $j \neq i$. 
E.\,g., for the tuple $t$ in Example~\ref{multiSpannerExample}, we have that $t[3,\varsx] = \open{\varsx}$ and $t[4,\varsy] = \close{\varsy} \open{\varsy}$. 

Let $\varset = \{\varsx_1, \varsx_2, \ldots, \varsx_m\}$ with $\varsx_1 \preceq \varsx_2 \preceq \ldots \preceq \varsx_m$, where $\preceq$ is some fixed linear order on $\varset$. 
For every $i \in [|\doc| + 1]$, we define $t[i] \ := \ \ t[i, \varsx_1] \cdot t[i, \varsx_2] \cdot \ldots \cdot t[i, \varsx_m]$.
We can then define the encoding of $t$ and $\doc$ as the multiref-word:
\begin{equation*}
	\nrefw{t, \doc} \ = \ \ t[1] \cdot \doc[1] \cdot t[2] \cdot \doc[2] \cdot \ldots \cdot \doc[|\doc|] \cdot t[|\doc|{+}1].
\end{equation*}
Note that $\nrefw{t, \doc}$ is
a multiref-word, $\getDoc{\nrefw{t, \doc}} = \doc$ and $\getSpanTuple{\nrefw{t, \doc}} = t$.
Thus, $\nrefw{t, \doc}$ is a correct and unique encoding for $t$ and $\doc$ as a multiref-word.
Next, we use this encoding to define \emph{extract-transform~programs}. 

\paragraph{Extract-transform programs} 
An \emph{extract-transform program}
(for short: \emph{ET program})
is a pair $(E, T)$ such that $E$ is a multispanner over $\Sigma$ and $\varset$, and $T$ is a string-to-string function with input alphabet $\Sigma \cup \Gamma_{\varset}$ and some output alphabet $\Omega$.
The semantics $\sem{E \cdot T}$ 
of an ET program $(E,T)$ is defined as follows. For an input document
$\doc$, we first apply $E$ on $\doc$, which produces a
multispan-relation $E(\doc)$. Then, for every multispan-tuple $t \in
E(\doc)$, we apply $T$ on the multiref-word $\nrefw{t,\doc}$ producing
a new document in the output. Note that while $\nrefw{\cdot, \cdot}$
is injective, the function $T$ might not be, which means that for $t,
t' \in E(\doc)$ with $t \neq t'$ we might have $T(\nrefw{t,\doc}) =
T(\nrefw{t',\doc})$. This leads to duplicates, but, as explained in
the introduction, we keep them, since we want each distinct $t \in
E(\doc)$ of the extraction phase to correspond to a transformed
document in the output. Hence, the output set is a \emph{bag} (also
known as \emph{multiset}).
As common in the literature, we use the notation
$\multiset{\ldots}$ for bags (we give a detailed definition of bags
in Appendix, Section~\ref{sec:defnAppendix}).  
Formally, for every $\doc \in \Sigma^*$, we define $\sem{E \cdot T}(\doc) \ = \ \multiset{T(\nrefw{t, \doc}) \mid t \in E(\doc)}$.

\paragraph{Polyregular ET programs} In the introduction, we have already discussed several examples of ET programs on an intuitive level. The multispanners of these examples were explicitly stated as multispanner expressions. The transformation functions were given as ad-hoc specifications, but they were all instances of 
so-called \emph{polyregular functions}. A polyregular function is a string-to-string function that is regular and has
a growth at most polynomial
in the size of the output string. Formally, $T$ is polyregular if it can be defined by a \emph{pebble transducer}:
a two-way transducer enhanced with a fixed number of pebbles (i.e., pointers) that the machine can drop or remove over the input word in a nested fashion and can read when it passes over them (see~\cite{engelfriet2002two} for a formal definition). Polyregular functions
form a robust class that has found several characterizations,
among them for-transducers, MSO interpretations,
and
others~\cite{Bojanczyk22}. Furthermore, polyregular functions are closed under composition and have good theoretical properties
(see \cite{Bojanczyk22} for an overview).
For this reason, we introduce the following definition.

A \emph{polyregular ET program} is an ET program $(E,T)$ where $E$ is a regex multispanner and $T$ is a polyregular function. All the ET programs mentioned in the introduction were examples of polyregular ET programs. In the example ET program $(E_2, T_2)$ discussed there, we need the following string-to-string function $T_2$ in the transformation phase:
\begin{align*}
&u \: \# \open{\varsx} u_{\varsx} \close{\varsx} ; \open{\varsy} u_{\varsy} \close{\varsy} ; \open{\varsz} v_{\varsz, 1} \close{\varsz}; \open{\varsz} v_{\varsz, 2} \close{\varsz} \ldots \open{\varsz} v_{\varsz, k} \close{\varsz} \: u' \: \mapsto\\
&u_{\varsy} \Vtextvisiblespace u_{\varsx}:v_{\varsz, 1} \# u_{\varsy} \Vtextvisiblespace u_{\varsx}:v_{\varsz, 2} \# \ldots \# u_{\varsy} \Vtextvisiblespace u_{\varsx}:v_{\varsz, k} \#\,,
\end{align*}
where $u, u' \in \Sigma^*$, $u_{\varsx}, u_{\varsy}, v_{\varsz, 1}, v_{\varsz, 2}, \ldots, v_{\varsz, k} \in \widehat{\Sigma}^*$.
It can be verified that this is indeed a polyregular function, by defining a pebble transducer or a for-transducer for it, which, due to space constraints, is done in the Appendix (see Section~\ref{sec:polyregAppendix}). 

Let us explain why the class of polyregular functions is a good choice for our setting. As mentioned before, polyregular functions can be seen as the class of MSO string-to-string functions, i.e., they have the same expressive power as MSO interpretations. Furthermore, polyregular functions are a robust and well-established class with several representations and good closure properties. Polyregular functions also capture most of the transformations that occur in practice. For example, practical regex use substitutions for transforming strings~\cite{substitution1,substitution2}. Interestingly, one can define all these operators by combining regex multispanners and polyregular functions. There are even substitutions in regex that define non-linear polyregular functions (e.g., see the operator~\verb*|$'| in [1,2]). Thus, polyregular functions are a natural choice. In the following we will focus on a fragment called \emph{linear ET programs}.

\paragraph{Linear ET programs}
We restrict attention to string-to-string functions of linear growth.
Specifically, we say that a string-to-string function $T$ has \emph{linear growth} if there exists a constant $c$, such that $|T(w)| \leq c \cdot |w|$ for every string $w$.
A \emph{linear ET program} is a polyregular ET program $(E,T)$ where $T$ has linear~growth. 

The class of linear polyregular functions forms a proper and well-behaved subclass of polyregular functions, which
allows for several equivalent representations like two-way
transducers, MSO transductions, or deterministic streaming string
transducers~\cite{Bojanczyk22,AlurC10}
(this last formalism will be crucial for the rest of this paper).
The string-to-string function $T_1$ (see Page~\pageref{T1Function})
that we used in the first example discussed in the introduction is an example of a linear polyregular function, and we shall see a deterministic streaming string transducer for it later on.

\paragraph{Deterministic streaming string transducers} Let $\Reg$ be a finite set of registers and $\Omega$ a finite alphabet. An \emph{assignment}\footnote{We assume that assignments are partial functions, in contrast to~\cite{AlurD11} that assume that assignments are functions. 
  Note that this generalization of the model in~\cite{AlurD11} comes without modifying the expressive power of standard
  streaming string transducers.
} is a partial function $\sigma: \Reg \pmap (\Reg \cup \Omega)^*$ that assigns each register $X \in \Reg$ to a string $\sigma(X)$ of registers and letters from $\Omega$.
We define the extension $\hat{\sigma}: (\Reg \cup \Omega)^* \pmap (\Reg \cup \Omega)^*$ of an assignment such that $\hat{\sigma}(\alpha_1 \ldots \alpha_n) = \sigma(\alpha_1) \cdot \ldots \cdot \sigma(\alpha_n)$ for every string $\alpha_1\ldots \alpha_n \in (\Reg \cup \Omega)^*$ where $\sigma(\oout) = \oout$ for every $\oout \in \Omega$. Further, we assume that $\hat{\sigma}(\alpha_1 \ldots \alpha_n)$ is undefined iff $\sigma(\alpha_i)$ is undefined for some $i \in [n]$. Given two assignments $\sigma_1$ and $\sigma_2$, we define its composition $\sigma_1 \circ \sigma_2$ as a new assignment such that $[\sigma_1 \circ \sigma_2](X) = \hat{\sigma}_1(\sigma_2(X))$. We say that an assignment $\sigma$ is \emph{copyless} if, for every $X \in \Reg$, there is at most one occurrence of $X$ in all the strings $\sigma(Y)$ with $Y \in \Reg$.
\begin{example}
	Consider $\Reg= \{X, Y\}$ and $\Omega = \{a,b\}$, and the following assignments where we write $X:= \alpha$ to mean $\sigma(X) = \alpha$:
	\[
	\begin{array}{ccccc}
		\begin{array}{rl}
			\sigma_1\colon \!\! & X := aXa \\
			& Y := XY
		\end{array} & &  
		\begin{array}{rl}
			\sigma_2\colon \!\!  & X := bXb \\
			& Y := b
		\end{array} &  &
		\begin{array}{rl}
			\sigma_3: \!\! & X := baXab \\
			& Y := b
		\end{array}
	\end{array}
	\]
	One can check that $\sigma_1 \circ \sigma_2 = \sigma_3$. Also, $\sigma_2$ and $\sigma_3$ are copyless assignments, but $\sigma_1$ is not. 
\end{example} 
Note that copyless assignments are closed under composition, namely, if $\sigma_1$ and $\sigma_2$ are copyless assignments, then $\sigma_1 \circ \sigma_2$ is copyless as well. We denote by $\Asg(\Reg, \Omega)$ the set of all copyless assignments over registers $\Reg$ and alphabet $\Omega$. 

A \emph{deterministic streaming string transducer} (DSST)~\cite{AlurC10} is a tuple
$
\nsst = (Q, \Sigma, \Omega,  \Reg, \Delta, q_0, F)$,
where $Q$ is a finite set of states, $\Sigma$ is the input alphabet, $\Omega$ is the output alphabet, $\Reg$ is a finite set of registers, $\Delta \colon Q \times \Sigma \pmap \Asg(\Reg, \Omega) \times Q$ is the transition function, $q_0$ is the initial state, and $F: Q \pmap (\Reg \cup \Omega)^*$ is a final partial function such that, for every $q \in Q$ and $X \in \Reg$, $X$ appears at most once in $F(q)$.
Intuitively, if $F(q)$ is defined, then $q$ is a final
(i.e., accepting)~state.
A \emph{configuration} of a DSST $\nsst$
is a pair $(q, \nu)$ where $q \in Q$ and $\nu \in \Val(\Reg, \Omega)$, and $\Val(\Reg, \Omega)$ is the set of all assignments $\Reg \pmap \Omega^*$, which are called \emph{valuations}.
A \emph{run} $\rho$ of $\nsst$ over a string $a_1 \ldots a_n$ is a sequence of configurations $(q_i, \nu_i)$ of the form:
\begin{equation*}
	\rho \ := \ (q_0, \nu_0) \xrightarrow{a_1/\sigma_1} (q_1, \nu_1) \xrightarrow{a_2 /\sigma_2} \ldots \xrightarrow{a_n/\sigma_{n}} (q_{n}, \nu_n)
\end{equation*}
such that $\Delta(q_i, a_{i+1}) = (\sigma_{i+1}, q_{i+1})$, $\nu_{0}$ is the empty assignment, i.\,e. $\nu_{0}(X) = \eword$ for every $X \in \Reg$, and $\nu_{i+1} = \nu_i \circ \sigma_{i+1}$ for every $i < n$.
A run $\rho$ is called
an \emph{accepting run} if $\nu_n(F(q_n))$ is defined.
The
\emph{output} of an accepting run $\rho$
is defined
as the string $\out(\rho) := \nu_n(F(q_n))$. 

Since $\Delta$ is a partial function, we can see that DSST are deterministic,
i.e.,
for every input word $w$, there exists at most one run $\rho$.
Thus,
every DSST $\nsst$ defines a string-to-string function  $\sem{\nsst}$ such that $\sem{\nsst}(w) = \out(\rho)$ iff $\rho$ is the run of $\nsst$ over $w$ and $\rho$ is accepting.
\begin{example}\label{example:names-sst}
Recall the string-to-string function $\firstExampleTrans \ = u \open{\varsx} v_{\varsx} \close{\varsx}; \open{\varsy} v_{\varsy} \close{\varsy} u' \:\: \mapsto \:\: v_{\varsy} \Vtextvisiblespace \, v_{\varsx}$,
where $u, u' \in \Sigma^*$ and $v_{\varsx}, v_{\varsy} \in
\widehat{\Sigma}^*$ (see Page~\pageref{T1Function}).
This function is implemented by the following DSST:

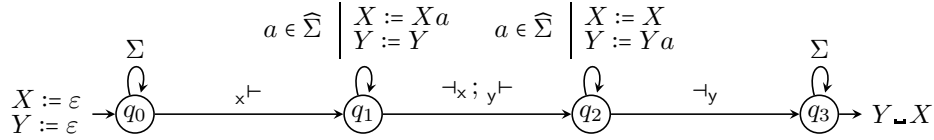
\begin{center}
\begin{tikzpicture}[defaultstyle,
	every text node part/.style={align=center},
	mystate/.style={state, inner sep=0pt, minimum size=5mm},
	]
	
	\node[mystate, initial left, initial text={$\begin{array}{l}
			X := \eword \vspace{-1mm} \\ 
			Y := \eword
		\end{array}\!\!\!$}] (s0) at (0, 0) {$q_0$};
	
	\draw[->] (s0) edge[loop above, pos=0.5] node {$\Sigma$} (s0);

	\node[mystate] (s1) at (3, 0) {$q_1$};
	\node[mystate] (s2) at (6, 0) {$q_2$};
	\node[mystate, accepting by arrow, accepting text={$Y \Vtextvisiblespace\,  X$}] (s3) at (9, 0) {$q_3$};

	\draw[->] (s0) edge[sloped] node {$
		\open{\varsx} 
		$} (s1);

	\draw[->] (s1) edge[loop above, pos=0.5] node {$a \in \widehat{\Sigma} \ \ \begin{array}{|l}
			X := X a \vspace{-1mm} \\ 
			Y := Y
		\end{array}$} (s1);

	\draw[->] (s2) edge[loop above, pos=0.5] node {$a \in \widehat{\Sigma} \ \ \begin{array}{|l}
			X := X \vspace{-1mm} \\ 
			Y := Y a
		\end{array}$} (s2);

	\draw[->] (s1) edge[sloped] node {$
		\close{\varsx} \, ; \, \open{\varsy} 
		$} (s2);

	\draw[->] (s2) edge[sloped] node {$
		\close{\varsy} 
		$} (s3);

	\draw[->] (s3) edge[loop above, pos=0.5] node {$\Sigma$} (s3);

\end{tikzpicture}
\end{center}

\end{example}

DSSTs form a well-behaved class of linear string-to-string functions, equivalent to the previously mentioned class of linear polyregular functions (see~\cite{Bojanczyk22,AlurC10}). Moreover, DSSTs perform a single pass over the input string,
while 
other equivalent models (e.g., two-way transducers)
may
do several passes. This streaming behavior
of DSSTs
will
turn out to
be very useful for designing algorithms for linear ET programs. For this reason, in the sequel we use DSSTs as our model of string-to-string functions for linear ET programs.

\section{Expressiveness of linear extract transform programs}\label{sec:properties}
	
\paragraph{Nondeterministic streaming string transducers} We
extend DSSTs
(cf.\ 
Section~\ref{sec:framework}) to the nondeterministic case and bag semantics. Let us explain the model by pointing out its differences to DSSTs.\footnote{Note that nondeterministic streaming string transducers with \emph{set semantics} instead of bag semantic have already been introduced in~\cite{AlurD11}.}

A \emph{nondeterministic streaming string transducer} (NSST) is a tuple $\nsst = (Q, \Sigma, \Omega,  \Reg, \Delta, I, F)$,
where $Q$, $\Sigma$, $\Omega$, $\Reg$ and $F$ have the same meaning as for DSSTs. The partial function $I: Q \pmap \Val(\Reg, \Omega)$ plays the role of the initial state, i.\,e., a state $q$ is a possible initial state if $I(q)$ is defined, and in this case $I(q)$ is an initial valuation of the registers. 
Moreover, $\Delta$ is not a
(partial)
function from $Q \times \Omega$ to $\Asg(\Reg, \Omega) \times Q$ anymore, but a \emph{bag} of transitions,
i.e., $\Delta$ is a bag of elements in $Q \times \Sigma \times \Asg(\Reg, \Omega) \times Q$.

The semantics of the model are as expected. A run over a string $a_1 \ldots a_n$ is a sequence of configurations $(q_i, \nu_i)$ of the form:
\begin{equation*}
	\rho \ := \ (q_0, \nu_0) \xrightarrow{a_1/\sigma_1} (q_1, \nu_1) \xrightarrow{a_2 /\sigma_2} \ldots \xrightarrow{a_n/\sigma_{n}} (q_{n}, \nu_n)
\end{equation*}
such that $I(q_0)$ is defined and $\nu_0 = I(q_0)$, $(q_i, a_{i+1}, \sigma_{i+1}, q_{i+1}) \in \Delta$ and $\nu_{i+1} = \nu_i \circ \sigma_{i+1}$ for every $i < n$. As for DSSTs, $\rho$ is an \emph{accepting run} if $\nu_n(F(q_n))$ is defined, and the output of an accepting run $\rho$ is the string $\out(\rho) = \nu_n(F(q_n))$. We define $\Runs_{\nsst}(a_1 \ldots a_n)$ as the bag of all accepting runs of $\nsst$ over~$a_1 \ldots a_n$.

Finally, we define the semantics of an NSST $\nsst$ over a string $w \in \Sigma^*$ as the
bag:
\begin{equation*}
	\sem{\nsst}(w) \ = \ \multiset{\, \out(\rho) \, \mid \, \rho \in \Runs_{\nsst}(w)\, }.
\end{equation*}

\paragraph{Equivalence of linear ET programs and NSSTs} We show that 
linear ET programs and NSSTs are equivalent. This is a fundamental insight with respect to the expressive power of ET programs. Moreover, the fact that the two-stage model of linear ET programs can be described by a single NSST will be important for our enumeration algorithm 
(Section~\ref{sec:enumeration}) and their composition (Section~\ref{sec:composition}).

We first discuss how linear ET programs can be transformed into NSSTs. 
Every multiref-word over $\Sigma$ and $\varset$ describes a document $\getDoc{w}$ and a multispan-tuple $\getSpanTuple{w}$.
Hence,
we can extend the encoding $\nrefw{\cdot}$ to multiref-words by setting $\nrefw{w} = \nrefw{\getSpanTuple{w}, \getDoc{w}}$. Intuitively speaking, applying the function $\nrefw{\cdot}$ on a multiref-word $w$
simply means that every maximal factor $u \in \Gamma_{\varset}^*$ of $w$ is re-ordered according to the order $\preceq$ on $\varset$, and superfluous matching brackets $\open{\varsx} \close{\varsx}$ are removed (since several of those in the same maximal factor over $\Gamma_{\varset}$ would describe the same empty span several times). We define $\nrefw{L} = \{\nrefw{w} \mid w \in L\}$ for multiref-languages $L$.

Let us consider a linear ET program $\sem{E \cdot T}$, i.\,e., $E$ is a regex multispanner represented by some multispanner-expression
$r$, and $T$ is a linear polyregular function represented by some DSST $\sst$. The high-level idea of the proof is to construct an NSST $\sst'$ that simulates a DFA $M$ for $\nrefw{\lang{r}}$ and the DSST $\sst$ in parallel. More precisely, we read an input $\doc \in \Sigma^*$, but between reading a symbol $\doc[i]$ and the next symbol $\doc[i+1]$, we pretend to read a sequence of symbols from $\Gamma_{\varset}$ with $\sst$ and $M$ at the same time. Thus, we virtually read some multiref-word $w$ with the property $\getDoc{w} = \doc$. We need the DSST $\sst$ for producing an output on that multiref-word, and we need the DFA $M$ to make sure that the virtual multiref-word has the form $\nrefw{t, \doc}$ for some $t \in E(\doc)$.

One may wonder why we want $M$ to be a DFA rather than an NFA. The reason is:
Having to deal with bag semantics (rather than just set semantics) makes things more complicated.
In particular, this means that we need $M$ to be deterministic to have a one-to-one correspondence between the accepting runs of $\sst'$ and the accepting runs of $M$ on the corresponding multiref-word.
Otherwise,
if $M$ was an NFA,
the different possible accepting paths of $M$ on the same multiref-word would translate into different accepting paths of $\sst'$ which would cause erroneous duplicates in $\sem{\sst'}(\doc)$.

We can transform $r$ into a DFA $M$ for $\nrefw{\lang{r}}$ by standard automata constructions, but the fact that $M$ needs to be deterministic means that the construction is
(one-fold)
exponential in $|r|$, and the fact that it has to accept $\nrefw{\lang{r}}$ and not just $\lang{r}$ means that the construction is also
(one-fold)
exponential in $|\varset|$.
In summary, we obtain the following result.

\begin{theorem}\label{ETtoNSST}
  Given a regex multispanner $E$ over $\Sigma$ and $\varset$ (represented by a multispanner-expression $r$), and a linear polyregular string-to-string function $T$ with input alphabet $\Sigma \cup \Gamma_{\varset}$ (represented by a DSST with $h$ states),
we can construct an NSST $\sst$ with $\sem{\sst} = \sem{E \cdot T}$ in time $\bigo(2^{4|r|+9|\varsetSmallFont|} |\Sigma|^3 |\varset|^2 h^2)$.
\end{theorem}

Let us now move on to representing general NSSTs by ET programs. When an NSST $\nsst$ is in a state $p$ and reads a symbol $b \in \Sigma$, then the number of possible nondeterministic choices it can make is the sum over all the multiplicities of elements $(p, b, \sigma, q) \in \Delta$ (recall that $\Delta$ is the bag of transitions). We shall call this number the \emph{nondeterministic branching factor of $p$ and $b$}.
The
\emph{nondeterministic branching factor} of $\nsst$ is the maximum over all the nondeterministic branching factors of $p$ and $b$ for all $(p,b) \in Q\times \Sigma$.

The only obstacle in the simulation of an NSST by a DSST is that the latter cannot deal with the nondeterministic choices. However, in linear ET programs, a DSST gets an annotated version of the actual document $\doc$ as input, i.\,e., a multiref-word $w$ with $\getDoc{w} = \doc$. Consequently, the DSST could interpret the additional information given by $w$ in the form of the symbols from $\Gamma_{\varset}$ as information that determines which nondeterministic choices it should make. More formally, we can
construct a regex multispanner that, for every $i \in [1, 2, \ldots, |\doc|]$, nondeterministically chooses some $\varsx_{\ell} \in \{\varsx_1, \varsx_2, \ldots, \varsx_m\}$, where $m$ is the NSST's nondeterministic branching factor, and puts the empty span $\spann{i}{i}$ into $t(\varsx)$. On the level of multiref-words, this simply means that every symbol is preceded by $\open{\varsx_\ell} \close{\varsx_\ell}$ for some $\ell$. Such an occurrence of $\open{\varsx_\ell} \close{\varsx_\ell}$ can then be interpreted by the DSST as an instruction to select the $\ell^{\text{th}}$ nondeterministic choice when processing the next symbol from $\Sigma$. 
In summary, we obtain the following~result.

\begin{theorem}\label{NSSTtoET}
Given an NSST $\sst$ with $n$ states, input alphabet $\Sigma$ and nondeterministic branching factor $m$,
we can construct a regex multispanner $E$ over $\Sigma$ and $\{\varsx_1, \varsx_2, \ldots, \varsx_{\max\{n, m\}}\}$ and a linear poly\-regular function $T$ with $\sem{E \cdot T} = \sem{\sst}$. Moreover, $E$ is represented by a multispanner-expression $r$ with $|r| = \bigo(|\Sigma| + \max\{n, m\})$, $T$ is represented by a DSST $\sst'$ with $\bigo(\max\{n, m\} \cdot n)$ states, and both $r$ and $\sst'$ can be constructed in time $\bigo(|\sst| + n \cdot \max\{n, m\} + |\Sigma|)$.
\end{theorem}

Let us conclude this section by mentioning that the statement of Theorem~\ref{ETtoNSST} is also true, if $E$ is a multispanner represented by an NFA (this stronger version is proven in the appendix).

\section{Evaluation of linear ET programs}\label{sec:enumeration}
	
In this section, we present the main technical result of the paper, regarding the evaluation of linear ET programs. Specifically, we consider the following enumeration problem 
(let $E$ be a regex multispanner and $T$ a linear polyregular function).
\begin{center}
	\framebox{
		\begin{tabular}{rl}
			\textbf{Problem:} & $\textsc{EnumLinearET}[(E, T)]$ \\
			\textbf{Input:} & A document $\doc$ \\
			\textbf{Output:} & Enumerate $\sem{E \cdot T}(\doc)$
		\end{tabular}
	}
\end{center}

\noindent Notice that $\sem{E \cdot T}(\doc)$ is a bag; thus, the task is to produce an enumeration that contains each element of the bag exactly once, e.\,g., $(a, a, b, c, a)$ is a possible enumeration of $\multiset{b, a, a, a, c}$.

As usual,
we measure the running time in \emph{data complexity}, namely, we assume that $E$ and $T$ are fixed. Given this assumption, we can assume that $E$ is given as a multispanner-expression, and $T$ as a DSST. Otherwise, we can convert $E$ and $T$ to satisfy this requirement. 

For this problem, we strive for an \emph{enumeration algorithm with linear preprocessing and output-linear delay}, i.\,e., in a \emph{preprocessing phase} it receives the input and produces some data structure $\texttt{DS}$ which encodes the expected output, and in the following \emph{enumeration phase} it produces a sequential enumeration $w_1, \ldots, w_\ell$ of the results from $\texttt{DS}$. Moreover, the time for the preprocessing phase is~$\bigo(|\doc|)$, the time for producing $w_1$
is less than $c\cdot |w_1|$, and the time between producing $w_{i-1}$ and $w_i$ is less than $c \cdot |w_i|$, for some fixed constant $c$ that does not depend on the input. As it is common~\cite{Segoufin13},
we assume the computational model of \emph{Random Access Machines} (RAM) with uniform cost measure and addition and subtraction as basic operations~\cite{aho1974design}. 
We obtain the following result.

\begin{theorem}\label{theo:enum}
$\textsc{\upshape{EnumLinearET}}[(E, T)]$ admits an enumeration algorithm with linear preprocessing time and output-linear delay.
\end{theorem}

Due to space restrictions, all the details and the analysis of the enumeration algorithm are deferred to the Appendix.
In the following, we
highlight
the main technical challenges of this algorithm.

\newcommand{\nsstET}{\nsst_{E,T}}

The first step of the algorithm is to apply Theorem~\ref{ETtoNSST} and convert the pair $(E,T)$ into an NSST $\nsstET$. This takes time exponential in $|E|$;
nevertheless, it does not depend on $|\doc|$ and so we can consider it as constant time. Therefore, our enumeration algorithm aims for computing $\sem{\nsstET}(\doc)$ with linear time preprocessing and output-linear delay. 

For evaluating $\nsstET$ over $\doc$, the first challenge that we need to overcome is that its runs could maintain registers with content that is not used at the end of the run. For an illustration, consider the following NSST (note that $I(q)$ is the valuation $X := \eword$ and $Y := \eword$, and $F(p) = X$ and $F(r) = Y$):
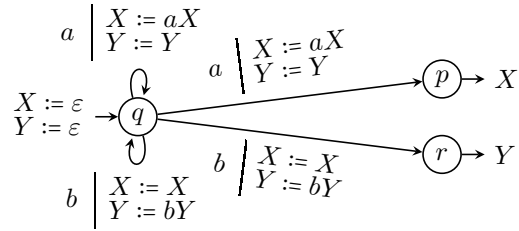
\begin{center}
\begin{tikzpicture}[defaultstyle,
	every text node part/.style={align=center},
	mystate/.style={state, inner sep=0pt, minimum size=5mm},
	]
	
	\node[mystate, initial left, initial text={$\begin{array}{l}
			X := \eword \vspace{-1mm} \\ 
			Y := \eword
		\end{array}\!\!\!$}] (s0) at (0, 0) {$q$};
	
	\draw[->] (s0) edge[loop above, pos=0.5] node {$a \ \ \begin{array}{|l}
			X := aX \vspace{-1mm} \\ 
			Y := Y
		\end{array}
		$} (s0);
	
	\draw[->] (s0) edge[loop below, pos=0.5] node {$b \ \ \begin{array}{|l}
			X := X \vspace{-1mm} \\ 
			Y := bY
		\end{array}$} (s0);
		
	\node[mystate, accepting by arrow, accepting text={$X$}] (s1) at (4, 0.5) {$p$};
	\node[mystate, accepting by arrow, accepting text={$Y$}] (s2) at (4, -0.5) {$r$};
	
	\draw[->] (s0) edge[sloped] node {$
		a \ \ \begin{array}{|l}
			X := aX \vspace{-1mm} \\ 
			Y := Y
		\end{array}
		$} (s1);
	
	\draw[->] (s0) edge[swap, sloped] node {$
		b \ \ \begin{array}{|l}
			X := X \vspace{-1mm} \\ 
			Y := bY
		\end{array}
		$} (s2);
	
\end{tikzpicture}
\end{center}

For each input word $w$ with $n$ $a$-symbols and $m$ $b$-symbols, the NSST outputs $a^n$ if $w$ ends with $a$ and $b^m$ if $w$ ends with $b$. Consequently, every run on a word that ends with $a$ produces ``garbage'' in register $Y$, since the content of this register is not used for the output (and analogously with register $X$ for inputs that end with $b$).
This behavior of storing ``garbage'' will be problematic for our enumeration approach, given that the delay depends on the (potentially useless) contents of the registers.

Given the above discussion, we formalize
the notion of ``garbage''
as follows. 
Consider an NSST $\nsst = (Q, \Sigma, \Omega,  \Reg, \Delta, I, F)$.
For $u \in (\Reg \cup \Omega)^*$, let $\reg(u)$ be the set of all registers $X \in \Reg$ that appear in $u$.
For
$\sigma: \Reg \pmap (\Reg \cup \Omega)^*$ let $\reg(\sigma) = \bigcup_{X \in \dom(\sigma)} \reg(\sigma(X))$, namely,
$\reg(\sigma)$ is
the set of all registers used by $\sigma$.  
We say that $\nsst$ is \emph{garbage-free} if, and only if, for every string $w = a_1\ldots a_n$ and every accepting run $\rho$ of the form:
\begin{equation*}
	\rho \ := \ (q_0, \nu_0) \xrightarrow{a_1/\sigma_1} (q_1, \nu_1) \xrightarrow{a_2 /\sigma_2} \ldots \xrightarrow{a_n/\sigma_{n}} (q_{n}, \nu_n)
\end{equation*}
it holds that $\dom(\nu_i) = \reg(\sigma_{i+1})$ for every $i < n$, and $\dom(\nu_n) = \reg(F(q_n))$. In other words, the registers $\dom(\nu_i)$ that we have filled with content so far coincide with the registers $\reg(\sigma_{i+1})$ that we use on the right hand sides of the next assignment. 

The first challenge is to show how to make NSSTs garbage-free.
\begin{proposition}\label{prop:garbageFree}
	For every NSST $\nsst$, there exists a garbage-free NSST $\nsst'$ such $\sem{\nsst}(w) = \sem{\nsst'}(w)$ for every string $w$.
\end{proposition}

The construction of Proposition~\ref{prop:garbageFree} causes a blow-up that is exponential in the number of registers, since we turn an NSST with $|Q|$ states into an NSST with $|Q| \times 2^{|\Reg|}$ states. Of course, if we start with a garbage-free NSST, this blow-up can be avoided. Interestingly, we can show that one can check the garbage-free property in polynomial time.
\begin{proposition} \label{prop:garbageFreePTime}
	Given an NSST $\nsst = (Q, \Sigma, \Omega,  \Reg, \Delta, I, F)$, we can decide in time $\bigo(|\Delta| \cdot |\Reg|)$ whether $\nsst$ is garbage-free.
\end{proposition}
The second challenge is maintaining the set of outputs compactly for enumerating them. The main problem here is that the output of a run is not produced linearly as a single string (as is the case for classical spanners~\cite{AmarilliBMN21,AmarilliBMN19}) but instead in parallel on different registers that are combined in some order 
at the end of the input. To solve this, we follow the approach in~\cite{MunozR22,MunozR23} and present an extension of \emph{Enumerable Compact Sets} (ECS), a data structure that stores sets of strings compactly and retrieves them with output-linear delay. We modify ECS for storing sets of valuations and we call it \emph{ECS with assignments}. For storing valuations, we use assignments in the internal nodes of the data structure, which allows us to encode the runs' output and keep the content of different registers~synchronous.

Although we can eliminate the garbage from an NSST
(cf., Proposition~\ref{prop:garbageFree}),
the machine can still swap and move registers during a run, producing assignments that permute the registers' content but do not contribute to a final output. We call these  assignments  \emph{relabelings}. Thus, the last challenge is taking care of relabeling assignments.
Similarly as in~\cite{MunozR23} (i.e., with shifts), we treat them as special objects and
compact
them in the data structure whenever possible.
These extensions require
to revisit ECS and to provide
new ways for extending or
unifying
sets of valuations by also taking care of relabeling~assignments.
 	
\section{Composition of linear ET programs}\label{sec:composition}
	
A valuable property of ET programs is that
they receive
documents as input and produce documents as outputs.
In this line, it is natural to think on reusing the output of one ET program $(E_1, T_1)$ to feed a second ET program $(E_2, T_2)$, namely, to compose them. Formally, we define the \emph{composition} $(E_1 \cdot T_1) \circ (E_2 \cdot T_2)$ as a function $\sem{(E_1 \cdot T_1) \circ (E_2 \cdot T_2)}$ that maps documents to a bag of documents such that for every document $\doc$:
\[
\sem{(E_1 \cdot T_1) \circ (E_2 \cdot T_2)} (\doc) \ = \ \bigcup_{\doc' \in \sem{E_1 \cdot T_1}(\doc)} \sem{E_2 \cdot T_2}(\doc').
\]
Note that we use here the union of bags, which is defined in the natural way (a formal definition is given in the Appendix, Section~\ref{sec:defnAppendix}).

For extraction and transformation of information, it is useful to evaluate the composition efficiently. One naive approach
is to evaluate $\sem{E_1 \cdot T_1}(\doc)$
(e.g., by using the evaluation algorithm of Section~\ref{sec:enumeration} in case that $(E_1,T_1)$ is a \emph{linear} ET program),
and for every output $\doc'$ in $\sem{E_1 \cdot T_1}(\doc)$ compute $\sem{E_2 \cdot T_2}(\doc')$, gathering all the outputs together.
Of course, this could be time consuming, since $|\sem{E_1 \cdot T_1}(\doc)|$ could be exponential in $|\doc|$ in the worst case.

Towards solving the previous algorithmic problem, in the next result we show that every composition of NSSTs can be defined by an NSST. Formally, let us denote the input and output alphabet of some NSST $\nsst$ by $\Sigma(\nsst)$ and $\Omega(\nsst)$, respectively.  Given two NSSTs $\nsst_1$ and $\nsst_2$ such that $\Omega(\nsst_1) \subseteq \Sigma(\nsst_2)$, we define the composition $\nsst_1 \circ \nsst_2$ as the function from documents to bags of documents: $\sem{\nsst_1 \circ \nsst_2}(\doc) = \ \bigcup_{\doc' \in \sem{\nsst_1}(\doc)} \sem{\nsst_2}(\doc')$.
We obtain the following result.
\begin{theorem}\label{theo:nsstcomposition}
For every pair of NSSTs $\nsst_1$ and $\nsst_2$ such that $\Omega(\nsst_1) \subseteq \Sigma(\nsst_2)$, there exists an NSST $\nsst$ such that $\sem{\nsst} = \sem{\nsst_1 \circ \nsst_2}$. The construction of $\nsst$ is effective, and $|\nsst| = \bigo(2^{\mathsf{poly(|\nsst_1|, |\nsst_2|)}})$.
\end{theorem}
The statement of Theorem~\ref{theo:nsstcomposition} for set semantics rather than bag semantics was obtained
by~\cite{AlurD11,AlurD22}.
The novelty of
Theorem~\ref{theo:nsstcomposition} is that we extend
the result to bag semantics; namely, we need to maintain the multiplicities of the final outputs correctly. The proof revisits (and simplifies) the construction in~\cite{AlurD11,AlurD22}: we simulate $\nsst_1$ over $\doc$ while $\nsst_2$ runs over the
registers'
content of $\nsst_1$, compactly representing subruns of $\nsst_2$ by using
pairs
of states and assignment summaries~\cite{AlurC10}. The extension requires two changes for working under bag semantics. First, similar to the evaluation of NSSTs, garbage on $\nsst_1$-registers could generate additional runs for $\nsst_2$, producing outputs with wrong multiplicities.
Therefore, our
first step is to remove this garbage
by
constructing equivalent garbage-free NSSTs by using Propositon~\ref{prop:garbageFree}.
Second, the construction in~\cite{AlurD11} guesses subruns of~$\nsst_2$ for each register content and each starting state. In the end, we will use one of these guesses, and then unused subruns could modify the output multiplicity. Therefore, we simplify the construction in~\cite{AlurD11} by guessing a single subrun for each register content, using the non-determinism to discard wrong guesses. Interestingly, this simplification overcomes the error pointed out by Joost Engelfriet in the construction of~\cite{AlurD11}, which was solved recently in~\cite{AlurD22}. Our construction does not need the machinery of~\cite{AlurD22}; thus, an adaptation of our construction for NSSTs (with set semantics) can be considered as a simplified version of the proof in~\cite{AlurD11,AlurD22}. 
We conclude by mentioning that as an alternative proof of Theorem~\ref{theo:nsstcomposition} we could move to MSO transductions~\cite{EngelfrietH01}, extend the logic with a bag semantics, and then do the composition at a logical level; however, this strategy would lead to a non-elementary blow-up.

By combining Theorem~\ref{ETtoNSST},~\ref{theo:nsstcomposition} and~\ref{NSSTtoET}, we get the following corollary regarding the expressiveness of linear ET programs and their composition.
\begin{corollary}
  For every pair of linear ET programs $(E_1, T_1)$ and $(E_2, T_2)$, there exists
  a linear
  ET program $(E,T)$ such that $\sem{E \cdot T} \ \ = \ \ \sem{(E_1 \cdot T_1) \circ (E_2 \cdot T_2)}$.
\end{corollary}
We conclude 
this section by 
the following corollary regarding the evaluation of
the
composition of
linear
ET programs, obtained by combining Theorem~\ref{theo:nsstcomposition} and~\ref{theo:enum}.
\begin{corollary}\label{cor:compAndEnum}
	Given linear ET programs $(E_1, T_1), \ldots, (E_k, T_k)$ we can evaluate the composition:
	\[
	\sem{(E_1 \cdot T_1) \circ \ldots \circ (E_k \cdot T_k)}(\doc)
	\] 
	with linear time preprocessing and output-linear delay in data complexity.
\end{corollary} 

\section{Future Work}\label{sec:conclusions}
	
There are two obvious further research questions directly motivated by our results. Firstly, the framework of multispanners deserves further investigation. In particular, we wish to find out which results and properties of classical spanners directly carry over to multispanners, and in which regards these two models differ. Secondly, while we provide a sound formalisation of the class of polyregular ET programs (and have argued that it covers relevant extract and transform tasks), our
technical
results focus on the subclass of linear ET programs.
Our future work will focus on proving similar results for the full class of polyregular ET programs. 
	
\bibliographystyle{plain}
\bibliography{biblio}

\begin{thebibliography}{10}

\bibitem{aho1974design}
Alfred~V Aho and John~E Hopcroft.
\newblock {\em The design and analysis of computer algorithms}.
\newblock Pearson Education India, 1974.

\bibitem{dagstuhlreport}
Rajeev Alur, Miko{\l}aj Boja\'{n}czyk, Emmanuel Filiot, Anca Muscholl, and
  Sarah Winter.
\newblock {Regular Transformations (Dagstuhl Seminar 23202)}.
\newblock {\em Dagstuhl Reports}, 13(5):96--113, 2023.

\bibitem{AlurC10}
Rajeev Alur and Pavol Cern{\'{y}}.
\newblock Expressiveness of streaming string transducers.
\newblock In {\em FSTTCS}, pages 1--12, 2010.

\bibitem{AlurD11}
Rajeev Alur and Jyotirmoy~V. Deshmukh.
\newblock Nondeterministic streaming string transducers.
\newblock In {\em ICALP}, volume 6756, pages 1--20, 2011.

\bibitem{AlurD22}
Rajeev Alur, Taylor Dohmen, and Ashutosh Trivedi.
\newblock Composing copyless streaming string transducers.
\newblock {\em CoRR}, abs/2209.05448, 2022.

\bibitem{AmarilliBMN19}
Antoine Amarilli, Pierre Bourhis, Stefan Mengel, and Matthias Niewerth.
\newblock Enumeration on trees with tractable combined complexity and efficient
  updates.
\newblock In {\em PODS}, pages 89--103. {ACM}, 2019.

\bibitem{DBLP:journals/sigmod/AmarilliBMN20}
Antoine Amarilli, Pierre Bourhis, Stefan Mengel, and Matthias Niewerth.
\newblock Constant-delay enumeration for nondeterministic document spanners.
\newblock {\em {SIGMOD} Rec.}, 49(1):25--32, 2020.

\bibitem{amarilli2021constant}
Antoine Amarilli, Pierre Bourhis, Stefan Mengel, and Matthias Niewerth.
\newblock Constant-delay enumeration for nondeterministic document spanners.
\newblock {\em ACM Transactions on Database Systems (TODS)}, 46(1):1--30, 2021.

\bibitem{AmarilliBMN21}
Antoine Amarilli, Pierre Bourhis, Stefan Mengel, and Matthias Niewerth.
\newblock Constant-delay enumeration for nondeterministic document spanners.
\newblock {\em {ACM} Trans. Database Syst.}, 46(1):2:1--2:30, 2021.

\bibitem{berstel2013transductions}
Jean Berstel.
\newblock {\em Transductions and context-free languages}.
\newblock Springer-Verlag, 2013.

\bibitem{Bojanczyk2018}
Mikolaj Bojanczyk.
\newblock Polyregular functions.
\newblock {\em CoRR}, abs/1810.08760, 2018.

\bibitem{Bojanczyk22}
Mikolaj Bojanczyk.
\newblock Transducers of polynomial growth.
\newblock In {\em LICS}, pages 1:1--1:27. {ACM}, 2022.

\bibitem{BourhisEtAl2021}
Pierre Bourhis, Alejandro Grez, Louis Jachiet, and Cristian Riveros.
\newblock Ranked enumeration of {MSO} logic on words.
\newblock In {\em 24th International Conference on Database Theory, {ICDT}
  2021, March 23-26, 2021, Nicosia, Cyprus}, pages 20:1--20:19, 2021.

\bibitem{DoleschalKM23}
Johannes Doleschal, Benny Kimelfeld, and Wim Martens.
\newblock The complexity of aggregates over extractions by regular expressions.
\newblock {\em Logical Methods in Computer Science}, 19(3), 2023.

\bibitem{DoleschalEtAl2019}
Johannes Doleschal, Benny Kimelfeld, Wim Martens, Yoav Nahshon, and Frank
  Neven.
\newblock Split-correctness in information extraction.
\newblock In {\em Proceedings of the 38th {ACM} Symposium on Principles of
  Database Systems, {PODS} 2019, Amsterdam, The Netherlands, June 30 - July 5,
  2019}, pages 149--163, 2019.

\bibitem{doleschal2019split}
Johannes Doleschal, Benny Kimelfeld, Wim Martens, Yoav Nahshon, and Frank
  Neven.
\newblock Split-correctness in information extraction.
\newblock In {\em Proceedings of the 38th ACM SIGMOD-SIGACT-SIGAI Symposium on
  Principles of Database Systems}, pages 149--163, 2019.

\bibitem{doleschal2022weight}
Johannes Doleschal, Benny Kimelfeld, Wim Martens, and Liat Peterfreund.
\newblock Weight annotation in information extraction.
\newblock {\em Logical Methods in Computer Science}, 18, 2022.

\bibitem{DriscollSST86}
James~R. Driscoll, Neil Sarnak, Daniel~Dominic Sleator, and Robert~Endre
  Tarjan.
\newblock Making data structures persistent.
\newblock In {\em STOC}, pages 109--121, 1986.

\bibitem{EngelfrietH01}
Joost Engelfriet and Hendrik~Jan Hoogeboom.
\newblock {MSO} definable string transductions and two-way finite-state
  transducers.
\newblock {\em {ACM} Trans. Comput. Log.}, 2(2):216--254, 2001.

\bibitem{engelfriet2002two}
Joost Engelfriet and Sebastian Maneth.
\newblock Two-way finite state transducers with nested pebbles.
\newblock In {\em MFCS}, pages 234--244. Springer, 2002.

\bibitem{FaginKRV13}
Ronald Fagin, Benny Kimelfeld, Frederick Reiss, and Stijn Vansummeren.
\newblock Spanners: a formal framework for information extraction.
\newblock In {\em PODS}, pages 37--48. {ACM}, 2013.

\bibitem{fagin2015document}
Ronald Fagin, Benny Kimelfeld, Frederick Reiss, and Stijn Vansummeren.
\newblock Document spanners: A formal approach to information extraction.
\newblock {\em Journal of the ACM (JACM)}, 62(2):1--51, 2015.

\bibitem{florenzano2020efficient}
Fernando Florenzano, Cristian Riveros, Mart{\'\i}n Ugarte, Stijn Vansummeren,
  and Domagoj Vrgo{\v{c}}.
\newblock Efficient enumeration algorithms for regular document spanners.
\newblock {\em ACM Transactions on Database Systems (TODS)}, 45(1):1--42, 2020.

\bibitem{freydenberger2018document}
Dominik Freydenberger and Mario Holldack.
\newblock Document spanners: From expressive power to decision problems.
\newblock {\em Theory of Computing Systems}, 62:854--898, 2018.

\bibitem{Freydenberger2019}
Dominik~D. Freydenberger.
\newblock A logic for document spanners.
\newblock {\em Theory Comput. Syst.}, 63(7):1679--1754, 2019.

\bibitem{FreydenbergerEtAl2018}
Dominik~D. Freydenberger, Benny Kimelfeld, and Liat Peterfreund.
\newblock Joining extractions of regular expressions.
\newblock In {\em Proceedings of the 37th {ACM} {SIGMOD-SIGACT-SIGAI} Symposium
  on Principles of Database Systems, Houston, TX, USA, June 10-15, 2018}, pages
  137--149, 2018.

\bibitem{FreydenbergerThompson2020}
Dominik~D. Freydenberger and Sam~M. Thompson.
\newblock Dynamic complexity of document spanners.
\newblock In {\em 23rd International Conference on Database Theory, {ICDT}
  2020, March 30-April 2, 2020, Copenhagen, Denmark}, pages 11:1--11:21, 2020.

\bibitem{FreydenbergerThompson2022}
Dominik~D. Freydenberger and Sam~M. Thompson.
\newblock Splitting spanner atoms: {A} tool for acyclic core spanners.
\newblock In {\em Proc. {ICDT} 2022}, pages 10:1--10:18, 2022.

\bibitem{FASTUS}
Jerry~R. Hobbs, Douglas~E. Appelt, John Bear, David~J. Israel, Megumi Kameyama,
  Mark~E. Stickel, and Mabry Tyson.
\newblock {FASTUS:} {A} cascaded finite-state transducer for extracting
  information from natural-language text.
\newblock {\em CoRR}, cmp-lg/9705013, 1997.

\bibitem{karttunen1997replace}
Lauri Karttunen.
\newblock The replace operator.
\newblock {\em Finite-State Language Processing}, pages 117--147, 1997.

\bibitem{maturana2018document}
Francisco Maturana, Cristian Riveros, and Domagoj Vrgoc.
\newblock Document spanners for extracting incomplete information:
  Expressiveness and complexity.
\newblock In {\em PODS}, pages 125--136, 2018.

\bibitem{MunozR22}
Martin Mu{\~{n}}oz and Cristian Riveros.
\newblock Streaming enumeration on nested documents.
\newblock In {\em ICDT}, volume 220 of {\em LIPIcs}, pages 19:1--19:18, 2022.

\bibitem{MunozR23}
Martin Mu{\~{n}}oz and Cristian Riveros.
\newblock Constant-delay enumeration for slp-compressed documents.
\newblock In {\em ICDT}, volume 255, pages 7:1--7:17, 2023.

\bibitem{MuschollP19}
Anca Muscholl and Gabriele Puppis.
\newblock The many facets of string transducers (invited talk).
\newblock In Rolf Niedermeier and Christophe Paul, editors, {\em STACS}, volume
  126 of {\em LIPIcs}, pages 2:1--2:21, 2019.

\bibitem{Peterfreund21}
Liat Peterfreund.
\newblock Grammars for document spanners.
\newblock In Ke~Yi and Zhewei Wei, editors, {\em ICDT}, volume 186 of {\em
  LIPIcs}, pages 7:1--7:18, 2021.

\bibitem{PeterfreundEtAl2019}
Liat Peterfreund, Dominik~D. Freydenberger, Benny Kimelfeld, and Markus
  Kr{\"{o}}ll.
\newblock Complexity bounds for relational algebra over document spanners.
\newblock In {\em Proceedings of the 38th {ACM} {SIGMOD-SIGACT-SIGAI} Symposium
  on Principles of Database Systems, {PODS} 2019, Amsterdam, The Netherlands,
  June 30 - July 5, 2019.}, pages 320--334, 2019.

\bibitem{PeterfreundCFK19}
Liat Peterfreund, Balder ten Cate, Ronald Fagin, and Benny Kimelfeld.
\newblock Recursive programs for document spanners.
\newblock In {\em 22nd International Conference on Database Theory, {ICDT}
  2019, March 26-28, 2019, Lisbon, Portugal}, pages 13:1--13:18, 2019.

\bibitem{RiverosJV23}
Cristian Riveros, Nicol{\'{a}}s Van~Sint Jan, and Domagoj Vrgoc.
\newblock Rematch: a novel regex engine for finding all matches.
\newblock {\em VLDB}, 16(11):2792--2804, 2023.

\bibitem{SchmidSchweikardt_ICDT21}
Markus~L. Schmid and Nicole Schweikardt.
\newblock A purely regular approach to non-regular core spanners.
\newblock In Ke~Yi and Zhewei Wei, editors, {\em 24th International Conference
  on Database Theory, {ICDT} 2021, March 23-26, 2021, Nicosia, Cyprus}, volume
  186 of {\em LIPIcs}, pages 4:1--4:19. Schloss Dagstuhl - Leibniz-Zentrum
  f{\"{u}}r Informatik, 2021.

\bibitem{schmid2021spanner}
Markus~L. Schmid and Nicole Schweikardt.
\newblock Spanner evaluation over slp-compressed documents.
\newblock In {\em PODS'21: Proceedings of the 40th {ACM} {SIGMOD-SIGACT-SIGAI}
  Symposium on Principles of Database Systems, Virtual Event, China, June
  20-25, 2021}, pages 153--165, 2021.

\bibitem{SchmidSchweikardt_Gems}
Markus~L. Schmid and Nicole Schweikardt.
\newblock Document spanners - {A} brief overview of concepts, results, and
  recent developments.
\newblock In Leonid Libkin and Pablo Barcel{\'{o}}, editors, {\em {PODS} '22:
  International Conference on Management of Data, Philadelphia, PA, USA, June
  12 - 17, 2022}, pages 139--150. {ACM}, 2022.

\bibitem{SchmidSchweikardtPODS2022}
Markus~L. Schmid and Nicole Schweikardt.
\newblock Query evaluation over slp-represented document databases with complex
  document editing.
\newblock In {\em {PODS} '22: International Conference on Management of Data,
  Philadelphia, PA, USA, June 12 - 17, 2022}, pages 79--89, 2022.

\bibitem{Segoufin13}
Luc Segoufin.
\newblock Enumerating with constant delay the answers to a query.
\newblock In {\em ICDT}, pages 10--20, 2013.

\bibitem{substitution1}
https://www.tutorialsteacher.com/regex/substitution, 2024.
\newblock Accessed on 2024-03-15.

\bibitem{substitution2}
https://learn.microsoft.com/en-us/dotnet/standard/base-types/substitutions-in-regular-expressions,
  2024.
\newblock Accessed on 2024-03-15.

\bibitem{vassiliadis2009survey}
Panos Vassiliadis.
\newblock A survey of extract--transform--load technology.
\newblock {\em International Journal of Data Warehousing and Mining (IJDWM)},
  5(3):1--27, 2009.

\end{thebibliography}

\newpage
\appendix

\section*{APPENDIX}

\section{Definitions for bags}\label{sec:defnAppendix}

Let us provide some formal details on bag semantics that will be useful in our formal proofs.

A \emph{bag} $B$ is a surjective function $B: I \rightarrow U$ where $I$ is a finite set of identifiers (or ids) and $U$ is the underlying set of the bag. Given any bag $B$, we refer to these components as $I(B)$ and $U(B)$, respectively. For example, a bag $B = \multiset{a, a, b}$ (where $a$ is repeated twice) can be represented with a surjective function $B_0 =\{0 \mapsto a, 1 \mapsto a, 2 \mapsto b\}$ where $I(B_0) = \{0, 1, 2\}$ and $U(B_0) = \{a,b\}$. In general, we will use the standard notation for bags $\multiset{a_0, \ldots, a_{n-1}}$ to denote the bag $B$ whose identifiers are $I(B) = \{0, \ldots, n{-}1\}$ and $B(i) = a_i$ for each $i \in I(B)$. 
When we refer to the \emph{elements of $B$} we mean $I(B)$, namely, considering every object of $B$ by its identity (i.e., its~identifier).

For every $a$, we define the multiplicity of $a$ in $B$ as $\mult{B}{a} = |\{i \mid B(i) = a\}|$, namely, the number of occurrences of $a$ in $B$. Furthermore, we define $|B| = |I(B)|$. We say that a bag $B$ is \emph{contained} in the bag $B'$, and write $B \subseteq B'$, if $\mult{B}{a} \leq \mult{B'}{a}$ for every~$a$. We also say that two bags $B$ and $B'$ are \emph{equal}, and write $B = B'$, if $B \subseteq B'$ and $B' \subseteq B$.
Given a set $A$, we say that \emph{$B$ is a bag of $A$} if $U(B) \subseteq A$.
Note that if $B: I \rightarrow U$ is injective, then $B$ encodes a standard set (i.e., no repetitions). Then we write $a \in B$ iff $B(i) = a$ for some $i \in I(B)$ and define the empty bag $\emptyset$ as the bag such that~$I(\emptyset) = \emptyset$.

We define the \emph{union} $B_1 \cup B_2$ of bags $B_1$ and $B_2$ such that $I(B) = I(B_1) \uplus I(B_2)$, where $\uplus$ is the disjoint union, and $B(i) = B_1(i)$ if $i\in I(B_1)$ and $B(i)= B_2(i)$, otherwise.
Note that the union is associative, up to bag equivalence.
For a function $f$ from a finite set $A$ to bags of $A$, and a bag $B$ of $A$, we write $\bigcup_{b \in B} f(b)$ to denote the bag representing $\bigcup_{i \in I(B)} f(B(i))$, namely, the \emph{generalized union} where each element of $B$ contributes to the union.

\section{Formalisms for Polyregular Functions}\label{sec:polyregAppendix}
	
As discussed in~\cite{Bojanczyk22, Bojanczyk2018}, there are many equivalent formalisms for polyregular functions. We discuss two classes of representations: pebble transducers and for-transducers. As an example polyregular function, let us consider the function $T_2$ of the example ET program $(E_2, T_2)$ discussed in the introduction, i.\,e., the function 
\begin{align*}
T_2:\ \ \  & u \: \# \open{\varsx} u_{\varsx} \close{\varsx} ; \open{\varsy} u_{\varsy} \close{\varsy} ; \open{\varsz} v_{\varsz, 1} \close{\varsz}; \open{\varsz} v_{\varsz, 2} \close{\varsz} \ldots \open{\varsz} v_{\varsz, k} \close{\varsz} \: u' \:\: \mapsto  \\
&  u_{\varsy} \Vtextvisiblespace u_{\varsx}:v_{\varsz, 1} \# u_{\varsy} \Vtextvisiblespace u_{\varsx}:v_{\varsz, 2} \# \ldots \# u_{\varsy} \Vtextvisiblespace u_{\varsx}:v_{\varsz, k} \#\,,
\end{align*}
where $u, u' \in \Sigma^*$, $u_{\varsx}, u_{\varsy}, v_{\varsz, 1}, v_{\varsz, 2}, \ldots, v_{\varsz, k} \in \widehat{\Sigma}^*$.

A for-transducer is a program of nested for-loops. Its input is some string $w$. The loops are of the form ``for $X$ in $k..l$ do'', where $X$ is a variable that iterates through $k$ to $l$. The lower bound $k$ is either $1$ ($|w|$) or the iterating variable of an earlier for-loop, and the upper bound $l$ is either $|w|$ ($1$, respectively) or the iterating variable of an earlier for-loop. 
Note that the loops can iterate either forwardly (i.e., from $1$ to $|w|$), or backwardly (i.e., from $|w|$ to $1$). 
In addition, we can use Boolean variables, if-conditionals, comparisons ``$w[X] = b$'', where $X$ is a for-loop variable and $b$ a symbol from the input alphabet, comparisons ``$X \leq Y$'' of for-loop variables $X$ and $Y$, and an instruction ``output $b$'' , which produces the output symbol $b$ as the next output symbol. 

The Algorithm~\ref{alg:fortransducer} represents a for-loop transducer for the string-to-string function $T_2$ from above (note that we use ``$w[X] \in \widehat{\Sigma}$'' as syntactic~sugar).

To see that this for-transducer is correct, observe that the outer for-loop outputs in Line $17$ all the input symbols from $\widehat{\Sigma}$ in their order given by the input string, but it starts with the first symbol of $v_{\varsz, 1}$, due to the Boolean variable $B_1$. Consequently, the outer loop produces the factors $v_{\varsz, j}$ for $j = 1, 2, \ldots, k$. However, every time a new factor $v_{\varsz, j}$ starts (indicated by the symbol $\open{\varsz}$), we output $\#$ and then use the two inner loops to first produce `$u_{\varsy} \Vtextvisiblespace$' (first inner loop) followed by `$u_{\varsx}$' (second inner loop), followed by the symbol `$:$'.

A \emph{$k$-pebble transducer} is a finite automaton with a two-way input head and $k$ pebbles that can be placed on input positions. However, the pebbles are organised in a stack: We can place pebble $1$ on the current position $i$ of the input head, which means that we push $i$ onto a stack, later in the computation, we can place another pebble on the current position $j$ of the input head, which means that we push $j$ onto the stack, and so on. Removing the pebbles can only be done by popping the positions from the stack, i.\,e., we can only remove the pebble that has last been placed. 

Depending on the current state, the current symbol, the symbols on the pebble positions, and the relative order of the pebble positions, the pebble transducer updates it states, moves the input head to the left or to the right or leaves it unchanged, places a pebble on the current head position (and pushes its position onto the stack) or pops the topmost pebble position from the stack, and outputs an output symbol (or nothing). Pebble transducers are equivalent to for-transducers and therefore describe the set of polyregular functions~\cite{engelfriet2002two,Bojanczyk22}.

A pebble transducer with $1$ pebble can realise the string-to-string function $T_2$ from above as follows. We find the first position of $\open{\varsz}$ and place a pebble on this position. Then we move back with the input head until we see $\open{\varsy}$ and we output $u_{\varsy} \Vtextvisiblespace$. Then we move back with the input head until we see $\open{\varsx}$ and we output $u_{\varsx}$. Then we move to the position of the pebble, we pop the pebble from the stack and we output $v_{\varsz, 1}$. The input head is now on the first occurrence of $\close{\varsz}$. We now check if there is a further occurrence of $\open{\varsz}$ to the right and, if not, we accept, and if yes, we place the pebble on this occurrence of $\open{\varsz}$. The pebble is now on the position of the second occurrence of $\open{\varsz}$. We can now repeat this procedure until there is no further occurrence of $\open{\varsz}$.

\begin{algorithm}[t]
	\caption{A for-transducer for the polyregular string-to-string function $T_2$}\label{alg:fortransducer}
	\begin{algorithmic}[1]
		\Procedure{for-transducer for $T_2$}{$w$}
		\State{$B_1 := \mathsf{false}$}
		\For{$Z$ \textbf{in} $1..|w|$}
		\If{$w[Z] = \open{\varsz}$}
		\State{$B_1 := \mathsf{true}$}
		\State{$B_2 := \mathsf{false}$}
		\For{$Y$ \textbf{in} $1..|w|$}
		\State{\textbf{if} $w[Y] = \open{\varsy}$ \textbf{then} $B_2 := \mathsf{true}$}
		\State{\textbf{if} $B_2$ and $w[Y] \in \widehat{\Sigma}$ \textbf{then} output $w[Y]$}
		\State{\textbf{if} $B_2$ and $w[Y] = \close{\varsy}$ \textbf{then} $B_2 := \mathsf{false}$}
		\EndFor
		\State{output `$\Vtextvisiblespace$'}
		\For{$X$ \textbf{in} $1..|w|$}
		\State{\textbf{if} $w[X] = \open{\varsx}$ \textbf{then} $B_2 := \mathsf{true}$}
		\State{\textbf{if} $B_2$ and $w[X] \in \widehat{\Sigma}$ \textbf{then} output $w[X]$}
		\State{\textbf{if} $B_2$ and $w[X] = \close{\varsx}$ \textbf{then} $B_2 := \mathsf{false}$}
		\EndFor
		\State{output `$:$'}
		\EndIf
		\State{\textbf{if} $B_1$ and $w[Z] \in \widehat{\Sigma}$ \textbf{then} output $w[Z]$}
		\State{\textbf{if} $w[Z] = \close{\varsz}$ \textbf{then} output `$\#$'}
		\EndFor	
		\EndProcedure
	\end{algorithmic}
\end{algorithm}

\section{Proofs of Section~\ref{sec:properties}}\label{sec:equiAppendix}

\subsection{Proof of Theorem~\ref{ETtoNSST}}

We first need the following definitions.

We extend the encoding $\nrefw{\cdot}$ to multiref-words and multiref-languages in a natural way. 
Let $\varset = \{\varsx_1, \varsx_2, \ldots, \varsx_m\}$ with $\varsx_1 \preceq \varsx_2 \preceq \ldots \preceq \varsx_m$. A string $u \in \Gamma_{\varset}^*$ is called a \emph{$\Gamma_{\varset}$-sequence} if, for every $\ell \in [m]$, the subsequence of the occurrences of $\open{\varsx_{\ell}}$ and of $\close{\varsx_{\ell}}$ of $u$ has the form $(\close{\varsx_{\ell}})^{o_{\ell}} (\open{\varsx_{\ell}} \close{\varsx_{\ell}})^{e_{\ell}} (\open{\varsx_{\ell}})^{c_{\ell}}$ for some $o_{\ell}, c_{\ell} \in \{0, 1\}$ and $e_{\ell} \geq 0$ (i.\,e., $u$ is the factor of a valid multiref-word). Moreover, for every, $\ell \in [m]$, we call $(\close{\varsx_{\ell}})^{o_{\ell}} (\open{\varsx_{\ell}} \close{\varsx_{\ell}})^{e_{\ell}} (\open{\varsx_{\ell}})^{c_{\ell}}$ the \emph{$\varsx_\ell$-sequence} of $u$. For every $\Gamma_{\varset}$-sequence $u$ with $\varsx_\ell$-sequence $u_\ell = (\close{\varsx_{\ell}})^{o_{\ell}} (\open{\varsx_{\ell}} \close{\varsx_{\ell}})^{e_{\ell}} (\open{\varsx_{\ell}})^{c_{\ell}}$ for every $\ell \in [m]$, we define $\nrefw{u_\ell} = (\close{\varsx_{\ell}})^{o_{\ell}} (\open{\varsx_{\ell}} \close{\varsx_{\ell}})^{e'_{\ell}} (\open{\varsx_{\ell}})^{c_{\ell}}$ where $e'_\ell = 0$ if $e_\ell = 0$ and $e'_\ell = 1$ otherwise. And we define $\nrefw{u} = \nrefw{u_1} \nrefw{u_2} \ldots \nrefw{u_m}$. Any $\Gamma_{\varset}^*$-sequence $u$ is called \emph{normalised}, if $u = \nrefw{u}$. For any multiref-word $w = u_1 b_1 u_2 b_2 \ldots u_n b_n u_{n+1}$ with $b_i \in \Sigma$ for every $i \in [n]$ and $u_j \in \Gamma_{\varset}^*$ for every $j \in [n + 1]$, we set $\nrefw{w} = \nrefw{u_1} b_1 \nrefw{u_2} b_2 \ldots \nrefw{u_n} b_n \nrefw{u_{n+1}}$. It can be easily verified that, for every multi-refword $w$, we have that $\nrefw{w} = \nrefw{\getSpanTuple{w}, \getDoc{w}}$. Finally, for any multiref-language $L$, we set $\nrefw{L} = \{\nrefw{w} \mid w \in L\}$. 

The next lemma will be crucial for the proof.

\begin{lemma}\label{constructNormDFALemma}
Given an NFA $M$ with $n$ states such that $\lang{M}$ is a multiref-language over $\Sigma$ and $\varset$ with $|\varset| = m$, and every word of $\lang{M}$ ends with a symbol of $\Sigma$. Then we can construct a DFA $N$ with $\bigo(2^{2n+3m} |\Sigma| m)$ states and $\lang{N} = \nrefw{\lang{M}}$ in time $\bigo(2^{2n + 3m} |\Sigma| m)$.
\end{lemma}

\begin{proof}
We first transform $M$ into an equivalent DFA $M_1$ with $\bigo(2^n)$ states by standard constructions. Without loss of generality, we assume that every state is reachable from the initial state, and every state can reach a final state. In particular, this means that all strings over $\Gamma_{\varset}^*$ that can be read between two states must be $\Gamma_{\varset}$-sequences.

Next, we transform $M_1$ into a DFA $M_2$ that simulates $M_1$, but whenever it has read some input $u \in \Gamma_{\varset}^*$ or some input with a suffix $b u \in \Sigma \Gamma_{\varset}^* $, then the current state stores $\nrefw{u}$, which we call the current $\Gamma_{\varset}$-sequence. This can be done by simply updating the current $\Gamma_{\varset}$-sequence after every transition of a run. More precisely, let us assume that the current $\Gamma_{\varset}$-sequence is $u_1 u_2 \ldots u_m$ such that, for every $\ell \in [m]$, $u_\ell = (\close{\varsx_\ell})^{o_\ell} (\open{\varsx_\ell} \close{\varsx_\ell})^{e_\ell} (\open{\varsx_\ell})^{c_\ell}$ for some $o_\ell, e_\ell, c_\ell \in \{0, 1\}$. If we read some symbol $b \in \Sigma$, then we update the current $\Gamma_{\varset}$-sequence to $\eword$. If we read a symbol $\gamma \in \{\open{\varsx_{\ell}}, \close{\varsx_{\ell}}\}$ for some $\ell \in [m]$, then we update the current $\Gamma_{\varset}$-sequence to $u'_1 u'_2 \ldots u'_m$, where $u'_j = u_j$ for every $j \in [m] \setminus \{\ell\}$, and $u'_{\ell} = \nrefw{u_\ell \gamma}$. Formally, we interpret the states of $M_2$ as pairs $(p, u)$, where $p$ is a state of $M_1$ and $u$ is some normalised $\Gamma_{\varset}$-sequence. Moreover, $(q_0, \eword)$ is the initial state, where $q_0$ is $M_1$'s initial state, and every state $(q_f, u)$ is a final state if $q_f$ is final in $M_1$. We observe that $M_2$ has $\bigo(2^n 8^{m}) = \bigo(2^{n + 3m})$ states, and $M_2$ can be constructed in time $\bigo(2^{n + 3m} (|\Sigma| + m))$. Also note that $|M_2| = \bigo(2^{n + 3m} (|\Sigma| + m))$. Obviously, $M_2$ is still deterministic and accepts the same language $\lang{M}$. 

Then, we construct a directed and arc-labelled graph $G$ whose nodes are the states from $M_1$. For each two states $(p_1, \eword)$ and $(p_2, u)$ of $M_2$, we compute whether in $M_2$ there is some path from $(p_1, \eword)$ to $(p_2, u)$ labelled with some sequence over $\Gamma_{\varset}$, or $p_1 = p_2$ and $u = \eword$. This can be done by removing all $\Sigma$-transitions from $M_2$ and then perform a breadth-first-search in each state $(p, \eword)$. Hence, time $\bigo(2^n |M_2|) = \bigo(2^{2n + 3m} (|\Sigma| + m))$ is sufficient. For all states $(p_1, \eword)$ and $(p_2, u)$ of $M_2$ such that the latter is $\Gamma_{\varset}$-reachable from the former, and every transition $((p_2, u), b, (p_3, \eword))$ of $M_2$ for some $b \in \Sigma$, we add to $G$ a length-$(|u b|)$ path of new nodes (we call these new nodes \emph{path states}) from $p_1$ to $p_3$ such that this path is labelled with $ub$ (note that this needs $|ub|-1 = \bigo(m)$ new path states). For this we have to consider each pair $(p_1, p_2)$ of states of $M_1$, each normalised $\Gamma_{\varset}$-sequence $u$, and every outgoing $\Sigma$-arc from $(p_2, u)$ in $M_2$ (which is deterministic), and for each of those we have to add at $\bigo(m)$ new path states. Hence, this can be done in time $\bigo(2^{2n} 2^{3m} |\Sigma| m) = \bigo(2^{2n + 3m} |\Sigma| m)$. Furthermore, the total number of nodes is also $\bigo(2^{2n+3m} |\Sigma| m)$.

We can now interpret $G$ as an NFA $M_3$, and it can be easily seen that this NFA accepts $\nrefw{\lang{M}}$. As observed above, $M_3$ has $\bigo(2^{2n+3m} |\Sigma| m)$ states.

It remains to turn $M_3$ into a DFA $N$. To this end, we first observe that it is possible that a non-path state $p$ has two different paths to some non-path state $q$ that are both labelled with the same string $u b$. This can happen, because in $M_2$ it is possible that we can go from $(p, \eword)$ to a state $(p', u)$, but with different permutations $u_1$ and $u_2$ of $u$, and there is a transition $((p', u), b, (q, \eword))$. We simply delete such duplicates. Now, for every non-path state $p$ that has an outgoing path labelled with $u b$ and a different outgoing path labelled with a different $u' b'$, neither $u b$ is a proper prefix of $u' b'$ nor $u' b'$ is a proper prefix of $u b$ (this is due to the fact that $u, u' \in \Gamma_{\varset}^*$ and $b, b' \in \Sigma$). Consequently, for every fixed non-path state, we can merge all its outgoing paths into a tree according to their common prefixes and know that these paths will all necessarily branch off before they reach their target state. The result is a determinised version of  the NFA, which we call $N$. This can be done in time $\bigo(2^{2n+3m} |\Sigma| m)$. We note that we still have $\lang{N} = \nrefw{\lang{M}}$ and $N$ has $\bigo(2^{2n+3m} |\Sigma| m)$ states.
\end{proof}

We are now ready to prove Theorem~\ref{ETtoNSST}, but we actually prove the following slightly stronger result (which allows $E$ to be a regular multispanner instead of only a regex multispanner). Note that the statement of Theorem~\ref{ETtoNSST} follows from the fact that any regular expression $r$ can be transformed into an equivalent NFA of size $\bigo(|r|)$.

\begin{theorem}
Given a regular multispanner $E$ over $\Sigma$ (represented by an NFA with $n$ states), and a linear polyregular string-to-string function $T$ with input alphabet $\Sigma \cup \Gamma_{\varset}$ (represented by a DSST with $h$ states), then we can construct an NSST $\sst$ with $\sem{\sst} = \sem{E \cdot T}$ in time $\bigo(2^{4n+9|\varset|} |\Sigma|^3 |\varset|^2 h^2)$.
\end{theorem}

\begin{proof}
Let us assume that $\varset = \{\varsx_1, \varsx_2, \ldots, \varsx_m\}$ with $m \geq 1$ and $\varsx_1 \preceq \varsx_2 \preceq \ldots \preceq \varsx_m$. Let:
$$
\widetilde{\sst} = (Q_{\widetilde{\sst}}, \Sigma, \Omega,  \Reg, \Delta_{\widetilde{\sst}}, I_{\widetilde{\sst}}, F_{\widetilde{\sst}})
$$ 
be the DSST that represents $T$. Let $N$ be the NFA that describes $E$. We assume that all multiref-words from $\lang{N}$ end with a symbol from $\Sigma$. This is without loss of generality: We extend $\Sigma$ by a dummy symbol $\#$ and change $N$ such that from every accepting state we can go with $\#$ to a new state, and we declare this new state to be the only accepting state. Then we change $\widetilde{\sst}$ such that it simply interprets every input $w \#$ as the input $w$. \par

By Lemma~\ref{constructNormDFALemma}, we can construct a DFA $M = (Q_{M}, \Sigma \cup \Gamma_{\varset}, q_{0, M}, \delta_M, F_M)$ with $\bigo(2^{2n+3m} |\Sigma| m)$ states in time $\bigo(2^{2n+3m} |\Sigma| m)$ that accepts $\nrefw{\lang{N}}$. Without loss of generality, we assume that $M$ is such that every state can be reached by the start state and every state can reach some accepting state (note that this means that $M$ may be incomplete).

Since $M$ accepts $\nrefw{\lang{N}}$, we know that if $M$ reads an input $u \in \Gamma_{\varset}^*$ or an input with a suffix $b u$ with $b \in \Sigma$ and $u \in \Gamma_{\varset}^*$, then $u$ is a normalised $\Gamma_{\varset}$-sequence.

We will now define an NSST $\sst = (Q_{\sst}, \Sigma, \Omega, \Reg, \Delta_{\sst}, I_{\sst}, F_{\sst})$ by combining $M$ and $\widetilde{\sst}$. 

Let us fix some notation for NSSTs with some input alphabet $A$ (and therefore also for DSSTs). Obviously, we can interpret an NSST as a directed graph whose arcs are labelled with pairs $(b, \sigma)$, where $b \in A$ and $\sigma$ is an assignment. A \emph{path} from a state $q$ to a state $q'$ in the NSST is then any sequence 
\begin{equation*}
p_0, (b_1, \sigma_1), p_1, (b_2, \sigma_2), p_2, \ldots, (b_k, \sigma_k), p_k\,, 
\end{equation*}
where $p_i$ is a state for every $i \in \{0\} \cup [k]$, $(b_j, \sigma_j)$ is the label of an arc from $p_{j-1}$ to $p_j$ for every $j \in [k]$, and $p_0 = q$ and $p_k = q'$. Furthermore, $b_1 b_2 \ldots b_k$ is called the \emph{label} of the path, and $\sigma_1 \circ \sigma_2 \circ \ldots \circ \sigma_k$ is called the \emph{composed assignment} of the path. We use analogous notions for NFAs and DFAs.

We define $Q_{\sst} = Q_M \times Q_{\widetilde{\sst}}$, we set $I_{\sst}((q_{0, M}, p)) = I_{\widetilde{\sst}}(p)$ for every $p \in Q_{\widetilde{\sst}}$ with $I_{\widetilde{\sst}}(p)$ defined (otherwise $I_{\sst}((q_{0, M}, p))$ is undefined), we set $F_{\sst}((q_f, p)) = F_{\widetilde{\sst}}(p)$ for every $q_f \in F_M$ and for every $p \in Q_{\widetilde{\sst}}$ with $F_{\widetilde{\sst}}(p)$ defined (otherwise $F_{\sst}((q_f, p))$ is undefined). It remains to define $\Delta_{\sst}$, which we shall do next.

For every $p, p' \in Q_M$, every $q, q' \in Q_{\widetilde{\sst}}$, every $b \in \Sigma$ and every normalised $\Gamma_{\varset}$-sequence $u \in (\Gamma_{\varset})^*$ such that
\begin{itemize}
\item in $M$ there is a path from $p$ to $p'$ with label $u b$, and
\item in $\widetilde{\sst}$ there is a path from $q$ to $q'$ with label $u b$ and some composed assignment $\sigma$, 
\end{itemize}
we add the transition $((p, q), b, \sigma, (p', q'))$ to $\Delta_{\sst}$. This concludes the definition of $\sst$. 

In order to effectively construct $\Delta_{\sst}$, we have to consider all states $p, p' \in Q_N$, $q, q' \in Q_{\widetilde{\sst}}$, all $b \in \Sigma$ and all normalised $\Gamma_{\varset}$-sequences $u \in (\Gamma_{\varset})^*$. For each such tuple $(p, p', q, q', b, u)$, the properties from above can be checked (and $\sigma$ computed) in time $\bigo(|u b|) = \bigo(|\varset|) = \bigo(m)$, since in each of $M$ and $\widetilde{\sst}$ there is at most one path to consider (recall that both $M$ and $\widetilde{\sst}$ are deterministic). Since there are $8^m$ different normalised $\Gamma_{\varset}$-sequences, we can construct $\Delta_{\sst}$ in time $\bigo(2^{4n+9m} |\Sigma|^3 m^2 |Q_{\widetilde{\sst}}|^2)$.

Let $\doc \in \Sigma^*$ be a document with $|\doc| = k$. We have to prove that $\sem{\sst}(\doc) = \sem{E \cdot T}(\doc)$. We first need the following definition. A \emph{condensed run} of $\widetilde{\sst}$ on some input $u_1 b_1 u_2 b_2 \ldots u_k b_k$ with $u_i \in \Gamma_{\varset}^*$ and $b_i \in \Sigma$ for every $i \in [k]$ is a sequence
\begin{equation*}
(q_0, \nu_0) \xrightarrow{u_1 b_1/\sigma_1} (q_1, \nu_1) \xrightarrow{u_2 b_2 /\sigma_2} \ldots \xrightarrow{u_k b_k/\sigma_{k}} (q_k, \nu_k)
\end{equation*}
where $q_0, q_1, \ldots, q_k$ are states with $I_{\widetilde{\sst}}(q_0)$ and $F_{\widetilde{\sst}}(q_k)$ defined, and, for every $i \in [k]$, there is a run
\begin{equation*}
(q'_{i, 0}, \nu'_{i, 0}) \xrightarrow{c_{i, 1}/\sigma'_{i, 1}} (q'_{i, 1}, \nu'_{i, 1}) \xrightarrow{c_{i, 2} /\sigma'_{i, 2}} \ldots \xrightarrow{c_{i, k_i}/\sigma'_{i, k_i}} (q'_{i, k_{i}}, \nu'_{i, k_{i}})
\end{equation*}
with $(q'_{i, 0}, \nu'_{i, 0}) = (q_{i-1}, \nu_{i-1})$, $(q'_{i, k_i}, \nu'_{i, k_i}) = (q_i, \nu_i)$, $c_{i, 1} c_{i, 2} \ldots c_{i, k_i} = u_i b_i$, and $\sigma'_{i, 1} \circ \sigma'_{i, 2} \circ \ldots \circ \sigma'_{i, k_i} = \sigma_i$.

We start with showing $\sem{\sst}(\doc) \subseteq \sem{E \cdot T}(\doc)$. Let $\doc' \in \sem{\sst}(\doc)$ and let 
\begin{equation*}
\rho \ := \ ((p_0, q_0), \nu_0) \xrightarrow{\doc[1]/\sigma_1} 
\ldots \xrightarrow{\doc[n]/\sigma_{n}} ((p_k, q_k), \nu_k)
\end{equation*}
be an accepting run of $\sst$ with $\doc' = \out(\rho)$. This also means that $p_0 = q_{M, 0}$, $p_k \in F_M$, $I_{\widetilde{\sst}}(q_0) \neq \bot$ and $F_{\widetilde{\sst}}(q_k) \neq \bot$. 

For every $i \in [m]$, let $u_i \in \Gamma_{\varset}$ be the normalised $\Gamma_{\varset}$-sequence that is responsible for the transition $((p_{i-1}, q_{i-1}), \doc[i], \sigma_i, (p_i, q_i))$, i.\,e., $u_i$ is such that in $M$ there is a path from $p_{i-1}$ to $p_i$ with label $u_i \doc[i]$, and in $\widetilde{\sst}$ there is a path from $q_{i-1}$ to $q_i$ with label $u_i \doc[i]$ and with composed assignment $\sigma_i$. Since $M$ and $\widetilde{\sst}$ are deterministic, for every $i \in [m]$ there is exactly one normalised $\Gamma_{\varset}$-sequence $u_i$ with these properties. Hence, the $u_1, u_2, \ldots, u_m$ are uniquely defined. 

Then, $w = u_1 \doc[1] u_2 \doc[2] \ldots u_k \doc[k]$ is accepted by $M$, and there is a condensed accepting~run
\begin{equation*}
\rho' \ := \ (q_0, \nu_0) \xrightarrow{u_1 \doc[1]/\sigma_1} \ldots \xrightarrow{u_n \doc[n]/\sigma_{n}} (q_k, \nu_n)
\end{equation*}
of $\widetilde{\sst}$ and $\out(\rho') = \out(\rho) = \doc'$. This implies that $\getSpanTuple{w} \in E(\doc)$ and $T(w) =  T(\nrefw{\doc, \getSpanTuple{w}}) = \doc'$; thus, $\doc' \in \sem{E \cdot T}$.

Moreover, each individual accepting run of $\sst$ on input $\doc$ with output $\doc'$ yields a distinct tuple $(u_1, u_2, \ldots, u_m)$ of normalised $\Gamma_{\varset}$-sequences, which means that there is a distinct $t \in E(\doc)$ with $T(\nrefw{\doc, t}) = \doc'$. This is due to the fact that two different accepting runs of $\sst$ on input $\doc$ with output $\doc'$ can only differ in their states, which, since $M$ and $\widetilde{\sst}$ are deterministic, means that for at leat one $i \in [k]$, the two transitions of the different runs that read symbol $\doc[i]$ must be due to different normalised $\Gamma_{\varset}$-sequences. 

Hence, for every $\doc' \in \Omega^*$, we have that $\mult{\sem{\sst}(\doc)}{\doc'} \leq \mult{\sem{E \cdot T}(\doc)}{\doc'}$. This means that $\sem{\sst}(\doc) \subseteq \sem{E \cdot T}(\doc)$.

Now let us prove $\sem{E \cdot T}(\doc) \subseteq \sem{\sst}(\doc)$. To this end, let $\doc' \in \sem{E \cdot T}(\doc)$, which means that there is a $t \in E(\doc)$ and $\doc' = T(\nrefw{\doc, t})$. Let $w = u_1 \doc[1] u_2 \doc[2] \ldots u_k \doc[k] = \nrefw{\doc, t}$. Since $M$ accepts $\nrefw{\lang{N}}$, we know that $w \in \lang{M}$. Let $p_0 = q_{0, M}$ and, for every $i \in [k]$, let $p_{i}$ be the state of $M$'s accepting run after reading $u_1 \doc[1] \ldots u_i \doc[i]$. Furthermore, there must be a condensed accepting run 
\begin{equation*}
\rho \ := \ (q_0, \nu_0) \xrightarrow{u_1 \doc[1]/\sigma_1} \ldots \xrightarrow{u_n \doc[n]/\sigma_{n}} (q_k, \nu_n)
\end{equation*}
of $\widetilde{\sst}$ and $\doc'$. Note that this means that $I_{\widetilde{\sst}}(q_0)$ and $F_{\widetilde{\sst}(q_k)}$ are defined, and that $\out(\rho) = \doc'$. 

By definition, this means that $\sst$ has a transition $((p_{i-1}, q_{i-1}), \doc[i], \sigma_i, (p_i, q_i))$ for every $i \in [k]$. Consequently, there is an accepting run
\begin{equation*}
\rho' \ := \ ((p_0, q_0), \nu_0) \xrightarrow{\doc[1]/\sigma_1} 
\ldots \xrightarrow{\doc[n]/\sigma_{n}} ((p_k, q_k), \nu_n)
\end{equation*}
of $\sst$. Furthermore, $\out(\rho') = \out(\rho) = \doc'$. This implies that $\doc' \in \sem{\sst}(\doc)$. 

Moreover, each distinct $t \in E(\doc)$ with $\doc' = T(\nrefw{\doc, t})$ yields a distinct multiref-word $w = u_1 \doc[1] u_2 \doc[2] \ldots u_k \doc[k] = \nrefw{\doc, t}$. Since $M$ and $\widetilde{\sst}$ are deterministic, this also yields a distinct accepting run of $M$ on $w$ and a distinct accepting run of $\widetilde{\sst}$ on $w$, and therefore a distinct accepting run of $\sst$. Hence, for every $\doc' \in \Omega^*$, we have that $\mult{\sem{E \cdot T}(\doc)}{\doc'} \leq \mult{\sem{\sst}(\doc)}{\doc'}$. This means that $\sem{E \cdot T}(\doc) \subseteq \sem{\sst}(\doc)$.

\end{proof}

\subsection{Proof of Theorem~\ref{NSSTtoET}}

We need the following definitions. Let $\nsst = (Q, \Sigma, \Omega,  \Reg, \Delta, I, F)$ be an NSST. When $\nsst$ is in a state $p \in Q$ and reads a symbol $b \in \Sigma$, then the number of possible nondeterministic choices it can make is the sum over all the multiplicities (recall that $\Delta$ is a bag) of elements $(p, b, \sigma, q) \in \Delta$. We shall call this number the \emph{nondeterministic branching factor of $p$ and $b$}. More formally, we define $\ndetbf{p, b} = \sum_{(p, b, \sigma, q) \in \Delta} \mult{\Delta}{(p, b, \sigma, q)}$. Furthermore, we define the \emph{nondeterministic branching factor of $\nsst$} as $\ndetbf{\nsst} = \max_{p \in Q, b \in \Sigma}\{\ndetbf{p, b}\}$.

Now, let us restate Theorem~\ref{NSSTtoET} and prove it.

\begin{theorem}
Given an NSST $\sst$ with $n$ states, input alphabet $\Sigma$ and nondeterministic branching factor $m$, then we can construct a regex multispanner $E$ over $\Sigma$ and variables $\{\varsx_1, \varsx_2, \ldots, \varsx_{\max\{n, m\}}\}$ and a linear polyregular function $T$ with $\sem{E \cdot T} = \sem{\sst}$. Moreover, $E$ is represented by a multispanner-expression $r$ with $|r| = \bigo(|\Sigma| + \max\{n, m\})$, $T$ is represented by an DSST $\sst'$ with $\bigo(\max\{n, m\} \cdot n)$ states, and both $r$ and $\sst'$ can be constructed in time $\bigo(|\sst| + n \cdot \max\{n, m\} + |\Sigma|)$.
\end{theorem}

\begin{proof}
Let $\sst \ = \ (Q_{\sst}, \Sigma, \Omega,  \Reg_{\sst}, \Delta_{\sst}, I_{\sst}, F_{\sst})$. For the sake of convenience, we represent the bag $\Delta_{\sst}$ as the set 
\begin{align*}
\widehat{\Delta}_{\sst} = &\{(p, b, \sigma, q, 1), (p, b, \sigma, q, 2), \ldots, \\
&(p, b, \sigma, q, \mult{\Delta}{(p, b, \sigma, q)}) \mid (p, b, \sigma, q) \in \Delta_{\sst}\}\,.
\end{align*} 
This is obviously an equivalent representation compared to the representation by a surjective function that has a set of identifiers as domain.

Let us assume that, for every $q \in Q_{\sst}$ and $b \in \Sigma$, we have a numbering of all $\ndetbf{p, b}$ transitions of the form $(p, b, \sigma, q, \textsf{id})$, i.\,e., every transition $(p, b, \sigma, q, \textsf{id}) \in \widehat{\Delta}_{\sst}$ has a unique number $\pi((p, b, \sigma, q, \textsf{id})) \in [\ndetbf{p, b}]$. Obviously, $\pi((p, b, \sigma, q, \textsf{id})) \leq \ndetbf{\nsst} = m$ for every transition $(p, b, \sigma, q, \textsf{id}) \in \widehat{\Delta}_{\sst}$. Such a numbering can be easily constructed from $\widehat{\Delta}_{\sst}$ in time $\bigo(|\widehat{\Delta}_{\sst}|)$.

We also extend the numbering $\pi$ to the states, i.\,e., every state $q \in Q$ has a unique number $\pi(q) \in [n]$. 

Let 
\begin{equation*}
\rho \ := \ (q_0, \nu_0) \xrightarrow{b_1/\sigma_1} (q_1, \nu_1) \xrightarrow{b_2 /\sigma_2} \ldots \xrightarrow{b_k/\sigma_{k}} (q_{k}, \nu_k)
\end{equation*}
be a run of $\sst$ on some input $b_1 b_2 \ldots b_k$. For every $i \in [k]$, $\sst$ reads symbol $b_i$ by a transition $s_i = (q_{i-1}, b_i, \sigma_i, q_i, \textsf{id}_i) \in \widehat{\Delta}_{\sst}$, which is one possible nondeterministic choice of the $\ndetbf{q_{i-1}, b_i}$ many possible transitions of the form $(q_{i-1}, b_i, \sigma, p, \textsf{id})$. Moreover, the choice of the start state $q_0$ with $I(q_0) \neq \bot$ is another nondeterministic choice. We can therefore uniquely represent the run $\rho$ of $\sst$ on input $b_1 b_2 \ldots b_k$ as a string from $[n] ([m] \times \Sigma)^*$, namely the string 
\begin{equation*}
\pi(q_0) (\pi(s_1), b_1) (\pi(s_2), b_2) \ldots (\pi(s_k), b_k)\,.
\end{equation*}
In the following, we call every string $h' (h_1, b_1) (h_2, b_2) \ldots (h_k, b_k)$ with $h' \in [n]$, $h_i \in [m]$ and $b_i \in \Sigma$ an \emph{nd-annotation of $b_1 b_2 \ldots b_k$}.

Let $w = h' (h_1, b_1) (h_2, b_2) \ldots (h_k, b_k)$ be an nd-annotation of $b_1 b_2 \ldots b_k$. We can now transform $w$ into a run $\rho_w$ of $\sst$ by the following procedure. We start in state $q_0$ with $\pi(q_0) = h'$ if $I_{\sst}(q_0)$ is defined, and interrupt otherwise. Then, for every $i = 1, 2, \ldots, k$, if we are currently in a state $q_{i-1}$, then we extend the current run with the transition $(q_{i-1}, b_i, \sigma_i, q_i, \textsf{id}_i)$ with $\pi((q_{i-1}, b_i, \sigma_i, q_i, \textsf{id}_i)) = h_i$ next, or interrupt if no such transition exist (which is the case if $h_i > \ndetbf{q_{i-1}, b_i}$). 

We say that the nd-annotation $w$ of some input $b_1 \ldots b_k$ is \emph{valid}, if it can be completely transformed into a run $\rho_w$ without interrupting the procedure from above, and it is called \emph{accepting}, if it is valid and $\rho_w$ is an accepting run. 

It can be easily verified that for every $\doc \in \Sigma^*$, 
\begin{equation*}
\sem{\sst}(\doc) = \multiset{\out(\rho_w) \mid \text{$\rho$ is an accepting nd-annotation of $\doc$}}\,.
\end{equation*}

Moreover, we can transform $\sst$ into a DSST $\sst'$ with input alphabet $[\max\{n, m\}] \cup \Sigma$ that, for every $\doc \in \Sigma^*$, maps each accepting nd-annotation $w$ of $\doc$ to $\out(\rho_w)$. More precisely, $\sst'$ simply carries out the run $\rho_w$ represented by the input $w$ as described above. When the input $w$ is not a valid nd-annotation, then $\sst'$ just stops and rejects. If $w$ is valid but not accepting, then it rejects as well. If $\rho_w$ is accepting, then $\sst'$ accepts and therefore outputs the element $\out(\rho_w)$. In particular, note that if $\sst$ has several different accepting runs on $\doc$ with the same output $\doc'$, then $\sst'$ will also produce $\doc'$ for each of the annotated inputs resulting from $\sst$'s different accepting runs. 

We observe that $\sst'$ can simulate $\sst$ in a completely deterministic way, since the explicit annotations in the input allow $\sst'$ to deterministically determine the transitions of $\sst$ that should be applied (or check whether the input is a valid nd-annotation). Since $\sst'$ needs to store in its states the number that precedes the next symbol of $\Sigma$ of the input, it needs $\bigo(n \cdot \max\{n, m\})$ states. Furthermore, $\sst'$ can be constructed in time $\bigo(|\sst| + n \cdot \max\{n, m\})$.

Let $\varset = \{\varsx_1, \varsx_2, \ldots, \varsx_{\max\{n, m\}}\}$. We represent nd-annotations as multiref-words over $\Sigma$ and $\varset$ in the obvious way, i.\,e., in an nd-annotation, we simply represent every number $i$ by the parentheses $\open{\varsx_i} \close{\varsx_i}$. This obviously represents an nd-annotation as a normalised multiref-word. 

Let us now define a multispanner-expression. For $i \in \mathbb{N}$, let $r_i = (\open{\varsx_1} \close{\varsx_1}) + (\open{\varsx_2} \close{\varsx_2}) + \ldots + (\open{\varsx_i} \close{\varsx_i})$, let $r = r_n \cdot \big(r_m \cdot \Sigma \big)^*$. Obviously, $\lang{r}$ is the set of all possible nd-annotations for an input over $\Sigma$ in the multiref-word interpretation explained above. Moreover, $\nrefw{\lang{r}} = \lang{r}$, $|r| = \bigo(|\Sigma| + \max\{m, n\})$ and $r$ can be constructed in time $\bigo(|\Sigma| + \max\{m, n\})$.

With our considerations from above, we can therefore conclude that $\sem{r \cdot \sst'}(\doc) = \sem{\sst}(\doc)$ for every~$\doc \in \Sigma^*$.
\end{proof}

\section{Proofs of Section~\ref{sec:enumeration}}\label{sec:enumAppendix}
	
\subsection{Proof of Proposition~\ref{prop:garbageFree}}

\begin{proof}
During this proof, the following notation will be useful. Let $R \subseteq \Reg$ be a set of registers. For an assignment $\sigma$, we denote by $\pi_R(\sigma)$ the restriction of $\sigma$ to $R$, namely, $\dom(\pi_R(\sigma)) = \dom(\sigma) \cap R$ and $[\pi_R(\sigma)](X) = \sigma(X)$ for every $X \in \dom(\sigma) \cap R$.

Let $\nsst = (Q, \Sigma, \Omega,  \Reg, \Delta, I, F)$ be an NSST, which we will transform into an equivalent NSST that is garbage-free. The intuitive idea is that we use states $(q, R)$ with $q \in Q$ and $R \subseteq \Reg$, where $R$ simply stores the registers that are to be used by the next assignment, i.\,e., $R = \reg(\sigma)$, where $\sigma$ is the assignment of the next transition. Moreover, to achieve garbage-freeness, $R$ must also be the domain of the assignment of the previous transition that led to state $(q, R)$. Consequently, if we move from state $(q, R)$ to $(q', R')$, this must be done by a transition with an assignment $\sigma$ that satisfies $R = \reg(\sigma)$ and $\dom(\sigma) = R'$. Note that for the next transition in a computation, $R'$ now plays the role of the registers to be used by the next assignment $\sigma'$, i.\,e., $R' = \reg(\sigma')$, where $\sigma'$ is the assignment of the following transition. If we simulate $\nsst$ while maintaining the sets $R$ of the states $(q, R)$ like described above, we satisfy the requirement of garbage-freeness, and obtain an equivalent NSST. Let us now define this formally.
	
	Let $\nsst' = (Q', \Sigma, \Omega, \Reg, \Delta', I', F')$ be an NSST such that $Q' = Q \times \powerset(\Reg)$, $I'((q,R)) = \pi_R(I(q))$ whenever $I(q) \neq \bot$, and $F'((q,R)) = F(q)$ iff $F(q) \neq \bot$ and $R = \reg(F(q))$ for every $(q, R) \in Q'$.
	Furthermore, we will have that $((q,R), a, \sigma, (q',R')) \in \Delta'$ if, and only if, there exists $(q, a, \tau, q') \in \Delta$ such that $\sigma = \pi_{R'}(\tau)$, and
	$R = \reg(\sigma)$.
	Furthermore, we assume that each transition $((q,R), a, \sigma, (q',R')) \in \Delta'$ is repeated in $\Delta'$ as many times as transitions $(q, a, \tau, q')$ exist satisfying the above~properties.
	
	First, we notice that $\nsst'$ is garbage-free. Indeed, consider an accepting run:
	\begin{equation*}
		\rho \ := \ ((q_0, R_0), \nu_0) \xrightarrow{a_1/\sigma_1} ((q_1, R_1), \nu_1) \xrightarrow{a_2 /\sigma_2} \ldots \xrightarrow{a_n/\sigma_{n}} ((q_{n}, R_n), \nu_n)
	\end{equation*}
	of $\nsst'$ over $a_1\ldots a_n$. One can easily prove by induction that $\dom(\nu_i) = R_i$ for every $i \leq n$. Given that $R_i = \reg(\sigma_{i+1})$ by the definition of $\Delta'$, we conclude that $\dom(\nu_i) = \reg(\sigma_{i+1})$. Furthermore, given that $R_n = \reg(F'((q_n, R_n)))$ by the definition of $F'$, then $\dom(\nu_n) = \reg(F'((q_n, R_n)))$ and we proved that $\nsst'$ is garbage-free.
	
	Next, we show that $\nsst'$ is equivalent to $\nsst$, namely, $\sem{\nsst}(w) = \sem{\nsst'}(w)$ for every string $w$. For one direction, let $\rho$ be an accepting run of $\nsst$ over $a_1\ldots a_n$ of the form:
	\begin{equation} \label{eq:run-garbage-free}
		\rho \ := \ (q_0, \nu_0) \xrightarrow{a_1/\tau_1} (q_1, \nu_1) \xrightarrow{a_2 /\tau_2} \ldots \xrightarrow{a_n/\tau_{n}} (q_{n}, \nu_n) \tag{\dag}
	\end{equation}
	Furthermore, let $R_0, \ldots, R_n$ be sets of registers defined recursively as follow: 
	$R_n = \reg(F(q_n))$ and for every $i < n$:
	$ 
	R_i = \bigcup_{X \in R_{i+1}} \reg(\sigma_{i+1}(X)).
	$
	Then, we can construct a run $\rho'$ of $\nsst'$ over $a_1 \ldots a_n$ such that: 
	\begin{equation*}  \label{eq:run-garbage-free2}
		\rho' \ := \ ((q_0, R_0), \nu_0') \xrightarrow{a_1/\sigma_1} ((q_1, R_1), \nu_1') \xrightarrow{a_2 /\sigma_2} \ldots \xrightarrow{a_n/\sigma_{n}} ((q_{n},R_n), \nu_n') \tag{\ddag}
	\end{equation*}
	where $\sigma_i = \pi_{R_i}(\tau_i)$ and $\nu_i' = \pi_{R_i}(\nu_i)$ for every $i \leq n$. By construction, we can verify that $\rho'$ is an accepting run of $\nsst'$ over $a_1\ldots a_n$ and $\nu_n(F(q_n)) = \nu_n'(F((q_n, R_n)))$. Conversely, we can verify that if we have an accepting run $\rho'$ of $\nsst'$ over $a_1\ldots a_n$ like (\ref{eq:run-garbage-free}), then we can find the corresponding run $\rho$ of $\nsst$ over $a_1 \ldots a_n$ like (\ref{eq:run-garbage-free2}). 
\end{proof}

\subsection{Proof of Proposition~\ref{prop:garbageFreePTime}}

\begin{proof}
	Let us assume that $\nsst$ is \emph{trim}, which means that every state can be reached by some run, and every run can be extended to some accepting run (by standard automata constructions, this is without loss of generality). We also assume that, for every state $p \in Q$, if $I(p)$ is defined, then $p$ has at least one outgoing transition, and if $F(p)$ is defined, then $p$ has at least one incoming transition. \par
	By definition, $\nsst$ is \emph{not} garbage-free if and only if there exists an accepting run 	
	\begin{equation*}
		\rho \ := \ ((q_0, R_0), \nu_0) \xrightarrow{a_1/\sigma_1} ((q_1, R_1), \nu_1) \xrightarrow{a_2 /\sigma_2} \ldots \xrightarrow{a_n/\sigma_{n}} ((q_{n}, R_n), \nu_n)
	\end{equation*}
	such that one of the following properties is satisfied:
	\begin{enumerate}
		\item[] $\dom(\nu_0) \neq \reg(\sigma_1)$.
		\item[] $\dom(\nu_i) \neq \reg(\sigma_{i+1})$, for some $i$ with $0 \leq i \leq n-1$.
		\item[] $\dom(\nu_n) \neq \reg(F(q_n))$. 
	\end{enumerate}
	
	We observe that, by induction, this means that $\nsst$ is \emph{not} garbage-free if and only if one of the following properties is satisfied:
	\begin{enumerate}
		\item There is a transitions $(p, a, \sigma, q) \in \Delta$ such that $I(p)$ is defined and $\dom(I(p)) \neq \reg(\sigma)$.
		\item There are transitions $(p, a, \sigma_1, q), (q, b, \sigma_2, r) \in \Delta$ such that $\dom(\sigma_1) \neq \reg(\sigma_2)$.
		\item There is a transition $(p, a, \sigma, q) \in \Delta$ such that $F(q)$ is defined and $\dom(\sigma) \neq \reg(F(q))$.
	\end{enumerate}
	
	For every state $p \in Q$, we say that $p$ is \emph{reg-consistent} if there are no two outgoing transitions $(p, a, \sigma_1, q), (p, b, \sigma_2, r) \in \Delta$ with $\reg(\sigma_1) \neq \reg(\sigma_2)$, and we say that $p$ is \emph{dom-consistent} if there are no two incoming transitions $(q, a, \sigma_1, p), (r, b, \sigma_2, p) \in \Delta$ with $\dom(\sigma_1) \neq \dom(\sigma_2)$.
	
	If there is a state that is not reg-consistent or not dom-consistent, then one of the three properties from above is satisfied. Indeed, let us assume that $p$ is not reg-consistent, which means that there are $(p, a, \sigma_1, q), (p, b, \sigma_2, r) \in \Delta$ with $\reg(\sigma_1) \neq \reg(\sigma_2)$. If $I(p)$ is defined, then either $\dom(I(p)) \neq \reg(\sigma_1)$ or $\dom(I(p)) \neq \reg(\sigma_2)$, which means that the first property is satisfied. If $I(p)$ is not defined, then, since $\nsst$ is trim, there must be at least one incoming transition $(s, c, \sigma_3, p) \in \Delta$. Thus, $\dom(\sigma_3) \neq \reg(\sigma_1)$ or $\dom(\sigma_3) \neq \reg(\sigma_2)$, which means that the second property is not satisfied. On the other hand, if $p$ is not dom-consistent, then there are $(q, a, \sigma_1, p), (r, b, \sigma_2, p) \in \Delta$ with $\dom(\sigma_1) \neq \dom(\sigma_2)$. Now if $F(p)$ is defined, then $\dom(\sigma_1) \neq \reg(F(p))$ or $\dom(\sigma_2) \neq \reg(F(p))$, which means that the third property is satisfied. If $F(p)$ is not defined, then, since $\nsst$ is trim, there must be at least one outgoing transition $(p, c, \sigma_3, s) \in \Delta$. Thus, $\dom(\sigma_1) \neq \reg(\sigma_3)$ or $\dom(\sigma_2) \neq \reg(\sigma_3)$, which means that the second property is satisfied.
	
	If all states are both reg-consistent and dom-consistent, then every state $p$ has a unique \emph{reg-value} $\mathsf{regv}(p)$ (i.\,e., the set $\reg(\sigma)$ for any transition $(p, a, \sigma, q) \in \Delta$, or undefined if $p$ has no outgoing transition) and a unique \emph{dom-value} $\mathsf{domv}(p)$ (i.\,e., the set $\dom(\sigma)$ for any transition $(q, a, \sigma, p) \in \Delta$, or undefined if $p$ has no incoming transitions). 
	
	We can check whether every state is reg- and dom-consistent as follows. We iterate through all transitions $(p, a, \sigma, q) \in \Delta$, and for each such transition, we set $\mathsf{regv}(q)$ to $\reg(\sigma)$ and we set $\mathsf{domv}(p)$ to $\dom(\sigma)$. Whenever we overwrite an already set dom- or reg-value of a state with a \emph{different} dom- or reg-value, we have discovered a state that is not dom- or not reg-consistent. Otherwise, all states are dom- and reg-consistent and we have computed the values $\mathsf{regv}(p)$ and $\mathsf{domv}(p)$ for every $p \in Q$. This requires time $\bigo(|\Delta| |\Reg|)$.
	
	If all the states are indeed dom- and reg-consistent, then the only way of how one of the three properties from above can be satisfied is that for some $p \in Q$ we have that $I(p)$ is defined and $\dom(I(p)) \neq \mathsf{regv}(p)$ (note that, by assumption, $\mathsf{regv}(p)$ must be defined), or that $F(p)$ is defined and $\reg(F(p)) \neq \mathsf{domv}(p)$ (note that, by assumption, $\mathsf{domv}(p)$ must be defined), or that both $\mathsf{regv}(p)$ and $\mathsf{domv}(p)$ are defined, but $\mathsf{regv}(p) \neq \mathsf{domv}(p)$. 
	
	Since we have already computed the values $\mathsf{regv}(p)$ and $\mathsf{domv}(p)$ for every $p \in Q$, this can be checked in time $\bigo(|Q||\Reg|)$.
	
	Hence, we can check whether $\nsst$ is garbage-free in time $\bigo(|\Delta| |\Reg|)$. 
\end{proof}

\subsection{Proof of Theorem~\ref{theo:enum}}

In this subsection, we present an enumeration algorithm with linear-time preprocessing for the problem $\textsc{EnumLinearET}[(E, T)]$. As stated before, we can assume that $E$ is given by a regex multispanner and $T$ by a DSST. Then, the algorithm starts by converting $(E,T)$ into an $\nsstET$ equivalent to $(E,T)$ by Theorem~\ref{ETtoNSST}.  For this reason, we dedicate the rest of this subsection on the following enumeration problem. Let $\nsst$ be an NSST.
\begin{center}
	\framebox{
		\begin{tabular}{rl}
			\textbf{Problem:} & $\textsc{EnumNSST}[\nsst]$ \\
			\textbf{Input:} & A string $w \in \Sigma^*$ \\
			\textbf{Output:} & Enumerate $\sem{\nsst}(w)$
		\end{tabular}
	}
\end{center}
By Proposition~\ref{prop:garbageFree}, we can further assume that $\nsst$ is garbage-free. We will use this property towards the end of this section, when we present the final evaluation algorithm. 

The rest of this section is as follows. We start by presenting the data structure, called Enumerable Compact Set with Assignments. We continue by showing how to enumerate the results from this data structure and the conditions we need to enforce output-linear delay. Then, we present the operations we use to extend and union results in the data structure. We end by showing and analyzing the evaluation algorithm, using the data structure and the garbage-free property. 

\paragraph{The data structure} Let $\Reg$ be a finite set of registers and $\Omega$ a finite alphabet. Recall that $\Asg(\Reg, \Omega)$ is the set of all copyless assignment and $\Val(\Reg, \Omega)$ the set of all valuations, both over $\Reg$ and $\Omega$.
We say that $\gamma \in \Asg(\Reg, \Omega)$ is a \emph{relabel} (or relabeling assignment) iff it is a partial function $\gamma: \Reg \pmap \Reg$. Note that, since a relabel $\gamma$ must be copyless, then $\gamma$ just reassigns the content between registers. We say that $\sigma$ is a \emph{non-relabeling} assignment if $\sigma$ is not a relabel. We usually denote $\sigma$ for an assignment, $\nu$ for a valuation, and $\gamma$ for a relabeling.

We define a \emph{Enumerable Compact Set with Assignments} (ECSA) as a tuple:
\[
\D = (\Reg, \Omega, \N, \ell, r, \lambda)
\] 
such that $\N$ is a finite set of nodes, $\ell\colon \N \pmap \N$ and $r\colon \N \pmap \N$ are the {\em left} and {\em right} partial functions, and $\lambda\colon \N\to\Asg(\Reg, \Omega) \cup\{\cup\}$ is a labeling function. For every node $\n \in \N$, if $\lambda(\n) = \cup$, then both $\ell(\n)$ and $r(\n)$ are defined; if $\lambda(\n)$ is an assignment, then $r(\n)$ is not defined (i.e., $\ell(n)$ could be defined or not); and  $\lambda(\n)$ is an assignment and $\ell(n) = \bot$ iff $\lambda(\n)$ is a valuation. 
Further, we assume that the directed graph $(\N, \{(\n,\ell(\n)), (\n, r(\n)) \mid \n \in \N\})$ is acyclic. Note that, by the previous definition, $\n$ is a bottom node in this acyclic graph iff $\lambda(\n)$ is a valuation. We define the size of $\D$ as $|\D| = |\N|$.

For the analysis, it will be useful to distinguish between different type of nodes. We will say that $\n$ is a {\em union node} (or $\cup$-node) if $\lambda(\n) = \cup$, and an output node, otherwise.
Between output nodes, we will say that $\n$ is a \emph{$\nu$-node} (or also a \emph{bottom} node) if $\lambda(\n)$ is a valuation, a \emph{$\gamma$-node} if $\lambda(\n)$ is relabel, and a \emph{$\sigma$-node} if neither a valuation nor a relabel. 
 
For each $\n \in \N$, we associate a bag of valuations $\sem{\D}(\n)$ recursively as follows: $\sem{\D}(\n) = \multiset{\lambda(\n)}$ whenever $\n$ is a $\nu$-node, $\sem{\D}(\n) = \sem{\D}(\ell(\n)) \cup \sem{\D}(r(\n))$ whenever $\n$ is a union node, and $\sem{\D}(\n) = \multiset{\nu \circ \lambda(\n) \mid \nu \in\sem{\D}(\ell(v))}$ whenever $n$ is either a $\sigma$-node or a $\gamma$-node.

\paragraph{The enumeration problem} The plan is to use ECSA to encode a set of valuations of an NSST over a string, and then enumerate the set of valuations efficiently. Specifically, we consider the following enumeration problem:
\begin{center}
	\framebox{
		\begin{tabular}{rl}
			\textbf{Problem:} & $\enumECSA$\\
			\textbf{Input:} & An ECSA $\D = (\Reg, \Omega, \N, \ell, r, \lambda)$ and $\n \in \N$.  \\
			\textbf{Output:} & Enumerate all valuations in the bag $\sem{\D}(\n)$.
		\end{tabular}
	}
\end{center}
As usual, we want to enumerate the set $\sem{\D}(\n)$ with output-linear delay. However, notice that we want to enumerate all valuations in $\sem{\D}(\n)$, possibly with repetitions. In the literature (see~\cite{Segoufin13}), the usual goal is to enumerate the outputs without repetitions. Nevertheless, we provide NSST with a bag semantics and, therefore, it makes sense to enumerate all valuations with repetitions. 

For solving this enumeration problem efficiently, we require to impose further restrictions over ECSA. We call these restrictions \emph{$k$-bounded} and \emph{garbage-free}, respectively.
\begin{enumerate}
	\item We define the notion of \emph{$k$-bounded} ECSA as follows. 
	Given $\D$, define the (left) output-depth of a node $\n\in \N$, denoted by $\odepth_{\D}(\n)$, recursively as follows:
	$\odepth_{\D}(\n) = 0$ whenever $\n$ is a $\nu$-node or a $\sigma$-node, and $\odepth_{\D}(\n) = \odepth_{\D}(\ell(\n))+1$ whenever $\lambda(v)$ is a union node or a $\gamma$-node.
	Then, for $k\in\bbN$ we say that $\D$ is $k$-bounded if $\odepth_{\D}(v)\leq k$ for all $\n\in \N$.
	
	\item We say that $\D$ is \emph{garbage-free} if, 
	for every $\sigma$-node or $\gamma$-node $\n$, it holds that $\dom(\nu) = \reg(\lambda(\n))$ for every $\nu \in \sem{\D}(\ell(\n))$. 
\end{enumerate}
Similar as in~\cite{MunozR22}, the intuition behind $k$-boundedness is that every node has an output node close by following a constant number of left children. Here, we have to distinguish between output nodes that are a relabel (i.e., a $\gamma$-node) and those that are not (i.e., $\nu$- or $\sigma$-nodes). Intuitively, a relabel does not contribute to the size of an output; therefore, it increases the size of the output depth in one. For this reason, our goal will be to restrict contiguous sequences of relabels in the data structure during node operations (see below). For the garbage-free restriction, the intuition comes that one could have a node $\n$ that produces a valuation $\nu \in \sem{\D}(\n)$ whose content $\nu(x)$ is not used by the parent node of $\n$. This situation can generate sequences of output nodes that later do not contribute to the size of an output, increasing the delay. For this reason, garbage-free ECSA forces all parent nodes to use all the registers of the child's valuations, and then all registers will contribute to the final output. 

\begin{algorithm}[t]
	\caption{Enumeration over a node $\n$ from some garbage-free ECSA $\D = (\Reg, \Omega, \N, \ell, r, \lambda)$.}\label{alg:enumeration}
	\smallskip
	\begin{varwidth}[t]{0.5\textwidth}
		\begin{algorithmic}[1]
			\Procedure{{enumerate}}{$\n$} \label{alg1enum1}
			\State $\uit \gets \textsc{create}(\n)$
			\While{$\uit.\textsc{next} = \bf{true}$}\label{alg1enum2}
			\State $\uit.\textsc{print}(\uid)$\label{alg1enum3}
			\EndWhile
			\EndProcedure
			
			\State
			
			\LineComment{$\nu$-node iterator $\uit_\nu$} \label{alg1bottom0}
			\Procedure{{create}}{$\n$}\,\Comment{$\n$ is a $\nu$-node} \label{alg1bottom1}
			\State $\bnodelabel \gets \n$ \label{alg1bottom2}
			\State $\hasnext \gets {\bf true}$ \label{alg1bottom3}
			\EndProcedure	
			
			\State
			
			\Procedure{{next}}{} \label{alg1bottom4}
			\If{$\hasnext = \bf{true}$} \label{alg1bottom5}
			\State $\hasnext \gets \bf{false}$ \label{alg1bottom6}
			\State \Return $\bf{true}$ \label{alg1bottom7}
			\EndIf
			\State \Return $\bf{false}$ \label{alg1bottom8}
			\EndProcedure
			
			\State
			
			\Procedure{{print}}{$\sigma$} \label{alg1bottom9}
			\State ${\bf print}: \lambda(\m) \circ \sigma$ \label{alg1bottom10}
			\EndProcedure \label{alg1bottom11}
			\State 
			\LineComment{$\sigma$-node iterator  $\uit_\sigma$} \label{alg1prod0}
			\Procedure{{create}}{$\n$}\,\Comment{$\n$ is a $\sigma$-node}\label{alg1prod1}
			\State $\bnodelabel \gets \n$
			\State $\uit \gets \textsc{create}(\ell(\bnodelabel))$ \label{alg1prod5}
			\EndProcedure	
			
			\State
			
			\Procedure{{next}}{} \label{alg1prod6}
			\State \Return $\uit.\textsc{next}$ \label{alg1prod12}
			\EndProcedure
			
			\State
			
			\Procedure{{print}}{$\sigma$} \label{alg1prod13}
			\State $\uit.\textsc{print}(\lambda(\bnodelabel) \circ \sigma)$ \label{alg1prodlast}
			\EndProcedure	
			\algstore{myalg}
		\end{algorithmic}
	\end{varwidth} \hfill
	\begin{varwidth}[t]{0.5\textwidth}
		\begin{algorithmic}[1]
			\algrestore{myalg}
			\LineComment{Union/$\gamma$-node iterator $\uit_{\cup/\gamma}$} \label{alg1union0}
			\Procedure{{create}}{$\n$}\label{alg1union1}
			\State $\ \ \ \ \ \ \ \ \ \ \ \ \ \ \ \ \ \ $ \Comment{$\n$ is a $\cup$- or $\gamma$-node}
			\State $\ustack\gets\text{\sf push}(\ustack, (\n, \uid)\,)$\label{alg1union4}
			\State $\ustack\gets\textsc{traverse}(\ustack)$\label{alg1union5}
			\State $\uit \gets \textsc{create}(\text{\sf top}(\ustack).\bnodelabel)$ \label{alg1union6}
			\EndProcedure	
			
			\State
			
			\Procedure{{next}}{}
			\If{$\uit.\textsc{next} = \bf{false}$}\label{alg1unionnext1}
			\State $\ustack\gets\text{\sf pop}(\ustack)$ \label{alg1unionnext2}
			\If{$\text{\sf length}(\ustack) = 0$} \label{alg1unionnext3}
			\State \Return $\bf{false}$
			\Else \label{alg1unionnext5}
			\State $\ustack\gets\textsc{traverse}(\ustack)$ \label{alg1unionnext6}
			\EndIf
			\State $\uit \gets \textsc{create}(\text{\sf top}(\ustack).\bnodelabel)$ \label{alg1unionnext7}
			\State $\uit.\textsc{next}$ \label{alg1unionnext8}
			\EndIf
			\State \Return $\bf{true}$
			\EndProcedure
			
			\State
			
			\Procedure{{print}}{$\sigma$} 
			\State $(\m, \gamma')\gets \text{\sf top}(\ustack)$ \label{alg1unionprint1}
			\State $\uit.\textsc{print}(\gamma' \circ \sigma)$ \label{alg1unionlast}
			\EndProcedure
			\State
			\Procedure{{traverse}}{$\ustack$}\label{alg1traverse0}
			\While{$\text{\sf top}(\ustack).\n$ is a $\cup$- or $\gamma$-node}
			\State $(\m,\gamma') \gets \text{\sf top}(\ustack)$
			\State $\ustack \gets \text{\sf pop}(\ustack)$
			\If{$\m$ is a $\gamma$-node} \label{alg1traverseIfStart}
			\State $\gamma'\gets \lambda(\m) \circ \gamma'$ \label{alg1relabelcomp}
			\State $\ustack\gets\text{\sf push}(\ustack, (\ell(\m), \, \gamma'))$ \label{alg1traverseIfEnd}
			\Else \label{alg1traverseElseStart}
			\State $\ustack\gets\text{\sf push}(\ustack, (r(\m), \gamma')\,)$
			\State $\ustack\gets\text{\sf push}(\ustack, (\ell(\m),\gamma')\,)$ \label{alg1traverseElseEnd}
			\EndIf
			\EndWhile
			\State \Return $\ustack$  \label{alg1traverseEnd}
			\EndProcedure
		\end{algorithmic}
	\end{varwidth}
	\smallskip
\end{algorithm} 
As we show in the next proposition, if $\D$ is garbage-free and $k$-bounded, then we can enumerate all valuations with output-linear delay.
\begin{proposition}\label{prop:enumeration}
	Assume that $k$ and $\Reg$ are fixed. Let $\D = (\Reg, \Omega, \N, \ell, r, \lambda)$ be a garbage-free and $k$-bounded ECSA. Then one can enumerate the set $\sem{\D}(\n)$ with output-linear delay for any node $\n\in \N$.
\end{proposition}
\begin{proof}	
	Let $\D = (\Reg, \Omega, \N, \ell, r, \lambda)$ be a garbage-free and $k$-bounded ECSA.
	We follow the same enumeration procedure for Enumerable Compact Sets in~\cite{MunozR22,MunozR23}, adapted for our purposes.
	This strategy follows a depth-first-search traversal of the data-structure, done in a recursive fashion to ensure that after retrieving some output $u$, the next one $u'$ can be printed in $O(k\cdot (|u| + |u'|))$ time. The entire procedure is detailed in Algorithm~\ref{alg:enumeration}.

	For simplifying the presentation, we use an \emph{iterator interface} that, given a node $\n$, contains all information and methods to enumerate the outputs $\sem{\D}(\n)$. Specifically, an iterator $\uit$ must implement the following three methods:
	\[
	\begin{array}{rclrclrcl}
		\textsc{create}(\n) & \!\!\!\!\rightarrow\!\!\!\! & \uit   \ \ \ \ \ \	 & \uit.\textsc{next} & \!\!\!\! \rightarrow \!\!\!\! & b  \ \ \ \ \ \ \ &  \uit.\textsc{print}(\sigma)  & \!\!\!\!\rightarrow\!\!\!\! & \emptyset
	\end{array}
	\]
	where $\n$ is a node, $b$ is boolean, and $\emptyset$ means that the method does not return an output. 
	The first method, \textsc{create}, receives a node $\n$ and creates an iterator $\uit$ of the type of $\n$ (e.g., $\nu$-node). We will implement three types of iterators, one for $\nu$-nodes ($\uit_\nu$), one for $\sigma$-nodes ($\uit_\sigma$), and one for $\cup$- and $\gamma$-nodes together ($\uit_{\cup/\gamma}$). The second method, $\uit.\textsc{next}$, moves the iterator to the next output, returning {\bf true} if, and only if, there is an output to print. Then the last method, $\uit.\textsc{print}$, receives any assignment $\sigma \in \Asg(\Reg, \Omega)$, and prints the current valuation $\nu$ pointed by $\uit$ to the output registers after composing it with $\sigma$, namely, it outputs $\nu \circ \sigma$. 
	We assume that one must first call $\uit.\textsc{next}$ to move to the first output before printing. Furthermore, if $\uit.\textsc{next}$ outputs {\bf false}, then the behavior of $\uit.\textsc{print}$ is undefined. Note that one can call $\uit.\textsc{print}$ several times, without calling $\uit.\textsc{next}$, and the iterator will write the same output each time in the output registers. 
	
	Assume we can implement the iterator interface for each type. Then the procedure $\textsc{Enumerate}(v)$ in Algorithm~\ref{alg:enumeration} (lines \ref{alg1enum1}-\ref{alg1enum3}) shows how to enumerate the set $\sem{\D}(v)$ by using an iterator $\uit$ for $\n$. We assume here that $\uid$ is the identity assignment, namely, $\uid(r) = r$ for every $r \in \Reg$.  In the following, we show how to implement the iterator interface for each node type and how the size of the next valuation bounds the delay between two outcomes.

	We start by presenting the iterator $\uit_\nu$ for a  $\nu$-node $\n$ (lines \ref{alg1bottom0}-\ref{alg1bottom11}), called a \emph{ $\nu$-node iterator}. We assume that each $\uit_\nu$ has two internal values, denoted by $\bnodelabel$ and $\hasnext$, where $\bnodelabel$ is a reference to $\n$ and $\hasnext$ is a boolean variable. The purpose of a bottom node iterator is only to print $\lambda(\n) \circ \sigma$ for some assignment $\sigma$. Towards this goal, when we create $\uit_\nu$, we initialize $\bnodelabel$ equal to $\n$ and $\hasnext = {\bf true}$ (lines \ref{alg1bottom2}-\ref{alg1bottom3}). When $\uit_\nu.\textsc{next}$ is called for the first time, we swap $\hasnext$ from {\bf true} to {\bf false} and output {\bf true} (i.e., there is one output ready to be returned). Then any following call to $\uit_\nu.\textsc{next}$ will be false (line \ref{alg1bottom8}). Finally, the $\uit_\nu.\textsc{print}$ writes $\lambda(\bnodelabel) \circ \sigma$ to the output registers (lines \ref{alg1bottom9}-\ref{alg1bottom11}).

	For a $\sigma$-node, we present a \emph{$\sigma$-node iterator} $\uit_\sigma$ in Algorithm~\ref{alg:enumeration} (lines \ref{alg1prod0}-\ref{alg1prodlast}). 
	This iterator receives a $\sigma$-node $\n$ and stores a reference of $\n$, called $\bnodelabel$, and an iterator $\uit$ for iterating through the valuations of $\ell(\n)$ (i.e., $\sem{\D}(\ell(n))$). 
	The \textsc{create} method initializes $\bnodelabel$ with $\n$ and creates the iterator $\uit$ (lines \ref{alg1prod1}-\ref{alg1prod5}).
	The purpose of $\uit_\sigma.\textsc{next}$ is just to iterate over all outputs of $\ell(\bnodelabel)$ (lines \ref{alg1prod6}-\ref{alg1prod12}). 
	Then, for printing, we recursively call the printing method of $\uit$ given as a parameter the compose assignment $\lambda(\bnodelabel) \circ \sigma$ (lines~\ref{alg1prod13}-\ref{alg1prodlast}).
	
	The most involved case is the \emph{union/$\gamma$-node iterator} $\uit_{\cup/\gamma}$ (lines \ref{alg1union0}-\ref{alg1unionlast}). 
	This iterator receives either a union node, or a $\gamma$-node. It keeps a \emph{stack} $\ustack$ and an iterator~$\uit$. The elements in the stack are pairs $(\bnodelabel, \gamma')$ where $\bnodelabel$ is a node and $\gamma'$ is a relabel. We assume the standard implementation of a stack with the native methods $\text{\sf push}$, $\text{\sf pop}$, $\text{\sf top}$, and $\text{\sf length}$: the first three define the standard operations over stacks, and $\text{\sf length}$ counts the elements in a stack. 
	The purpose of the stack is to perform a \emph{depth-first-search} traversal of all union and $\gamma$-nodes below the input node $\n$, reaching all possible output nodes $\bnodelabel$ such that there is a path of only union and $\gamma$-nodes between $\n$ and $\bnodelabel$. At every point, if an element $(\bnodelabel,\gamma')$ is in the stack, then $\gamma'$ is equal to the composition of all $\gamma$-nodes in the path from $\n$ to $\bnodelabel$. If the top node of $\ustack$ is a pair $(\bnodelabel, \gamma')$ such that $\bnodelabel$ is a non-relabeling node, then $\uit$ is an iterator for $\bnodelabel$, which enumerates all their valuations. If $p = (\bnodelabel,\gamma')$, we will use the notation $p.\bnodelabel$ to refer to $\bnodelabel$.
	
	To perform the \emph{depth-first-search} traversal of $\cup$- and $\gamma$-nodes, we use the auxiliary method called $\textsc{traverse}(\ustack)$ (lines \ref{alg1traverse0}-\ref{alg1traverseEnd}). While the node $\bnodelabel$ at the top of $\ustack$ is a $\cup$- or $\gamma$-node, we pop the top pair $(\bnodelabel,\gamma')$ from $\ustack$. If $\bnodelabel$ is a $\gamma$-node (lines~\ref{alg1traverseIfStart}-\ref{alg1traverseIfEnd}), we push the pair $(\ell(\bnodelabel), \lambda(\m) \circ \gamma')$ in the stack, namely, we move to the next node and compose the old relabel $\gamma'$ with the current relabel $\lambda(\bnodelabel)$.
	Otherwise, if $\bnodelabel$ is a union node (lines~\ref{alg1traverseElseStart}-\ref{alg1traverseElseEnd}), we first push the right pair $(r(\bnodelabel), \gamma')$ followed by the left pair $(\ell(\bnodelabel), \gamma')$ into the stack. 	
	The while-loop will eventually reach a $\nu$- or $\sigma$-node at the top of the stack and end. 
	It is important to note that $\textsc{traverse}(\ustack)$ takes $\cO(k)$ steps, given that the ECSA is $k$-bounded. Then if $k$ is fixed, the  $\textsc{traverse}$ procedure takes constant time (assuming that $|\Reg|$ is fixed).
	
	The methods of a union/$\gamma$ node iterator $\uit_{\cup/\gamma}$ are then straightforward. For \textsc{create} (lines \ref{alg1union1}-\ref{alg1union6}), we push $(\n,\uid)$ to the top of the stack (recall that $\uid$ is the identity assignment) and then  $\textproc{Traverse}(\ustack)$, finding the first leftmost output node from~$\n$~(lines~\ref{alg1union4}-\ref{alg1union5}). 
	Then we build the iterator $\uit$ of this output node for being ready to start enumerating their outputs (line~\ref{alg1union6}). 
	For \textsc{next}, we consume all outputs by calling $\uit.\textsc{next}$ (line~\ref{alg1unionnext1}). When there are no more outputs, we pop the top node from $\ustack$ and check if the stack is empty or not (lines~\ref{alg1unionnext2}-\ref{alg1unionnext3}). If this is the case, there are no more outputs and we output {\bf false}. Otherwise, we call the \textproc{Travese} method for finding the leftmost $\nu$- or $\sigma$-node from the top of the stack~(lines \ref{alg1unionnext5}-\ref{alg1unionnext6}). When the procedure is done, we create an iterator and move to its first valuation (lines~\ref{alg1unionnext7}-\ref{alg1unionnext8}). 
	For $\textsc{print}(\sigma)$, we see the pair $(\bnodelabel, \gamma')$ at the top, where we remind that $\gamma'$ represents the composition of all $\gamma$-nodes on the way to $\bnodelabel$ (line~\ref{alg1unionprint1}), and $\bnodelabel$ is assumed to be a $\nu$- or $\sigma$-node. Then, we call the print method of $\uit$ which is ready to write the current valuation, and over which we compose with  $\gamma' \circ \sigma$ (line~\ref{alg1unionlast}). 
	
	To prove the correctness of the enumeration procedure, one can verify that $\textsc{Enumerate}(\n)$ in Algorithm~\ref{alg:enumeration} enumerates all the outputs in the set $\sem{\D}(\n)$ one by one (possibly with repetitions). To bound the delay between outputs, we need some definitions and assumptions. First, we define the size of an assignment $\sigma$ as $|\sigma| = \sum_{X \in \dom(\sigma)} |\sigma(X)|$. Second, we assume that, given two assignments $\sigma$ and $\sigma'$, we can perform the composition $\sigma \circ \sigma'$ in time $O(|\Reg|)$, namely, in constant time assuming that $\Reg$ is fixed. One can do this by representing the concatenation of strings symbolically, and then representing $\sigma \circ \sigma'$ in this fashion. Then, if we want the real representation of $\sigma \circ \sigma'$, we can enumerate each string $\sigma \circ \sigma'(X)$ with constant delay, and then $\sigma \circ \sigma'$ in linear time in $|\sigma \circ \sigma'|$. This will allow us to perform a concatenation $\sigma_1 \circ  \ldots \circ \sigma_n = \sigma$ in time $O(n)$ to then construct $\sigma$ in time $O(|\sigma|)$. 
	
	Coming back to delay, the fact that $\D$ is $k$-bounded and garbage-free implies that the delay is bounded by $\cO(k\cdot |\nu_0|)$ if $\nu_0$ is the first output, or $O(k\cdot (|\nu| + |\nu'|))$ if $\nu$ and $\nu'$ are the previous and next outputs, respectively. Specifically:
	\begin{itemize}
		\item $\textsc{create}(\n)$ takes time $\cO(k\cdot |\nu_0|)$,
		\item $\textsc{next}$ takes time $\cO(k\cdot |\nu_0|)$ for the first call, and $\cO(k\cdot (|\nu|+|\nu'|))$ for the next call, and
		\item $\textsc{print}(\uid)$ takes time $\cO(k\cdot |\nu'|)$ where $\nu'$ is the current output to be printed.  
	\end{itemize}
	It is important to note that each composition operation $\sigma \circ \sigma'$ between two assignments in Algorithm~\ref{alg:enumeration} (e.g. lines \ref{alg1bottom10}, \ref{alg1prodlast}, or \ref{alg1unionlast}) takes linear time in the next valuation $\nu'$, given the assumption above and given that $\D$ is garbage-free and then all the content of $\sigma$ and $\sigma'$ will be used in $\nu'$. Instead, the compositions between relabels $\gamma \circ \gamma'$ (e.g., line~\ref{alg1relabelcomp}) take time $O(k\cdot |\Reg|)$, which is constant and do not affect the delay for the next output $\nu'$. 
	Overall, $\textsc{Enumerate}(\n)$ in Algorithm~\ref{alg:enumeration} requires $O(k\cdot (|\nu| + |\nu'|))$ delay to write the next valuation $\nu'$ in the output register, after printing the previous output $\nu$. 
	
	We end by noticing that the existence of an enumeration algorithm with delay $O(k\cdot (|\nu| + |\nu'|))$ between any consecutive outputs $\nu$ and~$\nu'$, implies the existence of an enumeration algorithm with output-linear delay as defined in Section~\ref{sec:enumeration} (see~\cite{MunozR22,MunozR23}).
\end{proof}

Note that the delay on the previous proposition depends on $|\Reg|$ that we assume that it is fix. For this reason, we get output-linear delay given that we assume that $|\Reg|$ is a constant value. 

\begin{figure}[t]
	\centering
	\begin{tikzpicture}[->,
		>=stealth',
		initial text= {},
		initial distance= {2.5mm},
		roundnode/.style={circle,inner sep=1pt},
		rootnode/.style={circle,inner sep=1pt},
		anynode/.style={circle,inner sep=2.5pt,draw},
		greynode/.style={circle,inner sep=1pt,black!60, node distance=5mm},
		squarednode/.style={rectangle,inner sep=2pt}]
		\node[rootnode] (s) at (0.5, 0) {$\sigma$};
		\node[greynode,left of=s] {$\s_1\colon$};
		
		\begin{scope}[xshift=2.5cm]
			\node[roundnode] (s) at (0, -1) {$\sigma$};	
			\node[rootnode] (r) at (0, 0) {$\gamma$};
			\node[greynode,left of=r] {$\s_2\colon$};
			
			\draw[solid] (r) to (s);	
		\end{scope}
		
		\begin{scope}[xshift=5cm]
			\node[rootnode] (u) at (0, 0) {$\cup$};	
			\node[greynode,left of=u] {$\s_3\colon$};
			\node[roundnode] (s) at (-0.5, -1) {$\sigma$};
			\node[anynode] (n) at (0.5, -1) {\,};
			\node[greynode,right of=n] {$:\!\m$};
			\draw[solid] (u) to (s);
			\draw[dashed] (u) to (n);	
		\end{scope}
		\begin{scope}[xshift=7.5cm]
			\node[rootnode] (u) at (0, 0) {$\cup$};	
			\node[greynode,left of=u] {$\s_4\colon$};
			\node[roundnode] (r) at (-0.5, -1) {$\gamma$};
			\node[roundnode] (s) at (-0.5, -2) {$\sigma$};
			\node[anynode] (n) at (0.5, -1) {\,};
			\node[greynode,right of=n] {$:\!\m$};
			\draw[solid] (u) to (r);
			\draw[solid] (r) to (s);
			\draw[dashed] (u) to (n);	
		\end{scope}
		\begin{scope}[xshift=10cm]
			\node[rootnode] (r) at (0, 0) {$\gamma$};
			\node[greynode,left of=r] {$\s_5\colon$};
			\node[roundnode] (u) at (0, -1) {$\cup$};	
			\node[roundnode] (s) at (-0.5, -2) {$\sigma$};
			\node[anynode] (n) at (0.5, -2) {\,};
			\node[greynode,right of=n] {$:\!\m$};
			\draw[solid] (r) to (u);
			\draw[solid] (u) to (s);
			\draw[dashed] (u) to (n);
		\end{scope}
		\begin{scope}[xshift=12.5cm]
			\node[rootnode] (r) at (0, 0) {$\gamma$};
			\node[greynode,left of=r] {$\s_6\colon$};
			\node[roundnode] (u) at (0, -1) {$\cup$};	
			\node[roundnode] (r') at (-0.5, -2) {$\gamma'$};
			\node[anynode] (n) at (0.5, -2) {\,};
			\node[greynode,right of=n] {$:\!\m$};
			\node[roundnode] (s) at (-0.5, -3) {$\sigma$};
			\draw[solid] (r) to (u);
			\draw[solid] (u) to (r');
			\draw[dashed] (u) to (n);
			\draw[solid] (r') to (s);
		\end{scope}
	\end{tikzpicture}
	\caption{Nodes $\s_1, \ldots, \s_6$ are the six possible cases for a node to be safe, where we assume that the rounded node labeled by $\m$ is any node such that $\odepth_\D(\m) \leq 2$. We use $\gamma$ and $\gamma'$ for a relabel and $\sigma$ for a non-relabeling assignment. Solid arrow is used for the $\ell$-edge and dashed arrow is used for the $r$-edge. Also, $\s_1$ and $\s_2$ are safe output-nodes, and $\s_3$ and $\s_4$ are safe union-nodes.}
	\label{fig:safe-nodes}
\end{figure}
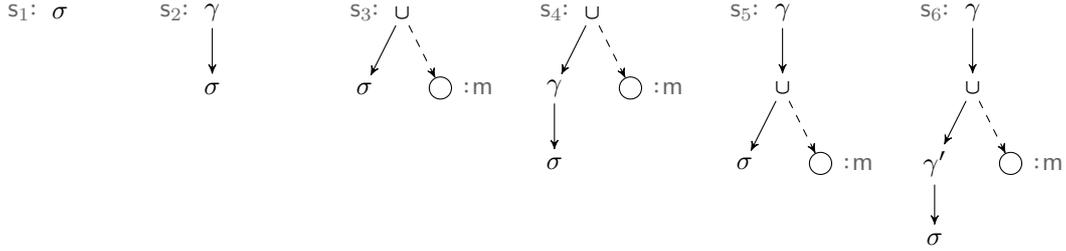

\paragraph{The operations} The next step is to provide a set of operations that allow extending a ECSA $\D$ in a way that maintains $k$-boundedness. Fix an ECSA $\D = (\Reg, \Omega, \N, \ell, r, \lambda)$. For any valuation $\nu \in \Val(\Reg, \Omega)$, any assignment $\sigma \in \Asg(\Reg,\Omega)$, and any nodes $\n$, $\n_1$, and $\n_2$, we define the operations:
\[
\begin{array}{rclrclrcl}
	\dadd(\nu) \! \! \! &\to&\! \! \! \n'  \text{\hspace{1cm}} & \dextend(\n, \sigma) \! \! \! &\to&\! \! \! \n' &  \text{\hspace{1cm}}  \dunion(\n_1, \n_2) \! \! \! &\to&\! \! \! \n' 
\end{array}
\]
such that $\sem{\D}(\n') := \multiset{\nu}$; $\sem{\D}(\n') := \multiset{\nu' \circ \sigma \mid \nu' \in \sem{\D}(\n)}$; and $\sem{\D}(\n') := \sem{\D}(\n_1) \cup \sem{\D}(\n_2)$, respectively.
For $\dextend$ and $\dunion$ we impose some prerequisites to ensure that the results is always garbage-free. Specifically, for $\dextend$ we assume that $\dom(\nu') = \reg(\sigma)$ for every $\nu' \in \sem{\D}(\n)$ and for $\dunion$ we assume that $\dom(\nu_1) = \dom(\nu_2)$ for every $\nu_1 \in \sem{\D}(\n_1)$ and $\nu_2 \in \sem{\D}(\n_2)$.
Both restrictions force that, if $\D$ is garbage-free, then the resulting ECSA after $\dunion$ or $\dextend$ is garbage-free as well. 
Finally, we suppose that these operations extend the data structure with new nodes and that old nodes are immutable after each operation. 

Following the same recipe as in~\cite{MunozR22}, we define a notion of \emph{safe node} for ECSA that will aid for maintaining the $k$-boundedness property after each operation.
For the sake of presentation, we present the notion of safeness by first introducing two additional notions, namely, safe output-node and safe union-node. 
We say that an node $\n$ is a \emph{safe output-node} iff $\n$ is a $\nu$-node or a $\sigma$-node, or if $\n$ is a $\gamma$-node, then $\ell(\n)$ is a $\nu$-node or a $\sigma$-node. 
Instead, we say that $\n$ is a \emph{safe union-node}
iff $\n$ is a union, $\ell(\n)$ is a safe output-node, and $\odepth_{\D}(r(\n)) \leq 2$. 
Finally, we say that $\s$ is a \emph{safe node} iff $\s$ is either a safe output-node, a safe union-node, or $\s$ is a $\gamma$-node and $\ell(\s)$ is a safe union-node. 
In Figure~\ref{fig:safe-nodes}, we display the six possible case for a node to be safe. 

The goal now is to implement $\dadd$, $\dextend$, and $\dunion$ such that the operations receive as input safe nodes and return safe nodes maintaining $\D$ $k$-bounded. For $\dadd(\nu)$, we can create a fresh node $\n'$ such that $\lambda(\n') = \nu$ and then output $\n'$ who is safe. For $\dextend(\n, \sigma)$, assume that $\n$ is safe and $\D$ is $k$-bounded. We have to distinguish two cases: whether $\sigma$ is a relabel or not. If $\sigma$ is not a relabel, then create a fresh node $\n'$ such that $\lambda(\n') = \sigma$ and $\ell(\n') = \n$. We can check that $\n'$ meets the semantics of $\dextend$, $\n'$ is safe, and $\D$ remains $k$-bounded. Instead, if $\sigma$ is a relabel, we have to consider two more cases: whether $\n$ is a $\gamma$-node or not. If $\n$ is not a $\gamma$-node (i.e., cases $\s_1$, $\s_3$, and $\s_4$ in Figure~\ref{fig:safe-nodes}), then add a fresh node $\n'$ such that $\lambda(\n') = \sigma$ and link $\n'$ with $\n$ (i.e., $\ell(\n') = \n$). Otherwise, if $\n$ is a $\gamma$-node with $\lambda(\n) = \gamma$ (i.e., cases $\s_2$, $\s_5$, and $\s_6$ in Figure~\ref{fig:safe-nodes}), then define the fresh node $\n'$ such that $\lambda(\n') = \sigma \circ \gamma$, and connect $\n'$ with $\ell(\n)$, namely, $\ell(\n') = \ell(\n)$. One can check that in both scenarios, $\n'$ is safe, it satisfies the $\dextend$-semantics, and $\D$ is $k$-bounded as expected.

For $\dunion(\n_1, \n_2)$, the construction is a bit more involved. For the sake of simplification, we will only present the case when $\n_1$ and $\n_2$ has the form of $\s_6$ like in Figure~\ref{fig:safe-nodes}. The other cases are variations or simplifications of this case. Assume then that $\n_1$ and $\n_2$ have the form of $\s_6$, specifically, like in Figure~\ref{fig:union-op} (left). The result of $\dunion(\n_1, \n_2)$ is the node $\n'$ at Figure~\ref{fig:union-op} (right).
In this gadget, all nodes are fresh nodes, except $\n_1'$, $\n_2'$, $\m_1$, and $\m_2$ that are reused (i.e., seven fresh nodes in total). First, we can easily check by inspection that $\sem{\D}(\n) = \sem{\D}(\n_1) \cup \sem{\D}(\n_2)$. Second, $\n'$ is safe. Indeed, it has the form of $\s_4$ in Figure~\ref{fig:safe-nodes} since $\n'$ is a safe union-node where $\ell(\n')$ is a safe output-node and $\odepth_\D(\unode_1) \leq 2$ with $\unode_1 = r(\n')$. Finally, the resulting ECSA is $k$-bounded if we choose $k = 4$: 
nodes $\n'$ and $ \unode_1$ are $2$-bounded, and $\unode_2$ and $r(\unode_2)$ are $4$-bounded and $3$-bounded, respectively. 
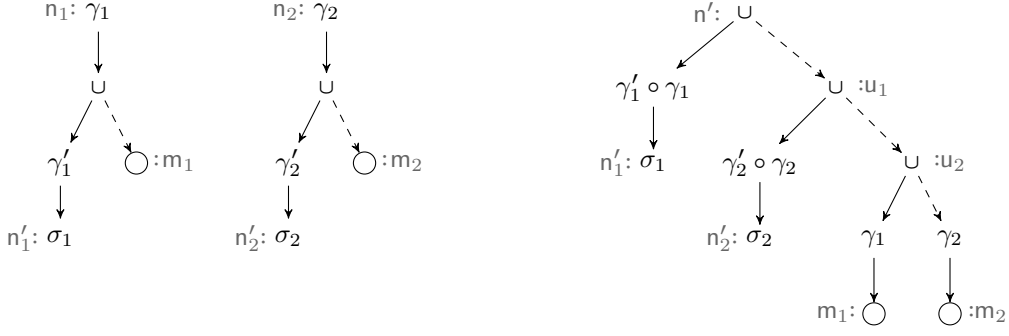
\begin{figure}[t]
	\centering
	\begin{tikzpicture}[->,
		>=stealth',
		initial text= {},
		initial distance= {2.5mm},
		roundnode/.style={circle,inner sep=1pt},
		rootnode/.style={circle,inner sep=1pt},
		anynode/.style={circle,inner sep=2.5pt,draw},
		greynode/.style={circle,inner sep=1pt,black!60, node distance=5mm},
		squarednode/.style={rectangle,inner sep=2pt}]
		
		\node[rootnode] (r) at (0, 0) {$\gamma_1$};
		\node[greynode,left of=r] {$\n_1\colon$};
		\node[roundnode] (u) at (0, -1) {$\cup$};	
		\node[roundnode] (r') at (-0.5, -2) {$\gamma_1'$};
		\node[anynode] (n) at (0.5, -2) {\,};
		\node[greynode,right of=n] {$:\!\m_1$};
		\node[roundnode] (s) at (-0.5, -3) {$\sigma_1$};
		\node[greynode,left of=s] {$\n_1'\colon$};
		\draw[solid] (r) to (u);
		\draw[solid] (u) to (r');
		\draw[dashed] (u) to (n);
		\draw[solid] (r') to (s);
			
		\begin{scope}[xshift=3cm]
			\node[rootnode] (r) at (0, 0) {$\gamma_2$};
			\node[greynode,left of=r] {$\n_2\colon$};
			\node[roundnode] (u) at (0, -1) {$\cup$};	
			\node[roundnode] (r') at (-0.5, -2) {$\gamma_2'$};
			\node[anynode] (n) at (0.5, -2) {\,};
			\node[greynode,right of=n] {$:\!\m_2$};
			\node[roundnode] (s) at (-0.5, -3) {$\sigma_2$};
			\node[greynode,left of=s] {$\n_2'\colon$};
			\draw[solid] (r) to (u);
			\draw[solid] (u) to (r');
			\draw[dashed] (u) to (n);
			\draw[solid] (r') to (s);
		\end{scope}
		
		\begin{scope}[xshift=8.5cm]
			\node[rootnode] (r) at (0, 0) {$\cup$};
			\node[greynode,left of=r] {$\n'\colon$};
			\node[squarednode] (l1) at (-1.2, -1) {$\gamma_1' \circ \gamma_1$};
			\node[squarednode] (o1) at (-1.2, -2) {$\sigma_1$};
			\node[greynode,left of=o1] {$\n_1'\colon$};
			
			\node[roundnode] (u1) at (1.2, -1) {$\cup$};
			\node[greynode,right of=u1] {$\colon\! \unode_1$};
			\node[squarednode] (l2) at (0.2, -2) {$\gamma_2' \circ \gamma_2$};
			\node[squarednode] (o2) at (0.2, -3) {$\sigma_2$};
			\node[greynode,left of=o2] {$\n_2'\colon$};
			
			\node[roundnode] (u2) at (2.2, -2) {$\cup$};
			\node[greynode,right of=u2] {$\colon\!  \unode_2$};
			\node[roundnode] (l1') at (1.7, -3) {$\gamma_1$};
			\node[anynode] (m1) at (1.7, -4) {\,};
			\node[greynode,left of=m1] {$\m_1\colon$};
			\node[roundnode] (l2') at (2.7, -3) {$\gamma_2$};
			\node[anynode] (m2) at (2.7, -4) {\,};
			\node[greynode,right of=m2] {$\colon\! \m_2$};

			\draw[solid] (r) to (l1);
			\draw[dashed] (r) to (u1);
			\draw[solid] (l1) to (o1);
			
			\draw[solid] (u1) to (l2);
			\draw[solid] (l2) to (o2);
			
			\draw[dashed] (u1) to (u2);
			\draw[solid] (u2) to (l1');
			\draw[dashed] (u2) to (l2');
			\draw[solid] (l1') to (m1);
			\draw[solid] (l2') to (m2);

		\end{scope}
		
	\end{tikzpicture}
	\caption{On the left, the structure of input nodes $\n_1$ and $\n_2$ and, on the right, the node $\n'$ that represents $\dunion(\n_1, \n_2)$. In this construction, nodes $\n_1'$, $\n_2'$, $\m_1$, $\m_2$ are reused and all other nodes are fresh nodes.}
	\label{fig:union-op}
\end{figure}

We can conclude that, if we start with a $4$-bounded ECSA $\D$ and we apply $\dadd$, $\dextend$, and $\dunion$ operators between safe nodes, then the operations retrieve safe nodes and the ECSA remains $4$-bounded as well.
If we add that the resulting ECSA is always garbage-free, then $\sem{\D}(\n)$ can be enumerated with output-linear~delay for every $\n$.
Further, we can note that the data structure is fully-persistent~\cite{DriscollSST86} in the sense that $\sem{\D}(\n)$ is immutable after each operation for every node $\n$ in $\D$. This last property will be important for the preprocessing of our algorithm. Summing up, we have proved the following result. 

\begin{theorem}\label{theo:data-structure}
	Assume that $\Reg$ is a fixed set. The operations $\dadd$, $\dextend$, and $\dunion$ take constant time and are fully persistent. Furthermore, if we start from an empty $\D$ and apply these operations over safe nodes, the result node $\n'$ is always a safe node and $\sem{\D}(\n)$ can be enumerated with output-linear delay for every node $\n$.
\end{theorem} 

\paragraph{The evaluation algorithm} In Algorithm~\ref{alg:evaluation}, we present an evaluation algorithm for $\textsc{EnumNSST}$ when the input NSST is garbage-free. By Theorem~\ref{prop:garbageFree}, we can convert any NSST into an equivalent garbage-free NSST and, therefore, Algorithm~\ref{alg:evaluation} works also for any NSST (after a preprocessing step). 
The main procedure is \textproc{Evaluation} that receives as input a garbage-free NSST $\nsst = (Q, \Sigma, \Omega, \Reg, \Delta, I, F)$ and a string $w =a_1 \ldots a_n$ and enumerates all the strings $\sem{\nsst}(w)$. Recall that $\sem{\nsst}(w)$ is a bag, and that these these may have repetitions during the enumeration~phase. 

\begin{algorithm}[t]
	\caption{The evaluation algorithm for garbage-free NSST. It receives as input a garbage-free NSST $\nsst = (Q, \Sigma, \Omega, \Reg, \Delta, I, F)$ and a string $w = a_1\ldots a_n$, and enumerates all the strings in $\sem{\nsst}(w)$ (with repetitions).}\label{alg:evaluation}
	\smallskip
	\begin{varwidth}[t]{0.6\textwidth}
		\begin{algorithmic}[1]
			\Procedure{Evaluation}{$\nsst, w$}
			\State $\D \gets \emptyset$
			\State $\Told \gets \textproc{InitialStates}(I)$
			\For{$i=1$ \textbf{to} $n$}
			\State $\Tnew \gets \emptyset$
			\ForEach{$p \in \dom(\Told)$}
			\ForEach{$(p,a_i, \sigma, q) \in \Delta$}
			\State $\Tnew \gets \textproc{UpdateLT}(\Tnew,q, \Told[p], \sigma)$
			\EndFor
			\EndFor
			\State $\Told \gets \Tnew$
			\EndFor
			\State $\textproc{Enumerate}(\Told, F)$
			\EndProcedure
			\State
			\Procedure{InitialStates}{$I$}
			\State $T \gets \emptyset$
			\ForEach{$p \in \dom(I)$}
			\State $T[p] \gets \D.\dadd(I(p))$
			\EndFor
			\State \Return $T$
			\EndProcedure
			\algstore{myalg}
		\end{algorithmic}
	\end{varwidth} 
	\begin{varwidth}[t]{0.45\textwidth}
		\begin{algorithmic}[1]
			\algrestore{myalg}
			\Procedure{UpdateLT}{$T, q, \n, \sigma$}
			\State $\n' \gets \D.\dextend(\n, \sigma)$
			\If{$T[q] = \bot$}
			\State $T[q] \gets \n'$
			\Else	
			\State $T[q] \gets \D.\dunion(T[q], \n')$
			\EndIf
			\State \Return $T$
			\EndProcedure
			\State
			\Procedure{Enumerate}{$T, F$}
			\ForEach{$p \in \dom(F) \cap \dom(T)$}
			\ForEach{$\nu \in \sem{\D}(T[p])$}
			\State \textbf{print} $\hat{\nu}(F(p))$
			\EndFor
			\EndFor
			\EndProcedure
		\end{algorithmic}
	\end{varwidth} 
\end{algorithm} 
Intuitively, the algorithm evaluates $\nsst$ over $w$ in one pass, by simulating all runs of $\nsst$ on-the-fly. It reads each symbol $a_i$ sequentially and keeps a set of \emph{active states}. A state $q \in Q$ is \emph{active} after reading letter $a_i$ if there exists a run of $\nsst$ over $a_1\ldots a_i$ that reaches $q$. Then for each active state $q$ we maintain a node $\n_q$ that represents all the valuations of runs that reach $q$. After reading $w$ completely, we can retrieve the outputs from the nodes whose corresponding active states are also final states. In the following, we explain this algorithm more in detail and discuss its correctness. 

In Algorithm~\ref{alg:evaluation}, we use two data structures: (1) an Enumerable Compact Set with Assignments (ECSA) $\D = (\Reg, \Omega, \N, \ell, r, \lambda)$ and (2)~lookup tables for storing nodes in $\N$.
For the ECSA $\D$, we use the methods $\dadd$, $\dextend$, and $\dunion$ explained in the previous subsections. By Theorem~\ref{theo:data-structure}, we assume that each such operations takes constant time. Additionally, we will write $\D \gets \emptyset$ to denote that we initialize $\D$ as empty. 
We will assume that $\D$ is globally available by all procedures and subprocedures of Algorithm~\ref{alg:evaluation}.

To remember the set of active states, we use \emph{lookup tables}, usually denoted by $T$, whose entries are states (i.e., active states) and each entry stores a node $\n$ from $\D$.
Formally, we can define a lookup table $T$ as a partial function $T: Q \pmap \N$. We write $\dom(T)$ to denote all the states that have an entry in $T$ and $\emptyset$ to denote the empty lookup table. If $q\in  \dom(T)$, we write $T[q]$ to retrieve the node at entry $q$. As usual, we write $T[q] = \bot$ when $q \notin \dom(T)$ (i.e., $T[q]$ is not defined). We use $T[q] \gets \n$ to declare that the $q$-entry of $T$ is updated with the node $\n$. Since we assume the RAM computational model, we can assume that each update or query (e.g., $q \in \dom(T)$) to the lookup table takes constant time. 

We give the preprocessing phase of Algorithm~\ref{alg:evaluation} in the $\textproc{Evaluation}$ procedure, specifically, from lines $2$ to $9$. This phase starts by initializing the ECSA $\D$ as empty (line 2) and setting the initial lookup table $\Told$ (line 3). We code the initialization of $\Told$ in the subprocedure $\textproc{InitialStates}$ (lines 12-16). The initialization receives as input the initial function $I$ of $\nsst$ and prepares $T$ with the set of initial states $p \in \dom(I)$ with their corresponding initial valuations $I(p)$. It finally returns $T$ which represents the set of initial active states. 

The main for-loop of the evaluation algorithm proceeds as follows (lines 4-9). It reads the string $w =a_1 \ldots a_n$ sequentially, from $i =1$ to $n$. 
For each $a_i$, it initializes a lookup table $\Tnew$ to empty (line 5), where $\Tnew$ represents the set of new active states after reading $a_i$.
Then for each active state $p \in \dom(\Told)$ and each transition $t$ of the form $t = (p, a_i, \sigma, q) \in \Delta$, it fires $t$ and updates $\Tnew$ accordingly (lines 6-8). To update the lookup table $\Tnew$ with $t$, we use the method \textproc{UpdateLT} (lines 17-23). The method receives a lookup table $T$, a state $q$, a node $\n$, and an assignment $\sigma$, and updates the entry $T[q]$ with $\n$ and $\sigma$. Specifically, if $T[q]$ is not defined, it inserts $\n' = \D.\dextend(\n,\sigma)$ in $T[q]$ (lines 19-20). Otherwise, it updates $T[q]$ with the union between the current node at $T[q]$ and $\n'$. Intuitively, in this step we are extending the outputs in $p$ represented by node $\n$ with the assignment $\sigma$ of the transition $t$ and store it in $\Tnew[q]$, accordingly.
Finally, after all transitions are fired, we swap $\Told$ with $\Tnew$ (line 9) and proceed to the next symbol. 

After we read the last symbol $a_n$, we continue with the enumeration phase, by calling the subprocedure \textproc{Enumerate} (line 10). This enumeration subprocedure (lines 25-28) receives as input the old table $\Told$, denoted by $T$, and the final function $F$.
For the enumeration, it considers each state $p$ that is both final (i.e., $p \in \dom(F)$) and active (i.e., $p\in \dom(T)$). We assume here that $\dom(F) \cap \dom(T)$ is computed once, before starting the iteration for each $p \in \dom(F) \cap \dom(T)$ (line 26). 
The enumeration proceeds by taking each valuation $\nu \in \sem{\D}(T[p])$ and printing $\hat{\nu}(F(p))$ (lines 27-28). By Proposition~\ref{prop:enumeration}, we know that we can enumerate all $\sem{\D}(T[p])$ with output-linear delay and without preprocessing. Then getting the next $\nu \in \sem{\D}(T[p])$ and printing the output string $u = \hat{\nu}(F(p))$ takes time proportional to $O(|u|)$, given that $\D$ is garbage-free (since $\nsst$ is garbage-free). In other words, we can conclude that the enumeration phase runs with output-linear~delay.

We end by discussing the correctness and running time of the evaluation algorithm. For the correctness, one can prove by induction on the length of $w$ that before reading each $a_i$, $\dom(\Told)$ is the set of active states and $\sem{\D}(\Told[p])$ are all valuations of  runs of $\nsst$ over $a_1 \ldots a_{i-1}$. Furthermore, $\D$ is always garbage-free given that $\nsst$ is garbage-free. Then, after reading the last letter $a_n$, $\Told$ contains all active states and its corresponding valuations, and the method \textproc{Enumerate} retrieves all strings in $\sem{\nsst}(w)$ correctly, and with output-linear delay. For the time complexity of the preprocessing phase, we can check that each step takes constant time (i.e., by Theorem~\ref{theo:data-structure}) and in each iteration we perform at most $|\nsst|$ number of steps. Overall, the evaluation algorithm takes at most  $O(|\nsst|\cdot |w|)$ number of steps to finish.

\section{Proofs of Section~\ref{sec:composition}}\label{sec:compAppendix}
	
\subsection{Proofs of Theorem~\ref{theo:nsstcomposition}}

\begin{proof}
Let $\nsst_1 = (Q_1, \Sigma, \Omega, \Reg_1, \Delta_1, I_1, F_1)$ and $\nsst_2 = (Q_2, \Omega, \Gamma, \Reg_2, \Delta_2, I_2, F_2)$ be two NSSTs such that $\Omega \subseteq \Gamma$. By Proposition~\ref{prop:garbageFree}, we can assume for the rest of the proof that $\nsst_1$ and $\nsst_2$ are garbage-free. 

Given sets $A$ and $B$, we define the set $\Asgs(A, B) \subseteq \Asg(A,B)$ of all copyless assignments that are copyless both in $A$ and $B$, simultaneously; namely, every element in $A \cup B$ appears at most once in the image of $\sigma$ for every $\sigma \in \Asgs(A,B)$. Note that $\Asgs(A, B)$ is a finite set. In the sequel, we will use $\zeta$ to denote an assignment in $\Asgs(A,B)$.

We define the NSST $\nsst$ that will work as the composition of $\nsst_1$ and $\nsst_2$ such that:
\[
\nsst \ = \ (Q, \Sigma, \Gamma, \Reg, \Delta, I, F)
\]
where $Q = Q_1 \times \{f\colon \Reg_1 \pmap (Q_2 \times \Asgs(\Reg_2, \Reg) \times Q_2) \}$. Intuitively, the states of $\nsst$ simulates a run of $\nsst_1$ in the first component. In the second component, $\nsst$ remember a \emph{subrun} $f(X) = (q_1, \sigma, q_2)$ over the content of register $X \in \Reg_1$ from $\nsst_1$, where $q_1 \in Q_2$ is the starting state, $q_2 \in Q_2$ is the ending state, and $\sigma$ is a \emph{summary} of the registers content of $\nsst_2$. To define the other components of $\nsst$ (i.e.,  $\Delta$, $I$, and $F$) we need several technical definitions that we introduce next.

\paragraph{Summaries} We start by formalizing the notion of summary, introduced in~\cite{AlurC10}. Given a word $w \in (\Reg_2 \cup \Gamma \cup \Reg)^*$ we define the \emph{summary of $w$ by $\Reg$} as the pair $(\bar{w}, \zeta)$ such that, if $w$ is of the form:
\[
w = u_1 X_1 u_2 X_2 \ldots X_k u_{k+1}
\]
for some $k \geq 0$ such that $X_i \in \Reg_2$ and $u_i \in (\Gamma \cup \Reg)^*$, then $\bar{w} = Y_1 X_1 Y_2 X_2 \ldots X_k Y_{k+1}$ for some pairwise distinct registers $Y_1, \ldots, Y_{k+1} \in \Reg$ and $\zeta \in \Asg(\{Y_1, \ldots, Y_{k+1}\}, \Gamma \cup \Reg)$ satisfies that $\zeta(Y_i) = u_i$ for every $i \leq k+1$. The idea of a summary $(\bar{w}, \zeta)$ is to replace each maximal subword $u_1, \ldots, u_{k+1} \in (\Gamma \cup \Reg)^*$ from $w$ with distinct registers $Y_1, \ldots, Y_{k+1}$, represented by $\bar{w}$, and store them into the registers $Y_1, \ldots, Y_{k+1}$, represented by $\zeta$. In particular, one can check that $\zeta(\bar{w}) = w$ and that $\zeta$ is always copyless.

We can generalize the notion of summary from words to assignments as follows. Let $\sigma \in \Asg(\Reg_2, \Gamma \cup \Reg)$ be an assignment and assume that $\dom(\sigma) = \{X_1,\ldots, X_\ell\}$. 
The summary of $\sigma$ by $\Reg$ is a pair $(\bar{\sigma}, \zeta)$ such that, if $(\bar{w}_1, \zeta_1), \ldots, (\bar{w}_\ell, \zeta_\ell)$ are summaries of $\sigma(X_1), \ldots, \sigma(X_\ell)$, respectively, with $\dom(\zeta_i) \cap \dom(\zeta_{i'}) = \emptyset$ for every $i \neq i'$, then $\bar{\sigma}(X_j) = \bar{w}_j$ for every $j \leq \ell$ and $\zeta = \zeta_1 \cup \ldots \cup \zeta_{\ell}$. Intuitively, we are summarizing every $\sigma(X_i)$ together into a single assignment $\bar{\sigma}$ and $\zeta$ is the disjoint union of $\zeta_i$.
Finally, the summary of a tuple $(\sigma_1, \ldots, \sigma_n)$ with each $\sigma_i \in \Asg(\Reg_2, \Gamma \cup \Reg)$ is a pair $((\bar{\sigma}_1, \ldots, \bar{\sigma}_n), \zeta_1 \cup \ldots \cup \zeta_n)$ such that, for every $i \leq n$, $(\bar{\sigma}_i, \zeta_i)$ is a summary of $\sigma_i$ and $\dom(\zeta_{i}) \cap \dom(\zeta_{j}) = \emptyset$ for every $j \neq i$.  
Notice that $\zeta$ is an assignment in $\Asg(\Reg, \Gamma)$ and it is always copyless, since we always pick disjoint set of registers for each~component.

\paragraph{Semigroup of subruns} The following semigroup will be useful for composing subruns and checking if the composition is valid. Consider the semigroup with domain $\mathcal{S} = (Q_2 \times \Asg(\Reg_2, \Gamma \cup \Reg) \times Q_2) \cup \{\bot\}$ where $\bot$ is a new fresh element. Then we define the operation $\circledast$ of the semigroup such that:
\[
(q_1, \sigma, q_2) \circledast (q_3, \sigma', q_4) \ = \ \left\{ 
\begin{array}{ll}
	(q_1, \sigma \circ \sigma', q_4) & \text{if } q_2 = q_3 \\
	\bot & \text{otherwise}
\end{array}
\right.
\]
Further, $s_1 \circledast s_2 = \bot$ whenever $s_1 = \bot$ or $s_2 = \bot$.
One can easily check that $(\mathcal{S}, \circledast)$ forms a semigroup, given that $\circledast$ is well-defined in $\mathcal{S}$ (i.e., the result is always an element in $\mathcal{S}$), and is associative. 
Below, we will use the semigroup $(\mathcal{S}, \circledast)$ to easily apply and check that subruns match. Indeed, one can check that $s_1 \circledast s_2 = s_3$ and $s_3 \neq \bot$, then $s_3$ is the subrun representing $s_1$ followed by $s_2$. Instead, if $s_1 \circledast s_2 = \bot$ then either $s_1$ or $s_2$ are not valid subruns, or $s_1$ followed by $s_2$ does not form a subrun. 

For using the semigroup $(\mathcal{S}, \circledast)$, the notion of homomorphism between semigroups will be also relevant. Let $A$ be any finite alphabet. An \emph{homomorphism} from $(A^*, \cdot)$ to $(\mathcal{S}, \circledast)$ is a function $H: A^* \to \mathcal{S}$ such that $H(w_1 \cdot w_2) = H(w_1) \circledast H(w_2)$. In the following, we will usually define homomorphisms as a function $H: A \to \mathcal{S}$, that extends naturally to a homomorphism $H: A^* \to \mathcal{S}^*$ such that $H(a_1\ldots a_n) = H(a_1) \circledast \ldots \circledast H(a_n)$.

\paragraph{Assignment indexing} The last technical ingredient is a notation for indexing all letters from an assignment in $\Asg(\Reg_1, \Omega)$. 
Formally, for $c \in \Reg_1 \cup \Omega$, let $\indx(c, X, i)$ be a function such that $\indx(c, X, i) = (c,X,i)$ if $c \in \Omega$ and $\indx(c, X, i) = c$, otherwise. 
Then we define $\indxsigma$ as an assignment in $\Asg(\Reg_1, \Omega \times \Reg_1 \times \bbN)$ such that $\dom(\indxsigma) = \dom(\sigma)$ and, if $\sigma(X) = c_1 \ldots c_n$,~then:
\[
\indxsigma(X) = \indx(c_1, X, 1) \ldots \indx(c_n, X, n).
\]
For the special case that $\sigma(X) = \epsilon$, we define $\indxsigma(X) = (\epsilon, X, 1)$.
In other words, $\indxsigma$ works as $\sigma$ where each letter in $\sigma(X)$ is indexed by $X$ and its position. For example, below we show an assignment $\sigma$ and its corresponding $\indxsigma$.
\[
\begin{array}{ll}
\begin{array}{rl}
	\sigma: \!\! & X := baXab \\
	& Y := X b \\
	& Z := \epsilon 
\end{array} & \ \ \ \ \ 
\begin{array}{rl}
	\indxsigma: \!\! & X := (b,X,1) (a,X,2) Y (a,X,4) (b,X,5) \\
	& Y := X (b,Y,2) \\
	& Z := (\epsilon, Z, 1)
\end{array} 
\end{array}
\]
Finally, define $\indx[\sigma]$ to be all letters of the form $(c,X,i)$ that appear in $\indxsigma$. We call $(c,X,i)$ an \emph{index} of $\sigma$. 

It will be convenient to extend the notation $\indxsigma$ to words. Specifically, for a word $w = c_1 \ldots c_n \in (R_1 \cup \Omega)^*$ we define $\indxw = \indx(c_1, 1) \ldots \indx(c_n, n)$ where $\indx(c,i) = (c,i)$ if $c\in \Omega$, and $\indx(c,i) = c$, otherwise. If $w = \epsilon$, we define $\indxw = (\epsilon, 1)$. For the word case, note that we omit the $X$ component in the index, given that it is unnecessary. 

The purpose of $\indxsigma$ is to identify each letter $a$ (or $\epsilon$) in $\sigma$ and to assign a special code $(a,X,i)$. The set $\indx[\sigma]$ will allow us to collect all these codes for deciding the transitions that we will use when the NSST $\nsst_1$ takes a transition. This will be more clear in the next subsection, when we finally define the transition relation of $\nsst$.

\paragraph{Definition of $\Delta$} We have all the ingredients to define the components $\Delta$, $I$, and $F$ of the NSST $\nsst$. For defining $\Delta$, we repeat the following steps, until all possible decisions are covered.
\begin{enumerate}
	\item Take a state $(p, f) \in Q$ (recall that $f: \Reg_1 \pmap Q_2 \times \Asgs(\Reg_2, \Reg) \times Q_2$).
	\item Take a transition $(p, a, \sigma, q) \in \Delta_1$. 
	\item Let $\indx[\sigma] = \{(a_1, X_1, i_1), \ldots, (a_\ell, X_\ell, i_\ell)\}$ be all indexes of $\sigma$.
	\item Choose transitions $(q_1, a_1, \sigma_1, q_1'), \ldots, (q_\ell, a_\ell, \sigma_\ell, q_\ell') \in \Delta_2$.
	For the special case that $a_i = \epsilon$ for some $i \leq \ell$, choose a pair of the form $(q_i, \epsilon, \sigma_{\operatorname{id}}, q_i)$ where $q_i \in Q_2$ and $\sigma_{\operatorname{id}}$ is the identity function over $\Reg_2$.
	\item Define the homomorphism $H: \Reg_1 \cup \indx[\sigma] \to \mathcal{S}$ from words to the semigroup $(\mathcal{S}, \circledast)$ such that
	 $H(X) = f(X)$ for every $X \in \Reg_1$ and $H((a_j, X_j, i_j)) = (q_j, \sigma_j, q_j')$ for every $j \leq \ell$. 
	\item If $H(\indxsigma(X)) \neq \bot$ for every $X \in \dom(\indxsigma)$, then continue.
	\item Let $\dom(\sigma) = \{Y_1, \ldots, Y_m\}$ and let $H(\indxsigma(Y_i)) = (p_i, \pi_i, p_i')$ for every $i \leq m$. 
	\item Construct the summary $((\bar{\pi}_1, \ldots, \bar{\pi}_m), \zeta)$ of $(\pi_1, \ldots, \pi_m)$.
	\item Define the partial function $f'$ such that $\dom(f') =\{Y_1, \ldots, Y_m\}$ and $f'(Y_i) = (p_i, \bar{\pi}_i, p_i')$.
	\item Add the transition $((p,f), a, \zeta, (q, f'))$ to $\Delta$.  
\end{enumerate}
For understanding the previous construction, recall that $(p, f) \in Q$ represents that $\nsst_1$ is in state $p$ and, for every register $X$ of $\nsst_1$,  $\nsst_2$ has a subrun $f(X)$ over the content of $X$ (step (1)). Then, if we move $\nsst_1$ by using $(p,a,\sigma, q) \in \Delta_1$ (step (2)), then we need to update the moves of $\nsst_2$, namely, update $f$ with $\sigma$ and~$\Delta_2$. The moves of $\nsst_2$ for registers in $\sigma$ are encoded in $f$, but we still have to decide how to move with letters in~$\sigma$. For this reason, we choose one transition $(q_j, a_j, \sigma_j, q_j')$ for each letter $a_j \in \Omega$ appearing in $\sigma$, represented by the index $(a_j, X_j, i_j)$ (step (3) and (4)). Thus, for simulating $\nsst_2$ over $\sigma$, we define the homomorphism $H$, which for every register content $X$ uses $f(X)$ and for every letter indexed as $(a_j, X_j, i_j)$ uses $(q_j, \sigma_j, q_j')$ (step (5)). To check that these guesses of the subruns and the transitions are correct, we check that the homomorphism $H$ applied to each content $\indxsigma(X)$ is different from $\bot$ (step (6)). If this is the case, it means that the subruns and guessed transitions coincide, and we proceed to construct the new transition in $\Delta$ (steps (7) to (10)). Towards this goal, we summarize the subrun $H(\indxsigma(Y_i)) = (p_i, \pi_i, p_i')$ and store it in $f'$ as $f'(Y_i) = (p_i, \bar{\pi}_i, p_i')$ (steps (7) to (9)). The summaries produce the copyless assignment $\zeta$ that stores the partial valuations of the summaries in $\Reg$-registers. Finally, we define the transition $((p,f), a, \zeta, (q, f')) \in \Delta$, which encodes the whole process above.

It is important to note that $\zeta$ in the transition $((p,f), a, \zeta, (q, f')) \in \Delta$ is a copyless assignment. Indeed, given that $H$ works as a homomorphism and there is a single appearance of each register $X$ in $\sigma$, then each register of $\Reg$ appears at most once over all assignments $\pi_1, \ldots, \pi_m$. Therefore, the construction here always produces copyless assignments.

\paragraph{Definition of $I$} For defining the initial partial function $I$, we proceed similar to $\Delta$:
\begin{enumerate}
	\item Take a state $p \in Q_1$ such that $I_1(p) = \nu$.
	\item Let $\indx[\nu] = \{(a_1, X_1, i_1), \ldots (a_\ell, X_\ell, i_\ell)\}$ be the indexes of $\nu$.
	\item Choose transitions $(q_1, a_1, \sigma_1, q_1'), \ldots, (q_\ell, a_\ell, \sigma_\ell, q_\ell') \in \Delta_2$.
	For the special case that $a_i = \epsilon$ for some $i \leq \ell$, choose a pair of the form $(q_i, \epsilon, \sigma_{\operatorname{id}}, q_i)$ where $q_i \in Q_2$ and $\sigma_{\operatorname{id}}$ is the identity function over $\Reg_2$.
	\item Define the homomorphism $H: \indx[\nu] \to \mathcal{S}$ from words to the semigroup $(\mathcal{S}, \circledast)$ such that $H((a_j, X_j, i_j)) = (q_j, \sigma_j, q_j')$ for every~$j \leq \ell$.  
	\item If $H(\indxnu(X)) \neq \bot$ for every $X \in \dom(\indxnu)$, then continue.
	\item Let $\dom(\nu) = \{Y_1, \ldots, Y_m\}$ and let $H(\indxnu(Y_i)) = (p_i, \pi_i, p_i')$ for every $i \leq m$. 
	\item Construct the summary $((\bar{\pi}_1, \ldots, \bar{\pi}_m), \zeta)$ of $(\pi_1, \ldots, \pi_m)$.
	\item Define the partial function $f$ such that $\dom(f) =\{Y_1, \ldots, Y_m\}$ and $f(Y_i) = (p_i, \bar{\pi}_i, p_i')$.
	\item Define $I((p, f)) = \zeta$. 
\end{enumerate}
The definition of $I$ follows the same approach as for the transition relation $\Delta$. Intuitively, if we start with a state $p$ such that $I_1(p)$ is defined and equal to $\nu$, then we run $\nsst_2$ over the content in $\nu$. Given that $\nu$ is an assignment, the definition follows the same ideas as for $\Delta$, by switching $\sigma$~with~$\nu$.  

\paragraph{Definition of $F$} The construction of the final function $F$ is slightly different to $\Delta$ and $I$, although we present it in a similar fashion. 
\begin{enumerate}
	\item Take a state $(p, f) \in Q$ such that $F_1(p)$ is defined. 
	\item Let $\indx[F_1(p)] = \{(a_1, i_1), \ldots (a_\ell, i_\ell)\}$ be the indexes of $F_1(p)$.
	\item Choose transitions $(q_1, a_1, \sigma_1, q_1'), \ldots, (q_\ell, a_\ell, \sigma_\ell, q_\ell') \in \Delta_2$ (not necessarily different).
	For the special case that $F_1(p) = \epsilon$ and $\indx[F_1(p)] = \{(\epsilon, 1)\}$, choose one tuple of the form $(r, \epsilon, \sigma_{\operatorname{id}}, r)$ for some $r \in Q_2$ and $\sigma_{\operatorname{id}}$ is the identity function over $\Reg_2$.
	\item Define the homomorphism $H: \Reg_1 \cup \indx[F_1(p)] \to \mathcal{S}$ from words to the semigroup $(\mathcal{S}, \circledast)$ such that
	$H(X) = f(X)$ for every $X \in \Reg_1$ and $H((a_j, i_j)) = (q_j, \sigma_j, q_j')$ for every $j \leq \ell$. 
	\item If $H(F_1(p)) = (p_0, \pi, p_f)$ and both $I_2(p_0)$ and $F_2(p_f)$ are defined, then define: 
	$$
	F((p,f)) = [I_2(p_0)]\big(\pi(F_2(p_f))\big).
	$$
\end{enumerate}
Here, the goal is to reconstruct the final output of $\nsst_2$, given that $\nsst_1$ reaches a final state where $F_1(p)$ is defined (step (1)). Then, similar to $\Delta$, we choose the transitions for the letters in $F_1(p)$ (steps (2) and (3)), and construct the homomorphism $H$ to simulate $\nsst_2$ over $F_1(p)$ (step (4)). If $H(F_1(p)) = (p_0, \pi, p_f)$, this means that the subruns and transitions of $\nsst_2$ match. However, we have to still check that $p_0$ is an initial state and $p_f$ is a final state. If this is the case, then we can reconstruct the output of $\nsst_2$ with the expression $[I_2(p_0)]\big(\pi(F_2(p_f))\big)$ (step (5)).

\paragraph{Correctness} This concludes the construction of $\nsst$. To show that $\sem{\nsst}(\doc) = \sem{\nsst_1 \circ \nsst_2}(\doc)$ for every document $\doc$, the proofs goes by making a bijection between accepting runs $\rho$ of $\nsst$ over $\doc$ and pair of runs $(\rho_1, \rho_2)$ such that $\rho_1$ is an accepting run of $\nsst_1$ over $\doc$ and $\rho_2$ is an accepting run of $\nsst_2$ over the output of $\rho_1$.
For defining this bijection $F$, we can follow the construction above and, for every run $\rho$ of $\nsst$ over $\D$, construct a unique pair $(\rho_1, \rho_2)$ by induction. For showing that $F$ is injective, the garbage-free assumption over $\nsst_1$ and $\nsst_2$ helps, given that all the content tracked by $\nsst$ is useful, and then always produce a unique pair $(\rho_1, \rho_2)$. For proving that $F$ is surjective, from a pair of runs $(\rho_1, \rho_2)$ we can construct the run $\rho$ over $\doc$ by following the construction above. 
Given that there is a bijection between $\sem{\nsst}(\doc)$ and $\sem{\nsst_1 \circ \nsst_2}(\doc)$ for every $\doc$, we conclude that~$\sem{\nsst} = \sem{\nsst_1\circ \nsst_2}$. 

\paragraph{Size of $\nsst$} To bound the size of the NSST $\nsst$, let $|Q_1| = k$, $|\Reg_1|=\ell$, $|Q_2| = m$, and $|\Reg_2| = n$. Recall that $Q = Q_1 \times \mathcal{F}$ where: 
\[
\mathcal{F} \ = \ \{f: \Reg_1 \pmap Q_2 \times \Asgs(\Reg_2, \Reg) \times Q_2\}.
\]

First, we bound the number of register in $\Reg$ that we need to do the summaries (i.e., $|\Reg|$). For this, we need to bound the number of states to do a summary of one $f \in \mathcal{F}$. For one assignment $\sigma \in \Asgs(\Reg_2, \Reg)$, we will need at most $(n+1)^2$ registers. Given that we have at most $\ell$ of these assignments in $f$, then we need at most $\ell \cdot (n+1)^2$ registers in $\Reg$, namely, $|\Reg| \leq \ell \cdot (n+1)^2$. To simplify the following calculations, let $N = \ell \cdot (n+1)^2$. 

Next, we bound the size of $\Asgs(\Reg_2, \Reg)$. For any $\sigma \in \Asgs(\Reg_2, \Reg)$ and any $X \in \Reg_2$, we can put at most $(n+N+1)!$ different strings at the entry $\sigma(X)$. Therefore, $|\Asgs(\Reg_2, \Reg)| \leq ((n+N+1)!)^n$.

To finish, we bound the size of $Q_1 \times \mathcal{F}$. Given that $\mathcal{F}$ are partial functions we have that the number of possible elements in the image is:
\[
|Q_2 \times \Asgs(\Reg_2, \Reg) \times Q_2| \leq m^2 \cdot ((n+N+1)!)^n
\]
Then $|\mathcal{F}| \leq (m^2 \cdot ((n+N+1)!)^n)^\ell$. We conclude that:
\[
\begin{array}{rcl}
	|Q| & \leq & |Q_1| \cdot |\mathcal{F}| \\
	& \leq & k \cdot (m^2 \cdot ((n+N+1)!)^n)^\ell \\
	& \leq & k \cdot m^{2\ell} \cdot ((n+N+1)!)^{n\ell} \\
	& \leq & k \cdot m^{2\ell} \cdot ((n+\ell \cdot (n+1)^2+1)!)^{n\ell} \\
	& \leq & 2^{\log(k)+2\ell\cdot \log(m)+n \ell \cdot (\ell \cdot (n+1)^2) \cdot \log(\ell \cdot (n+1)^2)}
\end{array}
\]
where for the last equation we are using that $x! \leq x^x$ for $x$ big enough. Then the number of states of $\nsst$ is in $O\big(2^{g(|\nsst_1|, |\nsst_2|)}\big)$ for some polynomial $g(x, y)$. Note that one can similarly bound the number of transitions in $O\big(2^{g(|\nsst_1|, |\nsst_2|)}\big)$.
\end{proof} 	

\end{document}